\pdfoutput=1
\documentclass[12pt]{article}
\usepackage[hmargin=1in,top=1in,bottom=1.25in]{geometry}

\usepackage[font=footnotesize]{caption}
\usepackage[normalem]{ulem}
\usepackage{relsize}
\usepackage{placeins}

\usepackage{ifthen}
\usepackage{style/qstyle} 

\newboolean{anonymous}

\AddToHook{env/theorem/begin}{\crefalias{definition}{theorem}}
\AddToHook{env/lemma/begin}{\crefalias{definition}{lemma}}
\AddToHook{env/corollary/begin}{\crefalias{definition}{corollary}}
\AddToHook{env/proposition/begin}{\crefalias{definition}{proposition}}
\AddToHook{env/conjecture/begin}{\crefalias{definition}{conjecture}}
\AddToHook{env/example/begin}{\crefalias{definition}{example}}

\crefname{theorem}{Theorem}{Theorems}
\crefname{lemma}{Lemma}{Lemmata}
\crefname{corollary}{Corollary}{Corollaries}
\crefname{proposition}{Proposition}{Propositions}
\crefname{conjecture}{Conjecture}{Conjectures}
\crefname{example}{Example}{Examples}
\crefname{definition}{Definition}{Definitions}

\makeatletter
\AddToHook{cmd/appendix/before}{\def\cref@section@alias{appendix}\def\cref@subsection@alias{appendix}}
\makeatother

\renewcommand{\emph}[1]{\textbf{#1}}
\renewcommand{\mathbf}[1]{\textit{\textbf{#1}}}
\newcommand{\mathbbf}[1]{\mathbbmss{#1}}
\newcommand{\mbf}[1]{\mathbf{#1}}
\newcommand{\mbbf}[1]{\mathbbf{#1}}
\newcommand{\hy}{\text{-}}
\newcommand{\LR}[1]{L_{#1}}

\newcommand{\heading}[1]{\medskip\noindent\textbf{#1}}

\title{Quantum polymorphism  characterisation of commutativity gadgets~in~all~quantum~models}

\ifthenelse{\boolean{anonymous}}{
\author{}
}{
\author[1]{Eric Culf}
\author[2]{Josse van Dobben de Bruyn}
\author[3]{Peter Zeman}

\affil[1]{\small{Department of Applied Mathematics and Institute for Quantum Computing, University of Waterloo, Canada, \texttt{eculf@uwaterloo.ca}}}
\affil[2]{Department of Applied Mathematics, Faculty of Mathematics and Physics, Charles University, Czech~Republic, \texttt{josse.van-dobben-de-bruyn@matfyz.cuni.cz}}
\affil[3]{Department of Algebra, Faculty of Mathematics and Physics, Charles University, Czech Republic, \texttt{peter.zeman@matfyz.cuni.cz}}
\date{\vspace{-1cm}}
}

\begin{document}

\maketitle

\begin{abstract}
	Commutativity gadgets provide a technique for lifting classical reductions between constraint satisfaction problems to quantum-sound reductions between the corresponding nonlocal games. We develop a general framework for commutativity gadgets in the setting of quantum homomorphisms between finite relational structures. Building on the notion of quantum homomorphism spaces, we introduce a uniform notion of commutativity gadget capturing the classical, finite-dimensional quantum, quantum approximate, and commuting-operator models. In the robust setting, we use the weighted-algebra formalism for approximate quantum homomorphisms to capture corresponding notions of robust commutativity gadgets.
	We introduce stable commutativity, an approximate commutation property for weighted algebras, in order to study commutativity gadgets in the robust setting.
	
	Our main results characterize both non-robust and robust commutativity gadgets purely in terms of quantum polymorphism spaces: in any model, existence of a commutativity gadget is equivalent to the collapse of the corresponding quantum polymorphisms to classical ones at arity~$|A|^2$, and robust gadgets are characterized by stable commutativity of the appropriate weighted polymorphism algebras at the same bounded arity. We extend this to the oracular framework as well. We use this characterisation to show relations between the classes of commutativity gadget, notably that existence of a robust commutativity gadget is equivalent to the existence of a corresponding non-robust one.
	
	Finally, we prove that quantum polymorphisms of complete graphs $K_n$ have a very special structure, wherein the noncommutative behaviour only comes from the quantum permutation group $S_n^+$. Combining this with techniques from combinatorial group theory, we construct separations between commutativity-gadget classes: we exhibit a relational structure admitting a finite-dimensional commutativity gadget but no quantum approximate gadget, and, conditional on the existence of a non-hyperlinear group, a structure admitting a quantum approximate commutativity gadget but no commuting-operator gadget.
\end{abstract}

\clearpage

\setcounter{tocdepth}{1}
\tableofcontents

\section{Introduction}

\emph{Constraint satisfaction problems (CSPs)} capture some of the most fundamental computational problems, including linear equations, graph colourings, and variants and generalisations of satisfiability. One way to phrase such a problem is as the question of the existence of a \emph{homomorphism} from a given domain relational structure $X$ to a fixed codomain structure $A$. The study of CSPs has a decades-long history in computational complexity, culminating in the celebrated CSP dichotomy theorem~\cite{Bul17,Zhu17}, which establishes that every finite-template CSP is either in $\tsf{P}$ or $\tsf{NP}$-complete. A central tool underlying much of this classification is the \emph{algebraic approach}, which relates the complexity of $\CSP(A)$ to the structure of the \emph{polymorphisms} of $A$---that is, the homomorphisms $A^k \to A$, where $k$ is any natural number~\cite{BKW17}.

Almost in parallel to the developments in complexity theory, CSPs appeared in the foundations of quantum theory.
Bell~\cite{Bel64} demonstrated that quantum systems can exhibit correlations that are unattainable in classical physics, ruling out local hidden-variable theories as a possible explanation of quantum phenomena.
CSPs provide a possible formulation of Bell's theorem: 
there are specific CSP templates for which the
classical value and the quantum value can differ, where the latter is defined as the optimal
success probability in a two-player interactive proof of the satisfiability of the CSP instance
where the players have access to quantum resources.
Possibly the simplest example is the Mermin-Peres magic square~\cite{Mer90,Per90a}, which is a system of six linear equations over nine variables.
This system has no solution, so the optimal success probability of classical players is strictly less than one, however, quantum players can achieve success probability one. Since the quantum and classical notions of satisfiability differ, this naturally leads to the question of the computational complexity of determining the quantum value of a CSP.

In recent years, there has been growing interest in the quantum analogues of constraint satisfaction. The natural framework for this is the theory of \emph{nonlocal games}: a CSP instance can be reformulated as a game between two non-communicating players who may share entanglement, and the central question becomes whether the players can satisfy all constraints perfectly or approximately using quantum correlations of various types---finite-dimensional or quantum~($q$), Connes-embeddable or quantum approximate ($qa$), or commuting-operator ($qc$). The resulting hierarchy of quantum CSP problems exhibits strikingly different behaviour from the classical setting and deep connections to operator algebras and quantum information emerge~\cite{JNV+21,MNY22}.

\heading{Nonlocal game complexity.} The complexity of deciding the value of a nonlocal game under quantum correlations can vary dramatically from the classical case, giving rise to a wide variety of classes of problems. In a landmark work, Ji, Natarajan, Vidick, Wright, and Yuen showed that it is undecidable to decide the quantum value of a general nonlocal game~\cite{JNV+21}. In more detail, they show that there is a polynomial-time reduction from the halting problem to the following problem: given a program to efficiently sample questions and decide answers for a nonlocal game, decide if the quantum value is $1$ or upper bounded by $1/2$, given the promise that one of these two cases hold. As such, this problem is complete for $\tsf{RE}$.

Since the set of quantum correlations is not closed~\cite{Slo19}, there is a subtle distinction between quantum correlations and quantum-approximate correlations, which can be realised as the limit of a sequence of quantum correlations. Nevertheless, the problem of deciding if there is a quantum correlation with value $1$ is as hard as the problem of deciding the quantum value~\cite{JNV+21}.

Unlike for the problem of deciding the classical value, where the gapped and gapless cases have equal complexity due to the PCP theorem~\cite{ALMSS98}, the gapless version of the quantum problem --- deciding if the quantum value of a nonlocal game is $1$ or strictly less --- is strictly harder. In fact, it is a $\Pi_2^0$-complete problem~\cite{MNY22}.

Generalising to commuting-operator correlations, which are the natural model of correlation in quantum field theory~\cite{CQK23}, gives rise to yet another inequivalent complexity class. Here, both the gapped and gapless problems of deciding if the commuting-operator value of a nonlocal game is~$1$ are complete for $\tsf{coRE}$~\cite{Lin25}.

\heading{CSP nonlocal games.} CSPs can be presented as nonlocal games in two different ways, which lead to non-equivalent decision problems.

When the relations in the CSP have arity $2$, such as for graph homomorphism problems, there is a natural nonlocal game --- the \emph{assignment game} or \emph{variable-variable-game}. In this game, the referee sends Alice variable $x$ and Bob a variable $y$, both in the domain relational structure $X$. The players must respond with answers $a$ and $b$, respectively, such that $a=b$ if $x=y$, and if $(x,y)$ satisfies a relation in $X$, then $(a,b)$ must satisfy the corresponding relation in the codomain structure $A$. 
This game dates back to early work on pseudotelepathy~\cite{BCWdW01,GW02,AHKS06} and on the quantum chromatic number of graphs~\cite{CMN07}, which has later been generalized to quantum homomorphisms between graphs \cite{MR16}.
These notions play an important role in nonlocal games \cite{Har24} and quantum graph theory \cite{AMRSSV19,LMR20,MR20}, and are also related to entanglement-assisted one-shot zero-error capacities of classical channels \cite[\S{}5]{MR16}.
Perfect quantum strategies for the assignment game admit a simple algebraic characterization in terms of a quantum satisfying assignment for the CSP \cite{MR16}.
Although no generalization of this nonlocal game is known for CSPs of higher arity, the induced algebraic characterization can be generalized in a straightforward way, enabling us to define quantum satisfying assignments for CSPs of arbitrary arity.
We refer to this definition of quantum assignments as the \emph{non-oracular framework}.

A different nonlocal game has been studied for CSPs of higher arity --- the \emph{constraint-variable game}. In this game Alice receives a constraint and Bob receives a variable from the domain structure~$X$. Alice must answer with a satisfying assignment for the variables occurring in the constraint and Bob must answer with an assignment to the variable, in such a way that their answers are consistent.
An alternate but closely-related nonlocal presentation of a CSP is the \emph{constraint-constraint game}, where both players receive constraints, and need to reply with satisfying and consistent assignments.
The constraint-variable and constraint-constraint games are well-defined for every CSP. Their perfect quantum strategies agree and have a simple algebraic characterisation.
However, as compared to the non-oracular framework, this framework requires the observables for variables occurring in the same constraint to commute, which is not always a natural restriction.
Constraint-variable games were first introduced in~\cite{CM14,Arkh12} for boolean constraint systems, and extended to general CSPs in~\cite{ABdSZ17}; constraint-constraint games were first introduced in~\cite{Ara04} and further studied in~\cite{KPS18,PS23,MS24}.
We refer to this definition of quantum assignment as the \emph{oracular framework}. This name is in reference to the fact that, for arity-$2$ CSPs, the games studied in this framework are \emph{oracularisations} --- in the sense of~\cite{JNV+21} --- of the games studied in the non-oracular framework.

Which of the two frameworks (oracular or non-oracular) makes more sense in a given situation depends on the problem at hand.
In the literature, graph colouring and graph homomorphism are commonly studied in the non-oracular framework \cite{MR16}, whereas linear equations are commonly studied in the oracular framework \cite{CLS17,Slo19,Slo20}.

\heading{Commutativity gadgets.} A key technique for establishing hardness results in the quantum CSP setting is the use of \emph{commutativity gadgets}. Informally, a commutativity gadget for a relational structure $A$ is an auxiliary structure $G$ with two distinguished variables $x$ and $y$ such that any quantum homomorphism from $G$ to $A$ forces the measurement operators associated to $x$ and $y$ to commute. This allows one to ``lift'' classical NP-hardness reductions to the quantum setting: the gadget ensures that entanglement provides no advantage in composing strategies, enabling the simulation of classical computation within quantum proof systems. Classical reductions between CSP languages, such as pp-definition and pp-construction, consist of rewriting constraints in the original CSP in terms of constraints of the new CSP. Adding commutativity gadgets to these local rewritings allows the reduction to remain sound, by only permitting classical assignment locally. The notion of a commutativity gadget was introduced in~\cite{Ji13}. Commutativity gadgets and their robust variants were systematically developed in~\cite{CDdBVZ25}, building on earlier work of~\cite{JNV+21,MS24,CM24}, and were used to establish undecidability and hardness results for a range of quantum CSP problems.

A natural question, then, is: \emph{when does a relational structure admit a commutativity gadget?} Recent work of Ciardo, Joubert, and Mottet~\cite{CJM25} answered this question for finite-dimensional quantum strategies by establishing an equivalence: a structure $A$ admits a $q$-commutativity gadget if and only if every quantum polymorphism of $A$ is classical. This bridges the algebraic approach to CSPs with the theory of quantum homomorphisms and nonlocal games.

\heading{Our contributions.} In this paper, we develop a comprehensive framework for commutativity gadgets across all standard models of quantum correlation, and characterise their existence purely in terms of quantum polymorphisms.
Furthermore, we compare the different models with the aim of proving that existence of a commutativity gadget in one model does or does not imply existence in another.
Our contributions are as follows.

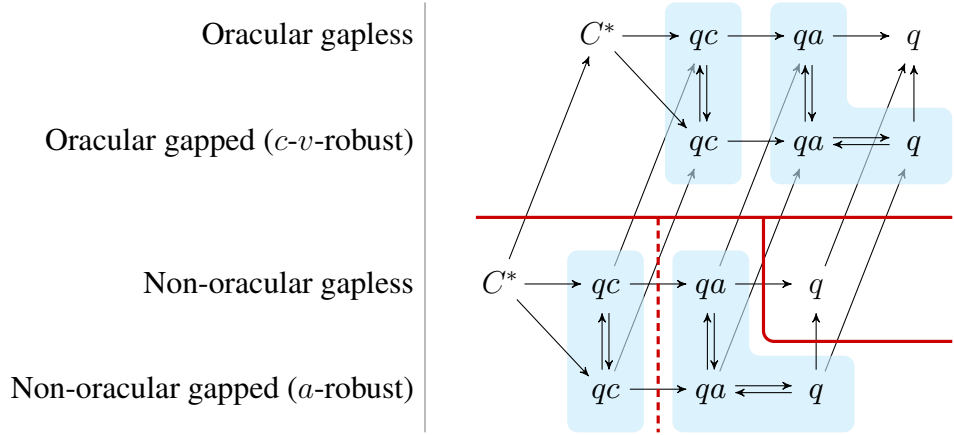
\begin{figure}[t]
    \centering
\def\arrowbaseshift{.6ex}
\def\doublearrowshift{1.3pt}
\def\boxwidth{8pt}
\def\xscale{1.4}
\def\yscale{1.4}
\begin{tikzpicture}[>=stealth',
	            label_node/.style={circle,anchor=base,inner sep=0pt,minimum width=1.5em,minimum height=1.5em},
	            box_fill/.style={rounded corners, fill=cyan!20, fill opacity=.6},
	            nonOracularToOracularSeparation/.style={red!80!black,very thick,rounded corners},
	            CstarToRobustQc/.style={},
	            GappedToGapless/.style={},
	            nonOracularToOracular/.style={}]
	
	\begin{scope}[xshift=1.3cm,xscale=\xscale,yscale=\yscale]
		\draw (0,0) coordinate (oc_coord);
		\draw (1,0) coordinate (oqc_coord);
		\draw (2,0) coordinate (oqa_coord);
		\draw (3,0) coordinate (oq_coord);
		\draw (1,-1) coordinate (roqc_coord);
		\draw (2,-1) coordinate (roqa_coord);
		\draw (3,-1) coordinate (roq_coord);
	\end{scope}
	
	\begin{scope}[yshift=-3.3cm,xscale=\xscale,yscale=\yscale]
		\draw (0,0) coordinate (c_coord);
		\draw (1,0) coordinate (qc_coord);
		\draw (2,0) coordinate (qa_coord);
		\draw (3,0) coordinate (q_coord);
		\draw (1,-1) coordinate (rqc_coord);
		\draw (2,-1) coordinate (rqa_coord);
		\draw (3,-1) coordinate (rq_coord);
	\end{scope}
	
	\foreach \x/\lbl in { c/C^\ast,  qc/qc,  qa/qa,  q/q,  rqc/qc,  rqa/qa,  rq/q,
	                     oc/C^\ast, oqc/qc, oqa/qa, oq/q, roqc/qc, roqa/qa, roq/q} {
		\draw (\x_coord) node[label_node] (\x) {$\lbl$};
	}
	
	\draw[->,nonOracularToOracular] (c) -- (oc);
	\draw[->,nonOracularToOracular] (qc) -- (oqc);
	\draw[->,nonOracularToOracular] (qa) -- (oqa);
	\draw[->,nonOracularToOracular] (q) -- (oq);
	\draw[->,nonOracularToOracular] (rqc) -- (roqc);
	\draw[->,nonOracularToOracular] (rqa) -- (roqa);
	\draw[->,nonOracularToOracular] (rq) -- (roq);
	
	\draw (roqc.south west) ++(-\boxwidth,-\boxwidth) coordinate (oqc_ll);
	\draw (oqc.north east) ++(\boxwidth,\boxwidth) coordinate (oqc_ur);
	\fill[box_fill] (oqc_ll) rectangle (oqc_ur);
	
	\draw (rqc.south west) ++(-\boxwidth,-\boxwidth) coordinate (qc_ll);
	\draw (qc.north east) ++(\boxwidth,\boxwidth) coordinate (qc_ur);
	\fill[box_fill] (qc_ll) rectangle (qc_ur);
	
	\draw (roqa.south west) ++(-\boxwidth,-\boxwidth) coordinate (oqa_ll);
	\draw (roq.north east) ++(\boxwidth,\boxwidth) coordinate (oq_ur);
	\draw (oqa.north east) ++(\boxwidth,\boxwidth) coordinate (oqa_ur);
	\draw (oqa_ll) ++(0,10pt) coordinate (oqa_ll_u);
	\draw (oqa_ll) ++(10pt,0) coordinate (oqa_ll_r);
	\fill[box_fill] (oqa_ll_r) -| (oq_ur) -| (oqa_ur) -| (oqa_ll_u) |- (oqa_ll_r);
	
	\draw (rqa.south west) ++(-\boxwidth,-\boxwidth) coordinate (qa_ll);
	\draw (rq.north east) ++(\boxwidth,\boxwidth) coordinate (q_ur);
	\draw (qa.north east) ++(\boxwidth,\boxwidth) coordinate (qa_ur);
	\draw (qa_ll) ++(0,10pt) coordinate (qa_ll_u);
	\draw (qa_ll) ++(10pt,0) coordinate (qa_ll_r);
	\fill[box_fill] (qa_ll_r) -| (q_ur) -| (qa_ur) -| (qa_ll_u) |- (qa_ll_r);
	
	\foreach \x/\lbl in {  qc/qc,  qa/qa,  rqc/qc,  rqa/qa,  rq/q,
	                      oqc/qc, oqa/qa, roqc/qc, roqa/qa, roq/q} {
		\draw (\x_coord) node[anchor=base] {$\lbl$};
	}
	
	\draw[->] ([yshift=\arrowbaseshift]oc.base east) -- ([yshift=\arrowbaseshift]oqc.base west);
	\draw[->] ([yshift=\arrowbaseshift]oqc.base east) -- ([yshift=\arrowbaseshift]oqa.base west);
	\draw[->] ([yshift=\arrowbaseshift]oqa.base east) -- ([yshift=\arrowbaseshift]oq.base west);
	\draw[->] ([yshift=\arrowbaseshift]roqc.base east) -- ([yshift=\arrowbaseshift]roqa.base west);
	\draw[->] ([yshift=\arrowbaseshift,yshift=\doublearrowshift]roqa.base east) -- ([yshift=\arrowbaseshift,yshift=\doublearrowshift]roq.base west);
	\draw[->] ([yshift=\arrowbaseshift,yshift=-\doublearrowshift]roq.base west) -- ([yshift=\arrowbaseshift,yshift=-\doublearrowshift]roqa.base east);
	\draw[->,CstarToRobustQc] (oc) -- (roqc);
	\draw[->,GappedToGapless] ([xshift=-\doublearrowshift]roqc.north) -- ([xshift=-\doublearrowshift]oqc.south);
	\draw[->,GappedToGapless] ([xshift=\doublearrowshift]oqc.south) -- ([xshift=\doublearrowshift]roqc.north);
	\draw[->,GappedToGapless] ([xshift=-\doublearrowshift]roqa.north) -- ([xshift=-\doublearrowshift]oqa.south);
	\draw[->,GappedToGapless] ([xshift=\doublearrowshift]oqa.south) -- ([xshift=\doublearrowshift]roqa.north);
	\draw[->,GappedToGapless] (roq) -- (oq);
	
	\draw[->] ([yshift=\arrowbaseshift]c.base east) -- ([yshift=\arrowbaseshift]qc.base west);
	\draw[->] ([yshift=\arrowbaseshift]qc.base east) -- ([yshift=\arrowbaseshift]qa.base west);
	\draw[->] ([yshift=\arrowbaseshift]qa.base east) -- ([yshift=\arrowbaseshift]q.base west);
	\draw[->] ([yshift=\arrowbaseshift]rqc.base east) -- ([yshift=\arrowbaseshift]rqa.base west);
	\draw[->] ([yshift=\arrowbaseshift,yshift=\doublearrowshift]rqa.base east) -- ([yshift=\arrowbaseshift,yshift=\doublearrowshift]rq.base west);
	\draw[->] ([yshift=\arrowbaseshift,yshift=-\doublearrowshift]rq.base west) -- ([yshift=\arrowbaseshift,yshift=-\doublearrowshift]rqa.base east);
	\draw[->,CstarToRobustQc] (c) -- (rqc);
	\draw[->,GappedToGapless] ([xshift=-\doublearrowshift]rqc.north) -- ([xshift=-\doublearrowshift]qc.south);
	\draw[->,GappedToGapless] ([xshift=\doublearrowshift]qc.south) -- ([xshift=\doublearrowshift]rqc.north);
	\draw[->,GappedToGapless] ([xshift=-\doublearrowshift]rqa.north) -- ([xshift=-\doublearrowshift]qa.south);
	\draw[->,GappedToGapless] ([xshift=\doublearrowshift]qa.south) -- ([xshift=\doublearrowshift]rqa.north);
	\draw[->,GappedToGapless] (rq) -- (q);
	
	\draw ($(qc_ur)!.5!(oqc_ll)$) coordinate (midY);
	\draw[nonOracularToOracularSeparation] (midY -| c.base west) -- (midY -| oq_ur);
	
	\draw ($(qa)!.5!(rq)$) coordinate (midQaQ);
	\draw[nonOracularToOracularSeparation] (midY -| midQaQ) |- (midQaQ -| oq_ur);
	
	\draw ($(qc)!.5!(rqa)$) coordinate (midQcQa);
	\draw[nonOracularToOracularSeparation, densely dashed] (midY -| midQcQa) -- (midQcQa |- qc_ll);
	
	\coordinate (legend) at (-1,0);
	\draw (oc_coord.base -| legend) node[anchor=base east] {Oracular gapless};
	\draw (roqc_coord.base -| legend) node[anchor=base east] {Oracular gapped ($c$-$v$-robust)};
	\draw (c.base -| legend) node[anchor=base east] {Non-oracular gapless};
	\draw (rqc.base -| legend) node[anchor=base east] {Non-oracular gapped ($a$-robust)};
	\draw[gray] (legend |- oqc_ur) -- (legend |- qc_ll);
\end{tikzpicture}
    \caption{An overview of the implications and separations for existence of commutativity gadgets in different models. An arrow from one model to another indicates that existence of a commutativity gadget in the first model implies existence of a commutativity gadget in the second. Models that are equivalent in this sense are grouped together in shaded areas. Red lines between the areas indicate separations, meaning that there are CSPs which admit a commutativity gadget in the weaker model but not in the stronger model. The dashed red line indicates a separation that is conditional on the existence of a non-hyperlinear group.}
    \label{fig:intro}
\end{figure}

\begin{enumerate}
    \item \textbf{Polymorphism characterisation across all models.} We extend the characterisation of~\cite{CJM25} from $q$-commutativity gadgets to all min-tensor monoids $Q$, including the quantum approximate ($qa$), commuting-operator ($qc$), and universal ($C^\ast$) settings. In each model, both in the non-oracular and oracular frameworks, existence of a $Q$-commutativity gadget is equivalent to all $Q$-quantum polymorphisms being classical at arity $|A|^2$ (\cref{thm:basic-char,thm:orac-basic-char}).

    \item \textbf{Robust commutativity gadgets and stable commutativity.} We introduce \emph{stable commutativity}, a notion capturing commutativity of an algebra in approximate representions. Using this, we characterise robust commutativity gadgets in terms of stable commutativity of polymorphism algebras at the same bounded arity (\cref{thm:a-robust-char,thm:c-v-robust-char}). This yields a number of equivalences and implications between different flavours of commutativity gadget (\cref{thm:implications-equivalences}).
    However, the behaviour of these implications is a bit subtle: if existence of commutativity gadgets in one model implies existence in another, this does not always mean that every commutativity gadget in the first model is a commutativity gadget in the second.
    The equivalences and implications that we prove are illustrated by the arrows in~\cref{fig:intro}.

    \item \textbf{Quantum polymorphisms of complete graphs.} We prove that the quantum polymorphism space of $K_n$ (for $n \geq 3$) decomposes completely: $\mor^+(K_n^k, K_n) \cong (S_n^+)^{\sqcup k}$, where $S_n^+$ is the quantum permutation group (\cref{thm:sn-decomposition}). In particular, the noncommutative behaviour of the quantum polymorphisms is entirely governed by $S_n^+$, which acts independently on each coordinate, similarly to the action of $S_n$ for the classical polymorphisms. This structural result is crucial for our separation theorems.

    \item \textbf{Separations between commutativity gadget classes.} Combining the polymorphism characterisation with techniques from combinatorial group theory and the theory of linear system games, we exhibit explicit separations between the different classes of commutativity gadgets (\cref{thm:separations}). Specifically, we construct a relational structure that admits a $q$-commutativity gadget but no $qa$-commutativity gadget. Conditionally on the existence of a non-hyperlinear group, we also exhibit a structure admitting a $qa$-commutativity gadget but no $qc$-commutativity gadget. These separations occur in the non-oracular framework (see \cref{fig:intro}). The separations also apply to CSP nonlocal games, since the CSPs we construct have arity~$2$. Analogous questions in the oracular framework remain open.

    \item \textbf{Undecidability of gadget existence.} We show that determining whether a given relational structure admits a $q$- or $qa$-commutativity gadget is undecidable, and that the analogous problem for $qc$- and $C^\ast$-commutativity gadgets is $\tsf{coRE}$-hard (\cref{thm:gadget-existence-hard}). These results are obtained by reducing from the undecidability of properties of solution groups of linear systems~\cite{Slo19,Slo20}.
\end{enumerate}

\heading{Organisation.} \ifthenelse{\boolean{anonymous}}{}{In \cref{sec:prelims}, we recall the necessary background on $\ast$-algebras, relational structures, quantum homomorphism spaces, nonlocal games, and commutativity gadgets.} \cref{sec:comm-gadget-hardness} discusses the connection between commutativity gadgets and hardness of quantum CSPs. In \cref{sec:gadget-polymorphisms}, we prove the polymorphism characterisations in the non-robust setting, extending the results of~\cite{CJM25}. \cref{sec:stability} introduces stable commutativity and establishes its basic properties. \cref{sec:robust-gadgets} contains the robust characterisation theorems and the resulting implications and equivalences between gadget classes. \cref{sec:complete-graphs} gives the full description of the quantum polymorphisms of complete graphs, and \cref{sec:separations} applies this to construct the separations. Finally, \cref{sec:gadget-complexity} addresses the complexity of deciding commutativity gadget existence. \ifthenelse{\boolean{anonymous}}{All necessary background on $\ast$-algebras, relational structures, quantum homomorphism spaces, nonlocal games, and commutativity gadgets can be found in \cref{sec:prelims}.}{}

\ifthenelse{\boolean{anonymous}}{}{
\paragraph{Acknowledgements.} 

We would like to thank William Slofstra and Matthijs Vernooij for helpful discussions. 
EC was supported by a CGS D scholarship from Canada's NSERC.
JvDdB was supported by GA\v{C}R grant 25-17377S.
PZ was funded by the European Union (ERC, POCOCOP, 101071674). Views and opinions expressed are however those of the author(s) only and do not necessarily reflect those of the European Union or the European Research Council Executive Agency. Neither the European Union nor the granting authority can be held responsible for them.
This paper was written in part while JvDdB and PZ were in residence at Institut Mittag-Leffler in Djursholm, Sweden during the research semester ``Operator Algebras and Quantum Information'' in the spring of 2026, supported by the Swedish Research Council under grant no.~2021-06594.
}

\ifthenelse{\boolean{anonymous}}{}{
\section{Preliminaries}

\label{sec:prelims}

\subsection{Notation}

We denote the space of bounded operators on a Hilbert space $H$ as $B(H)$. The space of $n\times n$ matrices is $\mbb{M}_n=B(\C^n)$.

The commutator of two operators is denoted $[S,T]=ST-TS$. The group commutator is denoted in the same way $[a,b]=aba^{-1}b^{-1}$.

\subsection{Finitely-presented groups}

Given a finite set $S$ of generators and a finite set of relations $R$ (words in $S\sqcup S^{-1}$), the \emph{finitely-presented group} $\gen*{S}{R}$ is the quotient of the free group generated by $S$ by the normal subgroup~$\angnormal*{R}$ generated by $R$. The \emph{free product} of two finitely presented groups $G_1=\gen*{S_1}{R_1}$ and $G_2=\gen*{S_2}{R_2}$ is the finitely-presented group $G_1\ast G_2=\gen*{S_1\sqcup S_2}{R_1\sqcup R_2}$. If $G_1$ and $G_2$ contain $H_1$ and $H_2$ isomorphic to $H$, then the \emph{amalgamated product} of $G_1$ and $G_2$ over $H$, denoted $G_1\ast_H G_2$, is the quotient of $G_1\ast G_2$ by the normal subgroup generated by the relations $\varphi_1(h)^{-1}\varphi_2(h)$ for all $h\in H$, where $\varphi_i:H\rightarrow H_i$ are the isomorphisms. If $H$ is finitely presented, then $G_1\ast_H G_2$ remains finitely presented. The natural homomorphisms $G_1\rightarrow G_1\ast_H G_2$ and $G_2\rightarrow G_1\ast_H G_2$ are injective~\cite{KS70}.

\subsection{\texorpdfstring{$\ast$}{\textasteriskcentered{}}-algebras to \texorpdfstring{$C^\ast$}{\textit{C}*}-algebras}

A unital \emph{$\ast$-algebra} is an algebra $\mc{A}$ over $\C$ (or $\R$) equipped with an antilinear involution $x\mapsto x^\ast$ such that $1^\ast=1$ and $(xy)^\ast=y^\ast x^\ast$. Denote the (real) subspace of hermitian elements by $\mc{A}_h=\set*{x\in\mc{A}}{x^\ast=x}$. A \emph{$\ast$-positive cone} over $\mc{A}$ is a set $\mc{A}_+\subseteq\mc{A}_h$ such that $1\in\mc{A}_+$, $x+y\in\mc{A}_+$ for all $x,y\in\mc{A}_+$, and $axa^\ast\in\mc{A}_+$ for all $x\in\mc{A}_+$ and $a\in\mc{A}$. A $\ast$-positive cone induces an order on $\mc{A}_h$ via $x\geq y$ iff $x-y\in\mc{A}_+$. We say $\mc{A}_+$ is \emph{archimedean} if for all $x\in\mc{A}$, there exists $R\in\R$ such that $x^\ast x\leq R1$; a $\ast$-algebra equipped with an archimedean $\ast$-positive cone is called a \emph{semi-pre-$C^\ast$-algebra}, following the notation of~\cite{Oza13}. A standard choice of $\ast$-positive cone is the \emph{sum-of-squares cone} given by $\mc{A}_+=\set*{\sum_{i=1}^nx_i^\ast x_i}{n\in\N,\;x_i\in\mc{A}}$. Given a set of generators $S$ and relations $R$, we denote the \emph{$\ast$-algebra generated $S$ subject to $R$} as $\C\!\gen*{S}{R}=F(S)/\angnormal*{R}$, where $F(S)$ is the \emph{free algebra} with generators $S$ and $\angnormal*{R}$ is the two-sided ideal of $F(S)$ generated by $R$. Note that a finitely-generated $\ast$-algebra is archimedean if and only if all its generators are bounded.

A semi-pre-$C^\ast$-algebra is a \emph{pre-$C^\ast$-algebra} if $x^\ast x\leq\varepsilon 1$ for all $\varepsilon>0$ implies that $x=0$ in $\mc{A}$. If $\mc{A}$ is a pre-$C^\ast$-algebra, we can define a norm on it by
\begin{align*}
    \norm{x}=\sup\set*{\norm{\pi(x)}}{\pi:\mc{A}\rightarrow B(H)\text{ is a $\ast$-homomorphism}}=\inf\set*{\lambda\in\R_{\geq0}}{x^\ast x\leq\lambda^21}.
\end{align*}
This norm satisfies the properties $\norm{xy}\leq\norm{x}\norm{y}$ and $\norm{x^\ast x}=\norm{xx^\ast}=\norm{x}^2$. See~\cite{Oza13} for details. A \emph{$C^\ast$-algebra} is a pre-$C^\ast$-algebra that is complete with respect to this norm. In a $C^\ast$-algebra $\mc{A}$, there is a unique $\ast$-positive cone given by $\mc{A}_+=\set*{x^\ast x}{x\in\mc{A}}$.

Given a semi-pre-$C^\ast$-algebra $\mc{A}$, we can convert it to a pre-$C^\ast$-algebra by taking the quotient by the infinitesimal ideal $I(\mc{A})=\set*{x\in\mc{A}}{x^\ast x\leq\varepsilon 1\;\forall\,\varepsilon\geq 0}$. We call the completion of $\mc{A}/I(\mc{A})$ with respect to the norm above the \emph{universal $C^\ast$-algebra} of $\mc{A}$, and denote it $C_u^\ast(\mc{A})$. Given a set of generators $S$ and a set of relations $R$, the \emph{$C^\ast$-algebra generated by $S$ subject to $R$} is $C^\ast\!\gen*{S}{R}=C^\ast_u(\C\!\gen{S}{R})$.

Given a (finitely-presented) group $G$, the \emph{group $\ast$-algebra} of $G$ is the $\ast$-algebra $\C[G]$ generated by $e_g$ for $g\in G$, subject to the relations $e_g^\ast=e_{g^{-1}}$ and $e_ge_h=e_{gh}$ for all $g,h\in G$. The $\ast$-representations of $\C[G]$ are in bijective correspondence with the unitary representations of $G$. We can see $\C[G]$ as a semi-pre-$C^\ast$-algebra with respect to the sum-of-squares cone. Then, the \emph{group $C^\ast$-algebra} of $G$ is the universal $C^\ast$-algebra $C^\ast G=C^\ast_u(\C[G])$.

A \emph{state} on a semi-pre-$C^\ast$-algebra $\mc{A}$ is a linear map $\rho:\mc{A}\rightarrow\C$ such that $\rho(1)=1$ and $\rho(x)\geq 0$ for all $x\in\mc{A}_+$. We say $\rho$ is \emph{tracial} if $\rho(xy)=\rho(yx)$ for all $x,y\in\mc{A}$; and we say $\rho$ is \emph{faithful} if $\rho(x^\ast x)=0$ implies $x=0$ for all $x\in\mc{A}_+$ The existence of a faithful state automatically guarantees that $\mc{A}$ is a pre-$C^*$-algebra, since the infinitesimal ideal is trivial. Tracial states are usually denoted by~$\tau$. Any state induces a seminorm $\norm{a}_\rho=\sqrt{\rho(a^\ast a)}$ called the \emph{$\rho$-norm}. Every state $\rho$ on a $\ast$-algebra induces a \emph{GNS representation}, which is a $\ast$-representation $\pi:\mc{A}\rightarrow B(H)$ such that there is a unit vector $\ket{\psi}\in H$ satisfying $\rho(x)=\braket{\psi}{\pi(x)}{\psi}$. A state is called \emph{finite-dimensional} if it admits a finite-dimensional GNS representation; a state is called \emph{Connes-embeddable} if it admits a GNS representation where the von Neumann algebra generated by $\pi(\mc{A})$ is Connes-embeddable (see below). A \emph{character} is a state that is a $\ast$-representation.

The \emph{tensor product} of $\ast$-algebras $\mc{A}$ and $\mc{B}$ is denoted $\mc{A}\otimes\mc{B}$, or sometimes $\mc{A}\odot\mc{B}$ in order to emphasize that we are not taking a completion. We implicitly use the isomorphism $\C\otimes\mc{A}\cong\mc{A}\otimes\C\cong\mc{A}$. For $C^\ast$-algebras $\mc{A}$ and $\mc{B}$, there are a variety of (generally inequivalent) ways to complete $\mc{A}\odot\mc{B}$ to a $C^\ast$-algebra. Notably, we make use of the \emph{minimal tensor product} (or min-tensor product), which is the completion $\mc{A}\otimes_{\min}\mc{B}$ of $\mc{A}\odot\mc{B}$ with respect to the norm $\norm{x}_{\min}=\norm{(\pi_A\otimes\pi_B)(x)}$ for $\pi_A$ and $\pi_B$ faithful $\ast$-representations of $\mc{A}$ and $\mc{B}$, respectively.

\subsection{von Neumann algebras}

A \emph{von Neumann algebra} is a unital $\ast$-subalgebra $\mc{M}\subseteq B(H)$ for a Hilbert space $H$ that is closed under the \emph{weak operator topology}, defined via convergence as $(x_\lambda)\rightarrow x$ iff $\braket{\psi}{x_\lambda}{\phi}\rightarrow\braket{\psi}{x}{\phi}$ for all $\ket{\psi},\ket{\phi}\in H$. Equivalently, $\mc{M}$ is a unital $\ast$-subalgebra of $B(H)$ equal to its bicommutant $\mc{M}''$, where the \emph{commutant} of $S\subseteq B(H)$ is $S'=\set*{a\in B(H)}{ax=xa\;\forall\,x\in S}$; or that $\mc{M}$ is closed under the \emph{strong operator topology}, defined via convergence as $(x_\lambda)\rightarrow x$ iff $\norm{(x_\lambda-x)\ket{\psi}}\rightarrow 0$ for all $\ket{\phi}\in H$. A von Neumann algebra is a \emph{factor} if $\mc{M}\cap\mc{M}'=\C\cdot 1$; any von Neumann algebra can be decomposed as a direct integral of factors. The finite-dimensional factors are exactly the matrix algebras $\mbb{M}_n$. A von Neumann algebra is \emph{finite} if there exists a faithful tracial state on~$\mc{M}$. An important example of a finite von Neumann algebra is the \emph{hyperfinite type $\mathrm{II}_1$ factor}~$\mc{R}$. One way of constructing $\mc{R}$ is by taking a closure of the algebra $\mc{A}=\bigcup_{n=1}^\infty\mbb{M}_{2^n}/\sim$, where $\sim$ is the equivalence relation $x\otimes 1\sim x$. To find the right notion of closure, note first that there is a tracial state $\tau$ on $\mc{A}$ given by the normalised trace on each component. We can complete $\mc{A}$ as a Hilbert space $H$ with respect to the inner product $\ang*{x,y}=\tau(x^\ast y)$. Then, $\mc{A}\subseteq B(H)$ via the action of left multiplication, and we can take the closure in the strong operator topology to get $\mc{R}$. A finite von Neumann algebra is \emph{Connes-embeddable} if there is an injective $\ast$-homomorphism $\mc{M}\rightarrow\mc{R}^\omega$, where $\mc{R}^\omega$ is an ultrapower of the hyperfinite $\mathrm{II}_1$ factor.

To define $\mc{R}^\omega$, we need the concept of an ultrafilter. An \emph{ultrafilter} $\omega$ on the natural numbers is a collections of subsets of $\N$ such that $\varnothing\notin\omega$; if $A,B\in\omega$, then $A\cap B\in\omega$; and for all $A\subseteq\N$, either $A\in\omega$ or $A^c\in\omega$. Note that these properties also imply that if $A\in\omega$ and $B\supseteq A$, then $B\in\omega$ as well. An ultrafilter is called \emph{principal} if it contains a singleton, and \emph{free} otherwise. The existence of a free ultrafilter follows from the axiom of choice. Ultrafilters allow us to define a robust notion of limit. Given an ultrafilter $\omega$, we say $x$ is the \emph{ultralimit} of a sequence $(x_n)$ if for each $\varepsilon>0$, $\set*{n\in\N}{|x-x_n|<\varepsilon}\in\omega$; we denote this as $\lim_{n\rightarrow\omega}x_n=x$ or $x_n\overset{\omega}{\rightarrow}x$. It holds that if the ultralimit exists, it is unique; and every bounded sequence admits an ultralimit. Further, if $\omega$ is free and $(x_n)$ converges in the usual way, then $\lim_{n\rightarrow\omega}x_n=\lim_{n\rightarrow\infty}x_n$. Given a sequence of finite von Neumann algebras $(\mc{M}_n,\tau_n)$, their \emph{ultraproduct} is the space of bounded sequences $(a_n)$ such that $a_n\in\mc{M}_n$, modulo the equivalence relation $(a_n)\sim(b_n)$ if $\norm{a_n-b_n}_{\tau_n}\overset{\omega}{\rightarrow}0$. This space, denoted $\prod_{n\rightarrow\omega}\mc{M}_n$ or $\prod_{\omega}\mc{M}_n$ if the index is clear from context, is a finite von Neumann algebra with tracial state $\tau_\omega(a)=\lim_{n\rightarrow\omega}\tau_n(a_n)$ for a representative $(a_n)$ of $a$. If $\mc{M}_n=\mc{M}$ for each $n$, then we write $\prod_\omega\mc{M}_n=\mc{M}^{\omega}$, and call it the \emph{ultrapower} of $\mc{M}$.

It is known that $\mc{R}^\omega\cong\mc{R}^{\omega'}$ for all free ultrafilters if and only if the continuum hypothesis holds (see~\cite{Gol23}). However, we can usually make use of weaker relations between different ultraproducts: below, we provide a proof of the (folklore) result that quantifies the sense in which we can treat $\mc{R}^\omega$ as ultraproduct of matrix algebras of increasing dimension.

\begin{lemma}\label{lem:ultrapower}
    Let $\omega$ be an ultrafilter on $\N$.
    \begin{enumerate}[(i)]
        \item For any sequence of integers $(k(n))_n$, there is an embedding $\prod_\omega\mbb{M}_{k(n)}\hookrightarrow\mc{R}^\omega$.
        \item There exists an ultrafilter $\omega'$ on $\N$ and a sequence $(k(n))_n$ such that there is an embedding~$\mc{R}^\omega\hookrightarrow\prod_{\omega'}\mbb{M}_{k(n)}$.
    \end{enumerate}
\end{lemma}

\begin{proof}
    \begin{enumerate}[(i)]
        \item By a standard result, there is a trace-preserving embedding $\iota_n:\mbb{M}_n\rightarrow\mc{R}$ for all~$n$~\cite{MvN43}. Define the map $\iota:\Pi_\omega\mbb{M}_{k(n)}\rightarrow\mc{R}^\omega$ by taking $\iota(a)$ to be the element with representative $(\iota_{k(n)}(a_n))$, for $(a_n)$ a representative of $a$. First, $\iota$ is well-defined: if $(a_n)\sim(b_n)$, then, $\norm{a_n-b_n}_{\tr_{k(n)}}\overset{\omega}{\rightarrow}0$; as $\iota_n$ is trace-preserving, then $\norm{\iota_{k(n)}(a_n)-\iota_{k(n)}(b_n)}_\tau\overset{\omega}{\rightarrow}0$. Therefore, $\iota$ is a $\ast$-homomorphism. Next, suppose the image of $\iota(a)=0$ in $\mc{R}^\omega$. Then, letting $(a_n)$ be a representative of $a$, $0=\lim_{n\rightarrow\omega}\tau(\iota_{k(n)}(a_n)^\ast \iota_{k(n)}(a_n))=\lim_{n\rightarrow\omega}\tr_{k(n)}(a_n^\ast a_n)$, so $a=0$ in $\prod_\omega\mbb{M}_{k(n)}$. So $\iota$ is injective.

        \item $\mc{R}$ is constructed as the strong operator topology closure of $\bigcup_{n=1}^\infty\mbb{M}_{2^n}$, subject to the natural identification $x\otimes 1=x$. In particular, for all $x\in\mc{R}$, there exists a sequence $(x_m)$ such that $x_m\in\mbb{M}_{2^m}$ and $x_m\rightarrow x$ with respect to the trace norm; via the axiom of choice, we can choose such a sequence for every element of $\mc{R}$. Fix any free ultrafilter $\upsilon$ on $\N$. Using~\cite[Definition III.2.3]{CLP15}, consider $\mc{M}=\prod_{(n,m)\rightarrow\omega\times\upsilon}\mbb{M}_{2^m}$. Since $\N\times\N\cong\N$, $\omega\times\upsilon$ can be seen as a free ultrafilter on $\N$, and hence we intend to construct an embedding of $\mc{R}^\omega$ in $\mc{M}$. Define the map $\phi:\mc{R}^\omega\rightarrow\mc{M}$ as $\phi(a)=[(a_{nm})_{n,m}]$ where $(a_n)$ is a representative of $a$. First, $\phi$ is well-defined, as for any $(a_n)\sim (b_n)$, $$\lim_{(n,m)\rightarrow\omega\times\upsilon}\norm{a_{nm}-b_{nm}}_{\tr_{2^m}}=\lim_{n\rightarrow\omega}\lim_{m\rightarrow\upsilon}\norm{a_{nm}-b_{nm}}_{\tr_{2^m}}=\lim_{n\rightarrow\omega}\norm{a_n-b_n}_\tau=0.$$
        Also, $\phi$ is a $\ast$-homomorphism. For example, if $a,b\in\mc{R}^\omega$ then $\phi(a+b)=\phi(a)+\phi(b)$ as
        $$\lim_{(n,m)\rightarrow\omega\times\upsilon}\norm{(a_n+b_n)_m-a_{nm}-b_{nm}}_{\tr_{2^m}}=\lim_{n\rightarrow\omega}\norm{a_n+b_n-a_n-b_n}_\tau=0.$$
        $\phi$ preserves the identity, products, and the adjoint in the same way. To finish, we show $\phi$ is injective. Suppose $\phi(a)=0$. Then, $0=\lim_{(n,m)\rightarrow\omega\times\upsilon}\norm{a_{nm}}_{\tr_{2^m}}=\lim_{n\rightarrow\omega}\norm{a_n}_\tau=0$, so~$a=0$.
    \end{enumerate}
\end{proof}

\subsection{Weighted algebras}

A \emph{weighted algebra} is a pair $(\mc{A},\mu)$ where $\mc{A}$ is a $\ast$-algebra and $\mu:\mc{A}\rightarrow\R_{\geq 0}$ is a finitely-supported function, called the \emph{weight function}. The \emph{defect polynomial} of a weighted algebra is
\begin{align*}
    D(\mc{A},\mu)=\sum_{a\in\mc{A}}\mu(a)a^\ast a.
\end{align*}
Given a tracial state $\tau:\mc{A}\rightarrow\C$, the \emph{defect} of $\tau$ is $$\defect(\tau;\mu)=\tau(D(\mc{A},\mu))=\sum_{a\in\mc{A}}\mu(a)\norm{a}_\tau^2.$$
We denote it $\defect(\tau)$ if the weight function is clear from context. Weighted algebras capture approximate representations of $\mc{A}/\angnormal{\supp(\mu)}$ and $\defect(\tau)$ captures how much the GNS representation of the tracial state $\tau$ deviates from an exact representation of this algebra.

Since we consider only the representations of weighted algebras induced by tracial states, the natural notion of order is weaker than that induced by the sums of squares. We say $a,b\in\mc{A}$ are \emph{cyclically equivalent} if there exist $c_1,d_1,\ldots,c_k,d_k\in\mc{A}$ such that $a-b=\sum_i[c_i,d_i]$; we denote this by $a\cyceq b$. Write $a\gtrsim b$ if there exist $s_1,\ldots,s_k\in\mc{A}$ such that $a-b\cyceq\sum_is_i^\ast s_i$.

A \emph{$C$-homomorphism} $\alpha:(\mc{A},\mu)\rightarrow(\mc{B},\nu)$ is a $\ast$-homomorphism $\alpha:\mc{A}\rightarrow\mc{B}$ such that $\alpha(D(\mc{A},\mu))\lesssim C\cdot D(\mc{B},\nu)$. If there is a $C$-homomorphism $\alpha:(\mc{A},\mu)\rightarrow(\mc{B},\nu)$, then for every tracial state $\tau$ on $\mc{B}$, there exists a tracial state $\tau'=\tau\circ\alpha$ on $\mc{A}$ such that $\defect(\tau')\leq C\defect(\tau)$.

The following inequality from \cite{CDdBVZ25} is often useful in the context of weighted algebras: for $a_1,\ldots,a_k\in\mc{A}$, and writing $\abs{a}^2=a^\ast a$,
\begin{align*}
    \abs[\Big]{\sum_{i=1}^ka_i}^2\leq k\sum_{i=1}^k\abs{a_i}^2.
\end{align*}

\subsection{Relational structures}


A \emph{relational structure} is a pair $A=(\dom(A),\rel(A))$ where $\dom(A)$ is a (usually finite) set called the \emph{domain} of $A$ and $\rel(A)$ is a (usually finite) set called the \emph{relations} of $A$, where for each $R\in \rel(A)$, there exists $n\in\N$, called the \emph{arity} of $R$ and denoted $\ar(R)$, such that $R\subseteq\dom(A)^{\ar(R)}$. When clear from context, we write $A$ to mean $\dom(A)$.

A \emph{relational signature} is a (usually finite) set $\sigma$ with an associated arity function $\ar:\sigma\rightarrow\N$. The elements of $\sigma$ are called \emph{relational symbols}. A \emph{relational structure over $\sigma$} (or a \emph{$\sigma$-structure}) is a relational structure $A$ with an associated bijective mapping $\sigma\rightarrow\rel(A)$, $R\mapsto R^A$ such that $\ar(R^A)=\ar(R)$. If $\sigma\subseteq\tau$, then every $\sigma$-structure is also a $\tau$-structure. A \emph{homomorphism} of $\sigma$-structures $f:A\rightarrow B$ is a map $f:\dom(A)\rightarrow\dom(B)$ such that for all $R\in\sigma$ and $(a_1,\ldots,a_{\ar(R)})\in R^A$, $(f(a_1),\ldots,f(a_{\ar(R)}))\in R^B$.

Given a relational structure $A$ over $\sigma$, the \emph{constraint satisfaction computational problem (CSP)} over $A$ is the problem $\CSP(A)$ with instances that are relational structures over $\sigma$, where $B$ is a yes instance if there exists a homomorphism $B\rightarrow A$ (\textit{i.e.} $\mor(B,A)$ is nonempty), and $B$ is a no instance otherwise. Due to the \emph{CSP dichotomy theorem}, we know that $\CSP(A)$ is either contained in $\tsf{P}$ or $\tsf{NP}$-complete~\cite{Bul17,Zhu17}.

Let $A$ and $B$ be $\sigma$-structures. The \emph{cartesian product} of $A$ and $B$ is the $\sigma$-structure~$A\times B$ with domain $\dom(A\times B)=\dom(A)\times\dom(B)$ and relations $$R^{A\times B}=\set*{((a_1,b_1),\ldots,(a_{\ar(R)},b_{\ar(R)}))}{\mathbf{a}\in R^A,\;\mathbf{b}\in R^B}.$$

Write $A^k=A\times A\times\cdots\times A$. Note that for graphs, the cartesian product as relational structures is the same as the categorical (or tensor) product as graphs.

We denote elements of a relational structure in standard italic font $a\in A$, tuples of elements with bold face $\mbf{a}=(a_1,\ldots,a_k)\in A^k$, and tuples of tuples of elements (as in the relations of a power of a relational structure) with blackboard bold face $\mbbf{a}=(\mbf{a}_1,\ldots,\mbf{a}_{l})=((a_{11},\ldots,a_{1k}),\ldots,(a_{l1},\ldots,a_{lk}))\in (A^k)^{l}$.

Let $\sigma_{\mathrm{LIN}}=\set*{\LR{b,n}}{b\in\Z_2,\,n\in\N}$ with $\ar(\LR{b,n})=n$; $\mathrm{LIN}$ is the $\sigma_{\mathrm{LIN}}$-structure with $\dom(\mathrm{LIN})=\Z_2$ and $\LR{b,n}^{\mathrm{LIN}}=\set*{\mbf{a}\in\Z_2^n}{a_1+\ldots+a_n=b}$. A relational structure over $\sigma_{\mrm{LIN}}$ is called \emph{linear}. A linear relational structure $A$ is called \emph{homogeneous} if $\LR{1,n}^A=\varnothing$ for all $n\in\N$. For a linear relational structure $A$, define its \emph{homogenisation} $A_H$ as the homogeneous linear relational structure defined as $\dom(A_H)=\dom(A)$, $\LR{0,n}^{A_H}=\LR{0,n}^{A}\cup\LR{1,n}^{A}$, and $\LR{1,n}^{A_H}=\varnothing$.

We will make use of the following standard result relating polymorphisms of a relational structure with the expressivity of its CSP.

\begin{definition}
	Let $A$ and $B$ be relational structures such that $\dom(A)=\dom(B)$. We say that $A$ \emph{pp-defines} $B$ if for all $S\in\rel(B)$, there exist $n,m\in\N$, $R_1,\ldots,R_n$ either in $\rel(A)$ or equal to the equality relation $\set*{(a,a)}{a\in A}$, $i_k\in[m]$ for $k\in[\ar(S)]$, and $i_{jk}\in[m]$ for $j\in[n]$ and $k\in[\ar(R_j)]$ such that
	$$S=\set*{(a_{i_1},\ldots,a_{i_{\ar(S)}})}{\mathbf{a}\in A^{m}\land(a_{i_{j1}},\ldots,a_{i_{j\ar(R_j)}})\in R_j\;\forall\,j\in[n]}.$$
\end{definition}

\begin{lemma}\label{lem:standard}
	Let $A$ and $B$ be relational structures such that $\dom(A)=\dom(B)$. Then $A$ pp-defines $B$ if and only if $\pol(A)\subseteq\pol(B)$.
\end{lemma}

This result was first shown in~\cite{Gei68,BKKR69a,BKKR69b}. See~\cite[Theorem 3.13]{Che09} for a self-contained proof; and see the above or~\cite{BKW17} for a general overview of the theory.

\subsection{Quantum relational structure homomorphisms}

A \emph{quantum space} is a virtual space $X$ corresponding to the ``algebra of functions'' on a $C^\ast$-algebra $C(X)$; if $C(X)$ is commutative, then $X$ can be realised as a compact set by Gelfand duality. The \emph{quantum space of homomorphisms} from a set $R$ to a set $S$ is the quantum space $T^+_{R,S}$ with $C^\ast$-algebra of functions $C(T^+_{R,S})$ generated by $p_{rs}$ for $r\in R$ and $s\in S$, subject to the relations $p_{rs}^2=p_{rs}^\ast=p_{rs}$, and $\sum_{s\in S}p_{rs}=1$. The \emph{quantum space of homomorphisms} from a $\sigma$-structure $A$ to a $\sigma$-structure $B$ is the quantum space $\mor^+(A,B)$ with algebra of functions $C(\mor^+(A,B))=\Mor^+(A,B)$ which is the quotient of $C(T^+_{\dom(A),\dom(B)})$ by the relations $p_{a_1b_1}\cdots p_{a_{\ar(R)}b_{\ar(R)}}=0$ for all $R\in\sigma$, $\mbf{a}\in R^A$, and $\mbf{b}\notin R^B$. The \emph{quantum space of oracular homomorphisms} from $A$ to $B$ is the quantum space $\mor^{o+}(A,B)$ with algebra of functions $C(\mor^{o+}(A,B))=\Mor^{o+}(A,B)$ which is the quotient of $\Mor^+(A,B)$ by the relations $[p_{a_ib},p_{a_jb'}]=0$ for all $R\in\sigma$, $\mbf{a}\in R^A$, $i,j\in[\ar(R)]$, and $b,b'\in B$.

Given $\sigma$-structures $A,B,C$, there is a $\ast$-homomorphism, called \emph{cocomposition}, $$\Delta_{A,C}^B:\Mor^+(A,C)\rightarrow\Mor^+(A,B)\otimes\Mor^+(B,C)$$ such that $\Delta_{A,C}^B(p_{ac})=\sum_{b\in B}p_{ab}\otimes p_{bc}$. Importantly, the cocomposition induces a composition for $\ast$-representations: given $\pi:\Mor^+(A,B)\rightarrow B(\mc{H})$ and $\varphi:\Mor^+(B,C)\rightarrow B(\mc{K})$, then their \emph{composition} is the $\ast$-representation $\varphi\circ\pi:\Mor^+(A,B)\rightarrow B(\mc{H}\otimes\mc{K})$ defined as $\varphi\circ\pi=(\pi\otimes\varphi)\Delta_{A,C}^B$. The same holds identically in the oracular framework.

Let $A$ be a linear relational structure. The \emph{solution group} of $A$ is the finitely-presented group $\Gamma(A)$ with generators $x_a$ for $a\in A$ and $J$ subject to the relations $x_a^2=1$ for all $a\in A$, $J^2=1$, $x_aJ=Jx_a$ for all $a\in A$, and $x_{a_1}\cdots x_{a_{n}}=J^b$ and $x_{a_i}x_{a_j}=x_{a_j}x_{a_i}$ if $\mbf{a}\in \LR{b,n}^A$ and $i,j\in[\ar(R)]$. The \emph{homogeneous solution group} of $A$ is the quotient group $\Gamma_H(A)=\Gamma(A)/\angnormal*{J}$. For linear relational structures $\Mor^{o+}(A,\mrm{LIN})\cong C^\ast\Gamma(A)/\angnormal*{J+1}$, the full group $C^\ast$-algebra modulo the relation $J=-1$, via the $\ast$-homomorphism $p_{ab}\mapsto\frac{1}{2}(1+(-1)^bx_a)$. Also, $\Gamma(A_H)=\Gamma_H(A)\times\Z_2$, so if $A$ is homogeneous $\Mor^{o+}(A,\mrm{LIN})\cong C^\ast\Gamma_H(A)$.

We use the weighted algebra formalism to capture approximate quantum homomorphisms between relational structures, as studied in~\cite{MS24,CM24,CDdBVZ25}. In the following, $A$ and $B$ are relational structures over a signature $\sigma$.

Let $\pi$ be a probability distribution on $\parens*{\bigcup_{R\in\sigma}\{R\}\times R^A}^2$. The \emph{constraint-constraint algebra} is the weighted algebra $\Mor^{c-c}_\pi(A,B)$ generated by PVMs $\{\Phi^{R,\mathbf{a}}_{\mathbf{b}}\}_{\mathbf{b}\in R^B}$ for all $R\in\sigma$ and $\mathbf{a}\in R^A$, equipped with the weight function
\begin{align*}
    \mu_{c-c,\pi}(\Phi^{R,\mathbf{a}}_{\mathbf{b}}\Phi^{R',\mathbf{a}'}_{\mathbf{b}'})=\pi((R,\mathbf{a}),(R',\mathbf{a}'))
\end{align*}
if there exist $i,j\in[\ar(R)]$ such that $a_i=a_j'$ but $b_i\neq b_j'$, and $0$ on all other elements. Denote the defect polynomial $D^{c-c}_\pi(A,B)=D(\Mor^{c-c}_\pi(A,B),\mu_{c-c,\pi})$.

Let $\pi$ be a probability distribution on $\bigcup_{R\in\sigma}\{R\}\times R^A$. The \emph{constraint-variable algebra} is the weighted algebra $\Mor^{c-v}_\pi(A,B)$ generated by PVMs $\{\Phi^{R,\mathbf{a}}_{\mathbf{b}}\}_{\mathbf{b}\in R^B}$ for all $R\in\sigma$ and $\mathbf{a}\in R^A$ and $\{p^a_b\}_{b\in B}$ for all $a\in A$, and equipped with the weight function
\begin{align*}
    \mu_{c-v,\pi}(\Phi^{R,\mathbf{a}}_{\mathbf{b}}(1-p^{a_i}_{b_i}))=\frac{\pi(R,\mathbf{a})}{\ar(R)},
\end{align*}
for all $R\in\sigma$, $\mathbf{a}\in R^A$, and $i\in[\ar(R)]$, and $0$ on all other elements. Denote the defect polynomial $D^{c-v}_\pi(A,B)=D(\Mor^{c-v}_\pi(A,B),\mu_{c-v,\pi})$.

Let $\pi$ be a probability distribution on $\bigcup_{R\in\sigma}\{R\}\times R^A$. The \emph{assignment algebra} is the weighted algebra $\Mor^a_{\pi}(A,B)$ generated by PVMs
$\{p^a_b\}_{b\in B}$ for all $a\in A$, and equipped with the weight function
\begin{align*}
    \mu_{a,\pi}(p^{a_1}_{b_1}\cdots p^{a_{\ar(R)}}_{b_{\ar(R)}})=\pi(R,\mathbf{a}),
\end{align*}
for all $R\in\sigma$, $\mathbf{a}\in R^A$, and $\mathbf{b}\notin R^B$, and $0$ on all other elements. Denote the defect polynomial $D^{a}_\pi(A,B)=D(\Mor^{a}_\pi(A,B),\mu_{a,\pi})$.

\subsection{Nonlocal games}\label{sec:nlg}

A $2$-player \emph{nonlocal game} $\ttt{G}$ consists of finite sets $X,Y,A,B$ called Alice's questions, Bob's questions, Alice's answers, and Bob's answers; a probability distribution $\pi:X\times Y\rightarrow[0,1]$, called the question distribution; and a function $V:A\times B\times X\times Y\rightarrow\{0,1\}$, $(a,b,x,y)\mapsto V(a,b|x,y)$, called the predicate. A strategy for a nonlocal game $\ttt{G}$ is given by a \emph{correlation}, which is a function $p:A\times B\times X\times Y$, $(a,b,x,y)\mapsto p(a,b|x,y)$, such that $(a,b)\mapsto p(a,b|x,y)$ is a probability distribution for all $x,y$. The \emph{value} of a nonlocal game $\ttt{G}$ with respect to a correlation $p$ is
$$\mfk{v}(\ttt{G},p)=\sum_{(x,y,a,b)\in X\times Y\times A\times B}\pi(x,y)V(a,b|x,y)p(a,b|x,y).$$
We say a correlation $p$ is \emph{perfect} for $\ttt{G}$ if $\mfk{v}(\ttt{G},p)=1$.

A correlation is \emph{deterministic} if there are functions $f:X\rightarrow A$, $g:Y\rightarrow B$ such that $p(a,b|x,y)=\delta_{f(x),a}\delta_{g(y),b}$. A correlation is \emph{classical} if there exists a random variable $\Lambda$ on a set $L$ and functions $f:X\times L\rightarrow A$ and $g:Y\times L\rightarrow B$ such that $p(a,b|x,y)=\expec\delta_{f(x,\Lambda),a}\delta_{f(y,\Lambda),b}$. A correlation is \emph{quantum} if there exist finite-dimensional Hilbert spaces $H_A$ and $H_B$, POVMs $\{P^x_a\}_{a\in A}\subseteq B(H_A)$ $\{Q^y_b\}_{b\in B}\subseteq B(H_B)$ for all $x\in X$ and $y\in Y$, and a state $\ket{\psi}\in H_A\otimes H_B$ such that $p(a,b|x,y)=\braket{\psi}{P^x_a\otimes Q^y_b}{\psi}$. A correlation is \emph{quantum approximate} if it is the limit of a sequence of quantum correlations. A correlation is \emph{quantum commuting} if there exists a Hilbert spaces $H$, POVMs $\{P^x_a\}_{a\in A}\subseteq  B(H)$ $\{Q^y_b\}_{b\in B}\subseteq  B(H)$ for all $x\in X$ and $y\in Y$, and a state $\ket{\psi}\in H$ such that $[P^x_a,Q^y_b]=0$ and $p(a,b|x,y)=\braket{\psi}{P^x_aQ^y_b}{\psi}$. A correlation is \emph{tracial} if there exists a von Neumann algebra $\mc{M}$, a POVM $\{P^x_a\}_{a\in A\cup B}\subseteq\mc{M}$ for all $x\in X\cup Y$, and a tracial state $\tau:\mc{M}\rightarrow\C$ such that $p(a,b|x,y)=\tau(P^x_aP^y_b)$. We denote sets of correlations by $C_Q(A,B|X,Y)$, where the symbol $Q=d$ for deterministic, $Q=c$ for classical, $Q=q$ for quantum, $Q=qa$ for quantum approximate, $Q=qc$ for quantum commuting, and $Q=t$ for tracial. We do not specify the question and answer sets when clear from context. Write also $C_{t,Q}(A,B|X,Y)=C_t(A,B|X,Y)\cap C_Q(A,B|X,Y)$. In this case, we have that $\mc{M}=\C$ if $p$ is deterministic, $\mc{M}$ is commutative if $p$ is classical, $\mc{M}$ is finite-dimensional if $p$ is quantum, $\mc{M}$ is Connes-embeddable if $p$ is quantum approximate, and $C_{t,qc}(A,B|X,Y)=C_t(A,B|X,Y)$. The \emph{value} of $\ttt{G}$ with respect to a set of correlations $C_Q$ is $\mfk{v}_Q(\ttt{G})=\sup_{p\in C_Q(A,B|X,Y)}\mfk{v}(\ttt{G},p)$. Note that $\mfk{v}_d(\ttt{G})=\mfk{v}_c(\ttt{G})$, $\mfk{v}_{t,d}(\ttt{G})=\mfk{v}_{t,c}(\ttt{G})$, $\mfk{v}_q(\ttt{G})=\mfk{v}_{qa}(\ttt{G})$, and $\mfk{v}_{t,q}(\ttt{G})=\mfk{v}_{t,qa}(\ttt{G})$.

A nonlocal game is \emph{synchronous} if $X=Y$, $A=B$, and $V(a,b|x,x)=0$ whenever $a\neq b$. Then, if $p$ is perfect for the game, it must also be synchronous. This implies that $p$ is tracial~\cite{PSSTW16}. Further, near-perfect quantum correlations can be well-approximated by tracial correlations~\cite{Vid22,MdlS23}, in the sense that if $\mfk{v}_q(\ttt{G})\geq 1-\varepsilon$, then $\mfk{v}_{t,q}(\ttt{G})\geq1-O((\varepsilon/C)^{1/4})$, under the assumption that the question distribution of $\ttt{G}$ satisfies $\pi(x,y)=\pi(y,x)$ and $\pi(x,x)\geq C\sum_{y\neq x}\pi(x,y)$; the same holds for quantum approximate and quantum commuting correlations. As such, we can always restrict to working with tracial correlations by restricting the games to synchronous games, which may be done by symmetrising the probability distribution and adding consistency checks.

There are several ways to express constraint satisfaction problems as nonlocal games. Let $A$ and $B$ be relational structures over the same signature $\sigma$. For a probability distribution $\pi:(\cup_{R\in\sigma}\{R\}\times R^A)^2\rightarrow[0,1]$, the \emph{constraint-constraint game} $\ttt{G}_{c-c}(A,B,\pi)$ is the synchronous nonlocal game with question sets $\bigcup_{R\in\sigma}\{R\}\times R^A$, answer sets $\bigcup_{R\in\sigma}B^{\ar(R)}$, question distribution $\pi$, and predicate $V(\mathbf{a},\mathbf{b}|(R,\mathbf{x}),(S,\mathbf{y}))=1$ iff $\mathbf{a}\in R^B$, $\mathbf{b}\in S^B$, and $a_i=b_j$ whenever $x_i=y_j$. For a probability distribution $\pi:\cup_{R\in\sigma}\{R\}\times R^A\rightarrow[0,1]$, the \emph{constraint-variable game} $\ttt{G}_{c-v}(A,B,\pi)$ is the nonlocal game with question sets $\bigcup_{R\in\sigma}\{R\}\times R^A$ and $A$, answer sets $\bigcup_{R\in\sigma}B^{\ar(R)}$ and $B$, question distribution $((R,\mathbf{x}),y)\mapsto\sum_{i:\,x_i=y}\frac{\pi(R,\mathbf{x})}{\ar(R)}$, and predicate $V(\mathbf{a},b|(R,\mathbf{x}),y)=1$ iff $\mathbf{a}\in R^B$ and $a_i=b$ when $x_i=y$. A third way of presenting the constraint problem as a nonlocal game arises in the $2$-CSP case, \textit{i.e.} $\ar(R)=2$ for all $R\in\sigma$. For a probability distribution $\pi:\cup_{R\in\sigma}\{R\}\times R^A\rightarrow[0,1]$, the \emph{assignment nonlocal game} $\ttt{G}_a(A,B,\pi)$ is the nonlocal game with question sets $A$, answer sets $B$, question distribution $(x,y)\mapsto\sum_{R\in\sigma:\,(x,y)\in R^A}\pi(R,(x,y))$, and predicate $V(a,b|x,y)=1$ iff $(a,b)\in R^B$ for all $R\in\sigma$ such that $(x,y)\in R^A$. If the probability distribution has full support, perfect deterministic tracial correlations for all three games correspond to relational structure homomorphisms $A\rightarrow B$. Further, there is a perfect classical correlation iff there is a homomorphism $A\rightarrow B$.

A tracial correlation is \emph{oracularisable} if $[P^x_a,P^y_b]=0$ whenever $\pi(x,y)>0$~\cite{JNV+21}. The \emph{oracularisation} of a nonlocal game is the game where Alice is asked both her and Bob's questions and Bob is asked one of the two questions; the players win if Alice responds with a correct answer and Bob responds with answer consistent with Alice's. This presentation of the oracularisation renders it as a constraint-variable game, but there are other equivalent presentations, for example as a constraint-constraint game~\cite{MS24}. The oracularisation transformation of a game is always sound, and it is also complete (with respect to tracial correlations) if the optimal correlation was oracularisable.

\subsection{Commutativity gadgets}

\begin{definition}
	A \emph{min-tensor monoid} is a set $Q$ of $C^\ast$-algebras that is closed under the minimal tensor product and contains $\C$.
	
	A state $\rho$ on a $\ast$-algebra $\mc{A}$ is \emph{over $Q$} if there exists a GNS representation of $\rho$ whose $C^\ast$-algebra is in $Q$.
\end{definition}

We call a $\ast$-representation $\pi:\Mor^+(X,A)\rightarrow\mc{A}$ such that $\mc{A}$ is a $C^\ast$-algebra that embeds into an element of $Q$ a \emph{$Q$-quantum morphisms} from $X$ to $A$, and write $\mor^Q(X,A)$ for the set of such representations. We define the \emph{oracular $Q$-quantum morphisms} $\mor^{oQ}(X,A)$ analogously.

\begin{definition}
	Let $Q$ be a min-tensor monoid, and let $A$ be a relational structure over a signature~$\sigma$. A \emph{(non-oracular) $Q$-commutativity gadget} is a relational structure $G$ over $\sigma$ along with two elements $x,y\in G$, called the \emph{distinguished variables}, such that
	\begin{enumerate}[(i)]
		\item For every $a,b\in A$, there exists a $\ast$-homomorphism $\pi_{a,b}:\Mor^+(G,A)\rightarrow\mc{A}$ such that $\mc{A}\in Q$ and $\pi_{a,b}(p_{xa})=\pi_{a,b}(p_{yb})=1$.
		
		\item For every $\ast$-homomorphism $\pi:\Mor^+(G,A)\rightarrow\mc{A}$ such that $\mc{A}\in Q$, $[\pi(p_{xa}),\pi(p_{yb})]=0$ for all $a,b\in A$.
	\end{enumerate}
	
	An \emph{oracular $Q$-commutativity gadget} is a relational structure $G$ over $\sigma$ along with two distinguished variables $x,y\in G$ such that
	\begin{enumerate}[(i)]
		\item For every $a,b\in A$, there exists a $\ast$-homomorphism $\pi_{a,b}:\Mor^{o+}(G,A)\rightarrow\mc{A}$ such that $\mc{A}\in Q$ and $\pi_{a,b}(p_{xa})=\pi_{a,b}(p_{yb})=1$.
		
		\item For every $\ast$-homomorphism $\pi:\Mor^{o+}(G,A)\rightarrow\mc{A}$ such that $\mc{A}\in Q$, $[\pi(p_{xa}),\pi(p_{yb})]=0$ for all $a,b\in A$.
	\end{enumerate}
	
	An \emph{$a$-robust $Q$-commutativity gadget} is a $Q$-commutativity gadget $(G,x,y)$ such that for all $\varepsilon>0$, there exists $\delta>0$ such that if $\defect(\tau)<\delta$ for a tracial state $\tau$ on $\Mor^{a}_{\mbb{u}}(G,A)$ over $Q$, then $$\sum_{a,b\in A}\norm{[p_{xa},p_{yb}]}_\tau^2<\varepsilon.$$
	
	A \emph{$c$-$v$-robust $Q$-commutativity gadget} is an oracular $Q$-commutativity gadget $(G,x,y)$ such that for all $\varepsilon>0$, there exists $\delta>0$ such that if $\defect(\tau)<\delta$ for a tracial state $\tau$ on $\Mor^{c-v}_{\mbb{u}}(G,A)$ over $Q$, then $$\sum_{a,b\in A}\norm{[p_{xa},p_{yb}]}_\tau^2<\varepsilon.$$
	
	A \emph{$c$-$c$-robust $Q$-commutativity gadget} is an oracular $Q$-commutativity gadget $(G,x,y)$ such that for all $\varepsilon>0$, there exists $\delta>0$ such that if $\defect(\tau)<\delta$ for a tracial state $\tau$ on $\Mor^{c-c}_{\mbb{u}}(G,A)$ over $Q$, then $$\frac{1}{m^2}\sum_{\substack{R\in\sigma,\mathbf{x}\in R^G,i\in[\ar(R)]:\,x_i=x\\S\in\sigma,\mathbf{y}\in S^G,j\in[\ar(R)]:\,y_j=y\\a,b\in A}}\norm{[\Phi^{R,\mathbf{x},i}_a,\Phi^{S,\mathbf{y},j}_b]}_\tau^2<\varepsilon.$$
\end{definition}

Since a $c$-$c$-robust and $c$-$v$-robust commutativity gadgets are equivalent~\cite[Lemma 4.3]{CDdBVZ25}, we refer to these as \emph{oracular robust commutativity gadgets}.

We work with the following min-tensor monoids: 
\begin{itemize}
	\item Deterministic $d=\{\C\}$
	\item Classical $c=\set*{C(X)}{X\text{ compact topological space}}$
	\item Quantum $q=\set*{B(\C^d)}{d\in\N}$
	\item Quantum approximate $qa=\set*{\mc{A}}{\mc{A}\hookrightarrow\mc{R}^\omega}$
	\item Quantum commuting $qc=\set*{\mc{A}}{\exists\;\tau:\mc{A}\rightarrow\C\text{ tracial state}}$
	\item All $C^\ast=\set*{\mc{A}}{\mc{A}\text{ is a $C^\ast$-algebra}}$
\end{itemize}

Commutativity gadgets are trivial for $Q=d,c$. The commutativity gadgets considered in \cite{CDdBVZ25} are generally the $qa$-commutativity gadgets, except for the algebraic commutativity gadgets, which are $C^\ast$-commutativity gadgets.

It follows by \cite[Lemmas 4.5 and 4.6]{CDdBVZ25} that every $C^\ast$-commutativity gadget is also a robust $qa$-commutativity gadget. Finally, since the $C^\ast$-algebras in $qa$ are (subalgebras of) ultraproducts of finite-dimensional $C^\ast$-algebras, we have that robust $q$-commutativity gadgets are robust $qa$-commutativity gadgets.

}

\section{Commutativity gadgets and hardness}\label{sec:comm-gadget-hardness}

Since commutativity gadgets allow classical reductions between CSP languages to be lifted to the entangled setting, their main application (currently the only known application) is in showing undecidability of entangled constraint system languages. We present these in terms of the quantum relational structure homomorphisms; however, there is a nearly equivalent characterisation in terms of the values of the associated nonlocal games, via the discussion in \cref{sec:nlg}, the only difference being that the non-oracular version of the game exists only for $2$-CSPs, such as graph homomorphism CSPs.

\begin{definition}
    Let $A$ be a relational structure, let $Q$ be a min-tensor monoid, let $w\in\{a,c-v,c-c\}$, and let $0\leq s\leq c\leq 1$. The promise problem $\CSP^Q_{w}(A)_{t,s}$ is the problem of deciding if, given a relational structure $B$, there exists a tracial state $\tau$ over $Q$ on $\Mor_{\mbb{u}}^w(B,A)$ such that $\defect(\tau)\leq 1-c$ or for all such tracial states $\defect(\tau)>1-s$.

    The promise problem $\SuccinctCSP^Q_{w}(A)_{t,s}$ is the problem of deciding if, given a relational structure $B$ implicitly via a Turing machine sampling a probability distribution $\pi$, there exists a tracial state $\tau$ over $Q$ on $\Mor_{\pi}^w(B,A)$ such that $\defect(\tau)\leq 1-c$ or for all such tracial states $\defect(\tau)>1-s$.
\end{definition}

Note that if $c=s=1$, then $\CSP^Q_a(A)_{1,1}$ is the problem of deciding if the $Q$-quantum morphisms $\mor^{\overline{Q}}(B,A)\neq\varnothing$; and $\CSP^Q_{c-c}(A)_{1,1}=\CSP^Q_{c-v}(A)_{1,1}$ is the problem of deciding if $\mor^{o\overline{Q}}(B,A)\neq\varnothing$, where $\overline{Q}$ is the closure of $Q$ under ultraproducts. Note that in this case these problems are languages.

\begin{definition}
    $\tsf{RE}=\Sigma_1^0$ is the class of languages decidable by sentences of the form $\exists\,x.\;\phi(x)$ where $\phi$ is a computable predicate over bit strings. A complete problem for $\tsf{RE}$ is the \emph{halting problem}, which is the problem of deciding whether a given Turing machine will ever halt.

    $\tsf{coRE}=\Pi_1^0$ is the class of languages decidable by sentences of the form $\forall\,x.\;\phi(x)$ where $\phi$ is a computable predicate over bit strings. A complete problem for $\tsf{coRE}$ is the \emph{cohalting problem}, which is the problem of deciding whether a given Turing machine will run forever.
\end{definition}

\begin{theorem}
    Let $A$ be a relational structure such that $\CSP(A)$ is $\tsf{NP}$-complete.
    \begin{enumerate}[(i)]
        \item If $A$ admits a $q$-commutativity gadget, then there is a polynomial-time reduction from the halting problem to $\SuccinctCSP_{a}^q(A)_{1,1}$.
        \item If $A$ admits an oracular $q$-commutativity gadget, then there is a polynomial-time reduction from the halting problem to $\SuccinctCSP_{c-v}^q(A)_{1,1}$ and $\SuccinctCSP_{c-c}^q(A)_{1,1}$.
        \item If $A$ admits an $a$-robust $q$-commutativity gadget, then there exists $\varepsilon>0$ such that there is a polynomial-time reduction from the halting problem to $\SuccinctCSP_{a}^q(A)_{1,1-\varepsilon}$, making this $\tsf{RE}$-complete.
        \item If $A$ admits a $c$-$v$-robust $q$-commutativity gadget, then there exists $\varepsilon>0$ such that there is a polynomial-time reduction from the halting problem to $\SuccinctCSP_{c-v}^q(A)_{1,1-\varepsilon}$ and $\SuccinctCSP_{c-c}^q(A)_{1,1-\varepsilon}$, making these $\tsf{RE}$-complete.
        \item If $A$ admits a $qa$-commutativity gadget, then there is a polynomial-time reduction from the halting problem to $\SuccinctCSP_{a}^{qa}(A)_{1,1}$.
        \item If $A$ admits an oracular $qa$-commutativity gadget, then there is a polynomial-time reduction from the halting problem to $\SuccinctCSP_{c-v}^{qa}(A)_{1,1}$ and $\SuccinctCSP_{c-c}^{qa}(A)_{1,1}$.
        \item If $A$ admits an $a$-robust $qa$-commutativity gadget, then there exists $\varepsilon>0$ such that there is a polynomial-time reduction from the halting problem to $\SuccinctCSP_{a}^{qa}(A)_{1,1-\varepsilon}$, making this $\tsf{RE}$-complete.
        \item If $A$ admits a $c$-$v$-robust $qa$-commutativity gadget, then there exists $\varepsilon>0$ such that there is a polynomial-time reduction from the halting problem to $\SuccinctCSP_{c-v}^{qa}(A)_{1,1-\varepsilon}$ and $\SuccinctCSP_{c-c}^{qa}(A)_{1,1-\varepsilon}$, making these $\tsf{RE}$-complete.
        \item If $A$ admits a $qc$-commutativity gadget, then there is a polynomial-time reduction from the cohalting problem to $\SuccinctCSP_{a}^{qc}(A)_{1,1}$, making this $\tsf{coRE}$-complete.
        \item If $A$ admits an oracular $qc$-commutativity gadget, then there is a polynomial-time reduction from the cohalting problem to $\SuccinctCSP_{c-v}^{qc}(A)_{1,1}$ and $\SuccinctCSP_{c-c}^{qc}(A)_{1,1}$, making these $\tsf{coRE}$-complete.
        \item If $A$ admits an $a$-robust $qc$-commutativity gadget, then there exists $\varepsilon>0$ such that there is a polynomial-time reduction from the cohalting problem to $\SuccinctCSP_{a}^{qc}(A)_{1,1-\varepsilon}$, making this $\tsf{coRE}$-complete.
        \item If $A$ admits a $c$-$v$-robust $qc$-commutativity gadget, then there exists $\varepsilon>0$ such that there is a polynomial-time reduction from the cohalting problem to $\SuccinctCSP_{c-v}^{qc}(A)_{1,1-\varepsilon}$ and $\SuccinctCSP_{c-c}^{qc}(A)_{1,1-\varepsilon}$, making these $\tsf{coRE}$-complete.
    \end{enumerate}
\end{theorem}

The classes $\tsf{RE}$ and $\tsf{coRE}$ do not depend on runtime, just computability, so completeness (or hardness) of any succinct problem for this class also implies completeness (or hardness) of the non-succinct version.

\begin{proof}[Proof sketch]
    (vii) and (viii) are the result of \cite[Theorem 5.10]{CDdBVZ25} (and also implicit in \cite{CM24}). (iii) and (iv) follow in the same way, noting that in~\cite{JNV+21}, the yes instances have perfect finite-dimensional strategies. Since the reductions of \cite{MS24,CM24} preserve finite-dimensionality of representations of the quantum homomorphism algebra, we can still assume that the yes instances have perfect finite-dimensional representation. (i), (ii), (v), and (vi) follow immediately from (iii), (iv), (vii), and (viii), respectively. We are not able to immediately assert completeness here since generally deciding the quantum value exactly is $\Pi_2^0$-complete, strictly harder than $\tsf{RE}$, due to~\cite{MNY22}, but the games for the yes instances do not have perfect finite-dimensional strategies; further, it is not clear that the strong answer reduction of~\cite{DFN+24} also applies in the zero-gap case. However, this answer reduction does apply in the commuting-operator model, so using the gapped commuting-operator hardness of~\cite{Lin25}, we can apply \cite[Theorem 5.10]{CDdBVZ25} to get (xi) and (xii). (ix) and (x) follow immediately, but we do in fact get completeness here since these problems are contained in \tsf{coRE}~\cite{Slo20,MNY22}.
\end{proof}

\section{Characterisation of commutativity gadgets by polymorphisms}\label{sec:gadget-polymorphisms}

In recent work, Ciardo, Joubert, and Mottet~\cite{CJM25} prove an equivalence between the existence of $q$-commutativity gadgets and commutativity of quantum polymorphisms, the elements of $\bigcup_{k=1}^\infty\mor^q(A^k,A)$. We extend this characterisation to all min-tensor monoids. The proofs work in exactly the same way, but we include them here, as they inform the proofs in the robust cases.

\begin{theorem}[\cite{CJM25}]\label{thm:basic-char}
	Let $A$ be a relational structure and let $Q$ be a min-tensor monoid. The following are equivalent.
	\begin{enumerate}[(i)]
		\item $\mor^{Q}(A^k,A)\subseteq\mor^c(A^k,A)$ for all $k\in\N$.
		\item $\mor^{Q}(A^{|A|^2},A)\subseteq\mor^c(A^{|A|^2},A)$.
		\item $A$ admits a $Q$-commutativity gadget.
	\end{enumerate}
\end{theorem}

It is clear that (i) $\Rightarrow$ (ii). Next, we show that (iii) $\Rightarrow$ (i) in contrapositive.

\begin{lemma}\label{lem:basic-nogo}
	Suppose that there exists $k\in\N$ such that $\mor^{Q}(A^k,A)\not\subseteq\mor^c(A^k,A)$. Then $A$ does not admit a $Q$-commutativity gadget.
\end{lemma}

\begin{proof}
	Let $\pi$ be a non-classical $\ast$-representation of $\Mor^+(A^k,A)$ over $Q$. Then, there exist $\mathbf{a},\mathbf{b}\in A^k$ and $a,b\in A$ such that $\pi(p_{\mathbf{a}a})$ and $\pi(p_{\mathbf{b}b})$ do not commute. Now, suppose $(G,x,y)$ is a $Q$-commutativity gadget for $A$. By definition of a $Q$-commutativity gadget, for each $a,b\in A$, there exists a $\ast$-representation $\pi_{a,b}:\Mor^+(G,A)\rightarrow\mc{A}_{a,b}$ such that $\mc{A}_{a,b}\in Q$ and $\pi_{a,b}(p_{xa})=\pi_{a,b}(p_{yb})=1$.  Now consider the map $\varphi:\Mor^+(G,A^k)\rightarrow\mc{A}_{a_1,b_1}\otimes\cdots\otimes\mc{A}_{a_k,b_k}$ defined on the generators as $\varphi(p_{z\mathbf{c}})=\pi_{a_1,b_1}(p_{zc_1})\otimes\cdots\otimes\pi_{a_k,b_k}(p_{zc_k})$. I claim this extends to a $\ast$-representation of $\Mor^+(G,A^k)$ and therefore is an element of $\mor^{Q}(G,A^k)$. To see this, simply note that $\varphi$ satisfies the defining relations of $\Mor^+(G,A^k)$. In fact, the $\varphi(p_{z\mathbf{c}})$ are projections as they are tensor products of projections. Next,
	\begin{align*}
		\sum_{\mathbf{c}\in A^k}\varphi(p_{z\mathbf{c}})=\sum_{c_1\in A}\pi_{a_1,b_1}(p_{zc_1})\otimes\cdots\otimes\sum_{c_k\in A}\pi_{a_k,b_k}(p_{zc_k})=1.
	\end{align*}
	Finally, suppose $R\in\sigma$, $\mathbf{z}\in R^G$, and $\mathbbm{c}\notin R^{A^k}$; this implies that $(c_{1i},\ldots,c_{ni})\notin R^A$ for some $i$. Writing $n=\ar(R)$, we have that
	\begin{align*}
		\varphi(p_{z_1\mathbf{c}_1})\cdots\varphi(p_{z_n\mathbf{c}_n})&=\bigotimes_{i=1}^k\pi_{a_i,b_i}(p_{z_1c_{1i}}\cdots p_{z_n c_{ni}})=0.
	\end{align*}
	Thus, $\varphi\in\mor^{Q}(G,A^k)$. Note also that by construction $\varphi(p_{x\mathbf{a}})=\varphi(p_{y\mathbf{b}})=1$. Composing $\varphi$ with $\pi$ as $Q$-morphisms, we have that $\pi\circ\varphi\in\mor^{Q}(G,A)$. But we have that
	\begin{align*}
		&(\pi\circ\varphi)(p_{xa})=\sum_{\mathbf{c}\in A^k}\varphi(p_{x\mathbf{c}})\otimes\pi(p_{\mathbf{c}a})=1\otimes\pi(p_{\mathbf{a}a})\\
		&(\pi\circ\varphi)(p_{yb})=\sum_{\mathbf{c}\in A^k}\varphi(p_{y\mathbf{c}})\otimes\pi(p_{\mathbf{c}b})=1\otimes\pi(p_{\mathbf{b}a}),
	\end{align*}
	which do not commute by hypothesis. Contradiction, so $A$ does not admit a commutativity gadget.
\end{proof}

We finish the proof of \cref{thm:basic-char} by showing (ii) $\Rightarrow$ (iii), which relies on \cref{lem:standard}.

\begin{lemma}\label{lem:basic-converse}
	Suppose $A$ is a relational structure such that $\mor^{Q}(A^{|A|^2},A)\subseteq\mor^c(A^{|A|^2},A)$. Then, $A$ admits a $Q$-commutativity gadget.
\end{lemma}

\begin{proof}
	Let $k=|A^2|=|A|^2$. By \cref{lem:standard}, $A$ pp-defines the relation $A^2$, as the set of polymorphisms that preserves $A^2$, which are all maps, contains $\pol(A)$. Now, we claim $(A^k,\mathbf{x},\mathbf{y})$ is a commutativity gadget for $A$. To show property (i), let $a,b\in A$. By the construction of \cref{lem:standard}, there exists a polymorphism $f:A^k\rightarrow A$ such that $f(\mathbf{x})=a$ and $f(\mathbf{y})=b$. So, we can take $\pi_{a,b}(p_{\mathbf{a}c})=\delta_{c,f(\mathbf{a})}$. Property (ii) follows directly from the hypothesis that $\mor^{qa}(A^k,A)=\mor^c(A^k,A)$.
\end{proof}

This proof extends identically to the oracular case.

\begin{theorem}[\cite{CJM25}]\label{thm:orac-basic-char}
	Let $A$ be a relational structure and let $Q$ be a min-tensor monoid. The following are equivalent.
	\begin{enumerate}
		\item $\mor^{oQ}(A^k,A)\subseteq\mor^c(A^k,A)$ for all $k\in\N$.
		\item $\mor^{oQ}(A^{|A|^2},A)\subseteq\mor^c(A^{|A|^2},A)$.
		\item $A$ admits an oracular $Q$-commutativity gadget.
	\end{enumerate}
\end{theorem}

\begin{proof}\phantom{}
	
	\fbox{(i)$\Rightarrow$(ii)} Immediate.
	
	\fbox{(ii)$\Rightarrow$ (iii)} Suppose $\mor^{oQ}(A^{|A|^2},A)=\mor^c(A^{|A|^2},A)$. Following the notation of \cref{lem:basic-converse}, we claim $(A^{|A|^2},\mathbf{x},\mathbf{y})$ is an oracular commutativity gadget for $A$. In fact, property (i) holds identically; and property (ii) holds since every oracular $Q$-morphism is classical.
	
	\fbox{(iii)$\Rightarrow$ (i)} To extend the proof of \cref{lem:basic-nogo}, it suffices to show that the constructed $\varphi$ also satisfies the additional commutativity relations of $\Mor^{o+}(G,A^k)$. Let $\sigma\in R$, $\mathbf{z}\in R^G$, $i\neq j\in[\ar(R)]$, and $\mathbf{c},\mathbf{d}\in A^k$. Then,
	\begin{align*}
		\varphi(p_{z_i\mathbf{c}})\varphi(p_{z_j\mathbf{d}})&=\pi_{a_1,b_1}(p_{z_ic_1})\pi_{a_1,b_1}(p_{z_jd_1})\otimes\cdots\otimes\pi_{a_n,b_n}(p_{z_ic_n})\pi_{a_n,b_n}(p_{z_jd_n})\\
		&=\pi_{a_1,b_1}(p_{z_jd_1})\pi_{a_1,b_1}(p_{z_ic_1})\otimes\cdots\otimes\pi_{a_n,b_n}(p_{z_jd_n})\pi_{a_n,b_n}(p_{z_ic_n})\\
		&=\varphi(p_{z_j\mathbf{d}})\varphi(p_{z_i\mathbf{c}}),
	\end{align*}
	as wanted. The rest of the proof proceeds identically.
\end{proof}

\section{Stable commutation}\label{sec:stability}

To extend the characterisation of commutativity gadgets in terms of polymorphisms to the robust setting, we make use of an approximate notion of commutativity.

\begin{definition}
	Let $(\mc{A},\mu)$ be a weighted algebra, and let $Q$ be a min-tensor monoid. We say $\mc{A}$ is \emph{$Q$-stably commutative} with respect to a finite generating set $X\subseteq\mc{A}$ if for all $\varepsilon>0$, there exists $\delta>0$ such that for any tracial state $\tau:\mc{A}\rightarrow\C$ over $Q$, $\defect(\tau)<\delta$ implies that $\norm{[x,y]}_\tau^2<\varepsilon$ for all $x,y\in X$.
\end{definition}

It is clear that if $\mc{A}$ is $Q$-stably commutative, then $\mc{A}/\angnormal{\supp\mu}$ is commutative with respect to any representation to a tracial $C^\ast$-algebra in $Q$. Also, for $Q\subseteq Q'$, if $\mc{A}$ is $Q'$-stably commutative, then $\mc{A}$ is $Q$-stably commutative.

Next we show that stable commutativity is independent of the choice of finite generating set.

\begin{lemma}\label{lem:generator-independence}
	Let $(\mc{A},\mu)$ be a weighted algebra that is archimedean with respect to the $\ast$-positive cone of sums of squares, and let $X$ and $Y$ be finite generating sets of $\mc{A}$. Then $\mc{A}$ is $Q$-stably commutative with respect to $X$ if and only if it is $Q$-stably commutative with respect to $Y$.
\end{lemma}

\begin{proof}
	By symmetry, we need only show one direction. Suppose $\mc{A}$ is $Q$-stably commutative with respect to $X$. Since $\mc{A}$ is archimedean, for each $x\in\mc{A}$, there exists $\norm{x}\in\R_{\geq 0}$ such that $x^\ast x,xx^\ast\leq\norm{x}^2$; choosing $\norm{x}$ minimal gives rise to a submultiplicative seminorm~\cite{Oza13}. Write $X=\{x_1,\ldots,x_n\}$ and $Y=\{y_1,\ldots,y_m\}$. As $Y\subseteq\mc{A}$, there exist noncommutative polynomials $p_i$ for each $i\in[m]$ such that
	$$y_i=p_i(x_1,\ldots,x_n)=\sum_{k=1}^{K_i}c_{k,i}x_{j_{k,i,1}}\cdots x_{j_{k,i,N_{k,i}}}.$$
	Let $$M=\max_{i,i'\in[m]}4K_iK_{i'}\sum_{k=1}^{K_i}\sum_{k'=1}^{K_{i'}}|c_{k,i}|^2|c_{k',i'}|^2N_{k,i}N_{k',i'}\sum_{\alpha=1}^{N_{k,i}}\sum_{\alpha'=1}^{N_{k',i'}}\prod_{\beta\neq\alpha}\norm{x_{j_{k,i,\beta}}}^2\prod_{\beta'\neq\alpha'}\norm{x_{j_{k',i',\beta'}}}^2$$
	Let $\varepsilon>0$. Then, as $\mc{A}$ is $Q$-stably commutative with respect to $X$, then there exists $\delta>0$ such that for any tracial state $\tau$ on $\mc{A}$ over $Q$ and $\defect(\tau)<\delta$, the norms $\norm{[x_j,x_{j'}]}_\tau^2<\frac{\varepsilon}{M}$. Now, let $\tau$ be a tracial state on $\mc{A}$ over $Q$ such that $\defect(\tau)<\delta$. First, for any $a,b,c\in\mc{A}$, the commutator
	\begin{align*}
		\norm{[ab,c]}_\tau^2&=\norm{abc-acb+acb-cab}_\tau^2=\norm{a[b,c]+[a,c]b}_\tau^2\\
		&\leq 2\norm{a}^2\norm{[b,c]}_\tau^2+2\norm{b}^2\norm{[a,c]}_\tau^2.
	\end{align*}
	Hence, by induction, for any $a_1,\ldots,a_N,b\in\mc{A}$, we have that
	\begin{align*}
		\norm{[a_1 \cdots a_N,b]}_\tau^2&\leq 2\norm{a_1}^2\cdots\norm{a_{\floor{N/2}}}^2\norm{[a_{\floor{N/2}+1}\cdots a_N,b]}_\tau^2+2\norm{a_{\floor{N/2}+1}}^2\cdots\norm{a_{N}}^2\norm{[a_{1}\cdots a_{\floor{N/2}},b]}_\tau^2\\
		&\leq 2N\sum_{i=1}^N\prod_{j\neq i}\norm{a_j}^2\cdot\norm{[a_i,b]}_\tau^2
	\end{align*}
	Hence, we can bound the commutator
	\begin{align*}
	\hspace*{-9mm}	\norm{[y_i,y_{i'}]}_\tau^2&\leq K_iK_{i'}\sum_{k=1}^{K_i}\sum_{k'=1}^{K_{i'}}|c_{k,i}|^2|c_{k',i'}|^2\norm{[x_{j_{k,i,1}}\cdots x_{j_{k,i,N_{k,i}}},x_{j_{k',i',1}}\cdots x_{j_{k',i',N_{k',i'}}}]}_\tau^2\\
		&\leq K_iK_{i'}\sum_{k=1}^{K_i}\sum_{k'=1}^{K_{i'}}|c_{k,i}|^2|c_{k',i'}|^24N_{k,i}N_{k',i'}\sum_{\alpha=1}^{N_{k,i}}\sum_{\alpha'=1}^{N_{k',i'}}\prod_{\beta\neq\alpha}\norm{x_{j_{k,i,\beta}}}^2\prod_{\beta'\neq\alpha'}\norm{x_{j_{k',i',\beta'}}}^2\cdot\norm{[x_{j_{k,i,\alpha}},x_{j_{k',i',\alpha'}}]}_\tau^2\\
		&<K_iK_{i'}\sum_{k=1}^{K_i}\sum_{k'=1}^{K_{i'}}|c_{k,i}|^2|c_{k',i'}|^24N_{k,i}N_{k',i'}\sum_{\alpha=1}^{N_{k,i}}\sum_{\alpha'=1}^{N_{k',i'}}\prod_{\beta\neq\alpha}\norm{x_{j_{k,i,\beta}}}^2\prod_{\beta'\neq\alpha'}\norm{x_{j_{k',i',\beta'}}}^2\cdot\frac{\varepsilon}{M}\\
		&\leq M\frac{\varepsilon}{M}=\varepsilon.
	\end{align*}
	As such, $\mc{A}$ is $Q$-stably commutative with respect to $Y$.
\end{proof}

Next, we show that stable commutativity is a weaker condition than commutativity of the quotient algebra.

\begin{lemma}\label{lem:quotient-commutativity}
	Suppose $(\mc{A},\mu)$ is a weighted algebra with finite generating set $X$, that is archimedean with respect to the $\ast$-positive cone of sums of squares. If $C_u^\ast(\mc{A}/\angnormal{\supp\mu})$ is commutative, then $\mc{A}$ is $Q$-stably commutative for any min-tensor monoid $Q$.
\end{lemma}

\begin{proof}
	Since $C_u^\ast(\mc{A}/\angnormal{\supp\mu})$ is commutative, $[x,y]=0$ in $C_u^\ast(\mc{A}/\angnormal{\supp\mu})$ for all $x,y\in X$. This implies that $[x,y]$ is in the infinitesimal ideal of $\mc{A}/\angnormal{\supp\mu}$. Let $\varepsilon>0$. By the above, we know $\abs*{[x,y]}^2\leq\frac{\varepsilon}{2}$ in $\mc{A}/\angnormal{\supp\mu}$. By definition, there exists $p_{x,y}\in\mc{A}/\angnormal{\supp\mu}$ positive such that $\frac{\varepsilon}{2}-\abs*{[x,y]}^2=p_{x,y}$. In $\mc{A}$, this gives that $\abs{[x,y]}^2-\frac{\varepsilon}{2}+p_{x,y}\in\angnormal{\supp\mu}$, so there exist $N_{x,y}\in\N$, $r_{x,y,i}\in\supp\mu$, and $a_{x,y,i},b_{x,y,i}\in\mc{A}$ such that
	\begin{align*}
		\abs{[x,y]}^2-\frac{\varepsilon}{2}+p_{x,y}=\sum_{i=1}^{N_{x,y}}a_{x,y,i}r_{x,y,i}b_{x,y,i}.
	\end{align*}
	Now, let $M=\max_{x,y}\sum_{i=1}^{N_{x,y}}\frac{\norm{b_{x,y,i}a_{x,y,i}}}{\sqrt{\mu(r_{x,y,i})}}$ and set $\delta=\parens*{\frac{\varepsilon}{2M}}^2$.
	
	Suppose $\tau$ is a tracial state on $\mc{A}$ over $Q$ such that $\defect(\tau)<\delta$. Then, for all $x,y\in X$,
	\begin{align*}
		\norm{[x,y]}_\tau^2&=\tau(\abs*{[x,y]}^2)=\tau\parens[\Big]{\frac{\varepsilon}{2}-p_{x,y}+\sum_{i=1}^{N_{x,y}}a_{x,y,i}r_{x,y,i}b_{x,y,i}}\\
		&=\frac{\varepsilon}{2}-\tau(p_{x,y})+\sum_{i=1}^{N_{x,y}}\tau(b_{x,y,i}a_{x,y,i}r_{x,y,i})\\
		&\leq\frac{\varepsilon}{2}+\sum_{i=1}^{N_{x,y}}\norm{b_{x,y,i}a_{x,y,i}}\norm{r_{x,y,i}}_\tau\\
		&\leq\frac{\varepsilon}{2}+\sum_{i=1}^{N_{x,y}}\norm{b_{x,y,i}a_{x,y,i}}\sqrt{\frac{\defect(\tau)}{\mu(r_{x,y,i})}}\\
		&<\frac{\varepsilon}{2}+M\sqrt{\delta}=\varepsilon.\qedhere
	\end{align*}
\end{proof}

To finish this section, we find some equivalences for stable commutativity between different min-tensor monoids. Here, we follow the notation of \cite{CLP15} for ultrafilters.

\begin{lemma}\label{lem:fd-rounding}
	Let $\mc{M}$ be a finite-dimensional von Neumann algebra with a faithful tracial state $\tr$. Then,
	\begin{enumerate}[(i)]
		\item For any $X\in\mc{M}$, there exists a unitary $U$ such that $\norm{X-U}_{\tr}\leq\norm{X^\ast X-1}_{\tr}$.
		\item For any unitary $U\in\mc{M}$, there exists an order-$d$ unitary $W$ such that $\norm{U-W}_{\tr}\leq\norm{U^d-1}_{\tr}$.
	\end{enumerate} 
\end{lemma}

\begin{proof}
	\begin{enumerate}[(i)]
		\item Let $X=V\Sigma W$ be the singular-value decomposition, and set $U=VW$. Then, $\norm{X-U}_{\tr}=\norm{\Sigma-1}_{\tr}$ by unitary invariance. For any $x\in\R_{\geq 0}$, we have that $(x^2-1)^2=(x+1)^2(x-1)^2\geq(x-1)^2$. As such, noting that there are $p_i>0$ such that $\tr(A)=\sum_ip_iA_{ii}$, we get that
		\begin{align*}
			\norm{\Sigma-1}_{\tr}^2=\sum_ip_{i}(\Sigma_{ii}-1)^2\leq\sum_ip_{i}(\Sigma_{ii}^2-1)^2=\norm{\Sigma^2-1}_{\tr}^2.
		\end{align*}
		Using unitary invariance again gives $\norm{X-U}_{\tr}^2\leq\norm{X^\ast X-1}_{\tr}^2$.
		
		\item Let $U=VDV^\ast$ be the spectral decomposition of $U$. We can write $D_{ii}=\omega^{a_i}e^{i\theta_i}$ where $\omega$ is a primitive $d$-th root of unity, $a_i\in\{0,\ldots,d-1\}$, and $\theta_i\in(-\pi/d,\pi/d]$. Let $\tilde{D}_{ij}=\omega^{a_i}\delta_{ij}$ and $W=V\tilde{D}V^\ast$. Also, note that $\abs{e^{id\theta_i}-1}=\abs{e^{i\theta_i}-1}\abs{1+e^{i\theta_i}+\ldots+e^{i(d-1)\theta_i}}\geq\abs{e^{i\theta_i}-1}$, since $\cos(k\theta_i),\sin(k\theta_i)\geq 0$ for $\theta_i\in(-\pi/d,\pi/d]$. Therefore,
		\begin{align*}
			\norm{U-W}_{\tr}^2=\norm{D-\tilde{D}}_{\tr}^2=\sum_ip_i\abs{e^{i\theta_i}-1}^2\leq\sum_ip_i\abs{e^{id\theta_i}-1}^2=\norm{D^d-1}_{\tr}=\norm{U^d-1}_{\tr}.
		\end{align*}
	\end{enumerate}
\end{proof}

\begin{lemma}\label{lem:q-to-qa}
	Let $(\mc{A},\mu)$ be a weighted algebra where $\mc{A}=\C\Z_{d_1}\ast\cdots\ast\C\Z_{d_k}$. If $\mc{A}$ is $q$-stably commutative, then $\mc{A}$ is $qa$-stably commutative.
\end{lemma}

\begin{proof}
	Let $\varepsilon>0$ and fix the generating set $X=\{u_1,\ldots,u_k\}$ of $\mc{A}$ such that $u_i^{d_i}=u_i^\ast u_i=1$ with no other relations. Since $\mc{A}$ is $q$-stably commutative, there exists $\delta>0$ such that if $\defect(\tau)<\delta$ for a tracial state $\tau:\mc{A}\rightarrow\C$ over $q$, then $\norm{[x,y]}_\tau^2<\varepsilon$ for all $x,y\in X$. Now, suppose $\tau$ is a tracial state on $\mc{A}$ over $qa$ such that $\defect(\tau)<\frac{\delta}{2}$. Now consider the GNS representation of $\tau$; we may assume that it is minimal, giving a $\ast$-representation $\pi:\mc{A}\rightarrow\mc{M}$, where $(\mc{M},\tr)$ is a finite von Neumann algebra that embeds into $(\mc{R}^\omega,\tau_\omega)$ and $\tau=\tr\circ\pi$. In particular $\tr$ is faithful. Further, by \cref{lem:ultrapower}, $\mc{R}^\omega$ embeds into $(\prod_{\omega'}\mc{M}_n,\tau_{\omega'})$, where $\omega'$ is some free ultrafilter and $\mc{M}_n$ are finite-dimensional von Neumann algebras. By composing $\pi$ with the embedding $\mc{M}\hookrightarrow\prod_{\omega'}\mc{M}_n$, we get a $\ast$-representation $\pi':\mc{A}\rightarrow\prod_{\omega'}\mc{M}_n$ such that $\tau=\tau_{\omega'}\circ\pi'$. This gives that $\pi'$ can be realised as a sequence of maps $(\pi_n)$ such that for all $a,b\in\mc{A}$ $\lim_{n\rightarrow\omega'}\norm{\pi_n(a)\pi_n(b)-\pi_n(ab)}_{\tr_n}^2=0$ and $\lim_{n\rightarrow\omega'}\norm{\pi_n(a^\ast)-\pi_n(a)^\ast}_{\tr_n}^2=0$. Now, we need to round the $\pi_n$ to $\ast$-representations of $\mc{A}$. Let $U_{n,i}$ be the nearest order-$d_i$ unitary to $\pi_n(u_i)$. By \cref{lem:fd-rounding}, we know that $\norm{\pi_n(u_i)-U_{n,i}}_{\tr_n}\overset{\omega'}{\rightarrow}0$. Then, we define the $\ast$-representation $\pi'_n:\mc{A}\rightarrow\mc{M}_n$ by $\pi_n'(u_i)=U_{n,i}$. We have that $\norm{\pi_n(a)-\pi_n'(a)}_{\tr_n}\overset{\omega'}{\rightarrow}0$ for all $a\in\mc{A}$, so the sequence $(\pi'_n)$ also induces the representation $\pi'$ on $\mc{R}^\omega$. Since $\defect(\tau)<\frac{\delta}{2}$, we know that $\defect(\tr_n\circ\pi_n)\overset{\omega}{\rightarrow}\defect(\tau)<\frac{\delta}{2}$. Let $N=\set*{n\in\N}{\defect(\tr_n\circ\pi_n)<\delta}$. By convergence, $N\in\omega'$ and by $q$-stable commutativity, for all $n\in N$, $\norm{[x,y]}_{\tr_n\circ\pi_n}^2<\varepsilon$ for all $x,y\in X$. Now, suppose there exist $x,y\in X$ such that $\norm{[x,y]}_\tau^2\geq\varepsilon$. Then, there exists $M\in\omega'$ such that for all $n\in M$, $\norm{[x,y]}_{\tr_n\circ\pi_n}^2\geq\varepsilon$. However, we get that $M$ and $N$ are disjoint, so $\varnothing=M\cap N\in\omega'$. Contradiction, so $\norm{[x,y]}_\tau^2<\varepsilon$ for all $x,y\in X$, giving that $\mc{A}$ is $qa$-stably commutative.
\end{proof}

\begin{lemma}\label{lem:ungapping}
	Let $(\mc{A},\mu)$ be a weighted algebra. If every representation of $C_u^\ast(\mc{A}/\angnormal{\supp\mu})$ into $\mc{R}^\omega$ (for some ultrafilter $\omega$) is commutative, then $\mc{A}$ is $q$-stably commutative.
\end{lemma}

\begin{proof}
	Fix a finite generating set $X$ of $\mc{A}$, and for a tracial state $\tau$ write $$\mathrm{comm\hy def}(\tau)=\sum_{x,y\in X}\norm{[x,y]}_\tau^2.$$ We proceed by contrapositive. Suppose $\mc{A}$ is not $q$-stably commutative. Then, there exists $\varepsilon>0$ such that for all $\delta>0$, there exists a finite-dimensional tracial state $\tau$ on $\mc{A}$ such that $\defect(\tau)<\delta$ but $\mathrm{comm\hy def}(\tau)\geq\varepsilon$. In particular, for each $n\in\N$, there exists a finite-dimensional tracial state $\tau_n$ such that $\defect(\tau_n)<\frac{1}{2^n}$ but but $\mathrm{comm\hy def}(\tau_n)\geq\varepsilon$. Now, consider the GNS representation $\pi_n:\mc{A}\rightarrow\mc{M}_n$, where $\mc{M}_n$ is a finite-dimensional von Neumann algebra with tracial state $\tr_n$; we may without loss of generality assume that $\dim\mc{M}_n\rightarrow\infty$ in $n$. Fix a non-principal ultrafilter $\omega$, and let $\mc{M}=\prod_\omega\mc{M}_n$ the ultraproduct. $\mc{M}$ has a tracial state $\tr$ defined as $\tr(x)=\lim_{n\rightarrow\omega}\tr_n(x_n)$, and embeds in $\mc{R}^{\omega}$ by \cref{lem:ultrapower}. Let $\pi:\mc{A}\rightarrow\mc{M}$ be defined as $\pi(a)=(\pi_n(a))$. This is a $\ast$-homomorphism and induces a tracial state $\tau=\tr\circ\pi$ on $\mc{A}$. We also have that
	\begin{align*}
		\defect(\tau)&=\sum_{a\in\mc{A}}\mu(a)\tau(a^\ast a)=\sum_{a\in\mc{A}}\mu(a)\lim_{n\rightarrow\omega}\tr_n(\pi_n(a^\ast a))=\lim_{n\rightarrow\omega}\sum_{a\in\mc{A}}\mu(a)\norm{a}_{\tau_n}^2=\lim_{n\rightarrow\omega}\frac{1}{2^n}=0.
	\end{align*}
	As such, $\pi$ factors through $\mc{A}/\angnormal{\supp\mu}$ to give a representation $\bar{\pi}:\mc{A}/\angnormal{\supp\mu}\rightarrow\mc{M}$. By universal property, this induces a representation $\bar{\bar{\pi}}:C_u^\ast(\mc{A}/\angnormal{\supp\mu})\rightarrow\mc{M}$. To conclude, it remains to see that $\im\bar{\bar{\pi}}=\overline{\im\pi}^{\norm{\cdot}}$ is not commutative. Let $a=\sum_{x,y\in X}\abs*{[x,y]}^2$. Since $\tr$ is faithful, if $\tau(a)=\mathrm{comm\hy def}(\tau)\neq 0$, then $\im\bar{\bar{\pi}}$ is not commutative. We have that
	\begin{align*}
		\tau(a)=\lim_{n\rightarrow\omega}\tau_n(a)=\lim_{n\rightarrow\omega}\mathrm{comm\hy def}(\tau_n).
	\end{align*}
	However, $\mathrm{comm\hy def}(\tau_n)\geq\varepsilon$ for all $n$, and therefore $\set*{n\in\N}{\abs*{\mathrm{comm\hy def}(\tau_n)}<\varepsilon}=\varnothing\notin\omega$. Thus, $(\mathrm{comm\hy def}(\tau_n))$ does not tend to $0$, so $\bar{\bar{\pi}}$ is a noncommutative representation of the $C^\ast$-algebra $C_u^\ast(\mc{A}/\angnormal{\supp\mu})$ into $\mc{M}\subseteq\mc{R}^{\omega}$.
\end{proof}
 
\begin{lemma}\label{lem:qc-ungapping}
	Let $(\mc{A},\mu)$ be a weighted algebra. If every representation of $C_u^\ast(\mc{A}/\angnormal{\supp\mu})$ into $\mc{B}\in qc$ is commutative, then $\mc{A}$ is $qc$-stably commutative.
\end{lemma}

The proof follows in the same way as \cref{lem:ungapping}, by noting that $qc$ is closed under ultraproducts.

\section{Robust commutativity gadgets}
\label{sec:robust-gadgets}

Using the stability notions of \cref{sec:stability}, we can fully characterise when robust commutativity gadgets exist in terms of properties of the quantum polymorphisms.

\begin{theorem}\label{thm:a-robust-char}
	Let $A$ be a relational structure and let $Q\subseteq qc$ be a min-tensor monoid. The following are equivalent.
	\begin{enumerate}[(i)]
		\item $\Mor^a_{\mbb{u}}(A^k,A)$ is $Q$-stably commutative for all $k\in\N$.
		\item $\Mor^a_{\mbb{u}}(A^{|A|^2},A)$ is $Q$-stably commutative.
		\item $A$ admits an $a$-robust $Q$-commutativity gadget.
	\end{enumerate}
\end{theorem}

\begin{proof}\phantom{}
	
	\fbox{(i)$\Rightarrow$(ii)} Immediate.
	
	\fbox{(ii)$\Rightarrow$(iii)} Suppose $\Mor^a_{\mbb{u}}(A^{|A|^2},A)$ is $Q$-stably commutative. Following the notation of \cref{lem:basic-converse}, we claim $(A^{|A|^2},\mathbf{x},\mathbf{y})$ is an $a$-robust $Q$-commutativity gadget for $A$. Since $\Mor^a_{\mbb{u}}(A^{|A|^2},A)$ is $Q$-stably commutative, we know that every $\ast$-representation of the quantum homomorphism algebra $\Mor^+(A^{|A|^2},A)=C^\ast_u(\Mor^a_{\mbb{u}}(A^{|A|^2},A)/\angnormal{\supp\mu_{a,\mbb{u}}})$ over $Q$ is commutative, and hence $\mor^{Q}(A^{|A|^2},A)\subseteq\mor^c(A^{|A|^2},A)$. By \cref{lem:basic-converse}, $(A^{|A|^2},\mathbf{x},\mathbf{y})$ is a commutativity gadget for $A$. For the additional $a$-robustness property, it follows from the stable commutativity property that for any $\varepsilon>0$, there exists $\delta>0$ such that for any tracial state $\tau$ on $\Mor^a_{\mbb{u}}(A^{|A|^2},A)$ over $Q$, if $\defect(\tau)<\delta$, then $\norm{[p^{\mathbf{a}}_a,p^{\mathbf{b}}_b]}_\tau^2\leq\frac{\varepsilon}{|A|^2}$. So, this implies that $\sum_{a,b\in A}\norm{[p^{\mathbf{x}}_a,p^{\mathbf{y}}_b]}_\tau^2<\varepsilon$, and hence $(A^{|A|^2},\mathbf{x},\mathbf{y})$ is an $a$-robust commutativity gadget.
	
	\fbox{(iii)$\Rightarrow$(i)} As in \cref{lem:basic-nogo}, we show the contrapositive by contradiction. Suppose $(G,x,y)$ is an $a$-robust commutativity gadget. There exists $k\in\N$ such that $\Mor^a_{\mbb{u}}(A^k,A)$ is not $Q$-stably commutative. Write $m=\sum_{R\in\sigma}|R^G|$ and $M=\sum_{R\in\sigma}\abs{R^{A^k}}$. Then, there exists $\varepsilon>0$ such that for any $\delta>0$ there is a tracial state $\tau$ on $\Mor^a_{\mbb{u}}(A^k,A)$ over $Q$ such that $\defect(\tau)<\frac{\delta}{M^2}$ but $\norm{[p^{\mathbf{a}}_a,p^{\mathbf{b}}_b]}_\tau^2\geq\varepsilon$ for some $\mathbf{a},\mathbf{b}\in A^k$ and $a,b\in A$. Fix $\delta>0$, and the corresponding $\tau,\mathbf{a},\mathbf{b},a,b$. Let $\pi:\Mor^a_{\mbb{u}}(A^k,A)\rightarrow\mc{M}$ be the GNS representation of $\tau$, with a tracial state $\rho$ on $\mc{M}$ such that $\tau=\rho\circ\pi$. Let $\varphi$ be the representation of $\Mor^+(G,A^k)$ defined in \cref{lem:basic-nogo}. Write $\nu$ for the tracial state on $\im\varphi$. Now, define the $\ast$-representation $\pi'$ of $\Mor^{a}_{\mbb{u}}(G,A)$ on the generators by
	\begin{align*}
		\pi'(p^z_c)=\sum_{\mathbf{c}\in A^k}\varphi(p_{z\mathbf{c}})\otimes\pi(p^{\mathbf{c}}_c).
	\end{align*}
	This definition clearly extends to a $\ast$-representation as the $\pi'(p^z_c)$ are projections and
	$$\sum_{c\in A}\pi'(p^z_c)=\sum_{c\in A,\mathbf{c}\in A^k}\varphi(p_{z\mathbf{c}})\otimes\pi(p^{\mathbf{c}}_c)=1.$$
	Note also that $(\nu\otimes\rho)$ is a tracial state on $\im\pi'$, which induces a tracial state $\tau'=(\nu\otimes\rho)\circ\pi'$ on $\Mor^a_{\mbb{u}}(G,A)$. The defect of this state is upper-bounded as
	\begin{align*}
		\defect(\tau')&=\frac{1}{m}\sum_{R\in\sigma,\mathbf{z}\in R^G,\mathbf{c}\notin R^A}\norm{p^{z_1}_{c_1}\cdots p^{z_{\ar(R)}}_{c_{\ar(R)}}}_{\tau'}^2\\
		&=\frac{1}{m}\sum_{R\in\sigma,\mathbf{z}\in R^G,\mathbf{c}\notin R^A}\norm[\Big]{\sum_{\mathbf{c}_1,\ldots,\mathbf{c}_k\in A^k}\varphi(p_{z_1\mathbf{c}_1}\cdots p_{z_{\ar(R)}\mathbf{c}_{\ar(R)}})\otimes\pi(p^{\mathbf{c}_1}_{c_1}\cdots p^{\mathbf{c}_{\ar(R)}}_{c_{\ar(R)}})}_{\nu\otimes\rho}^2\\
		&=\frac{1}{m}\sum_{R\in\sigma,\mathbf{z}\in R^G,\mathbf{c}\notin R^A}\norm[\Big]{\sum_{\mathbbf{c}\in R^{A^k}}\varphi(p_{z_1\mathbf{c}_1}\cdots p_{z_{\ar(R)}\mathbf{c}_{\ar(R)}})\otimes\pi(p^{\mathbf{c}_1}_{c_1}\cdots p^{\mathbf{c}_{\ar(R)}}_{c_{\ar(R)}})}_{\nu\otimes\rho}^2\\
		&\leq\frac{1}{m}\sum_{R\in\sigma,\mathbf{z}\in R^G,\mathbf{c}\notin R^A,\mathbbm{c}\in R^{A^k}}\abs{R^{A^k}}\cdot\norm{p_{z_1\mathbf{c}_1}\cdots p_{z_{\ar(R)}\mathbf{c}_{\ar(R)}}}_{\nu\circ\varphi}^2\cdot\norm{p^{\mathbf{c}_1}_{c_1}\cdots p^{\mathbf{c}_{\ar(R)}}_{c_{\ar(R)}}}_{\tau}^2\\
		&\leq\frac{1}{m}\sum_{R\in\sigma,\mathbf{c}\notin R^A,\mathbbm{c}\in R^{A^k}}\abs{R^G}\cdot\abs{R^{A^k}}\cdot\norm{p^{\mathbf{c}_1}_{c_1}\cdots p^{\mathbf{c}_{\ar(R)}}_{c_{\ar(R)}}}_{\tau}^2\\
		&\leq M^2\defect(\tau)<\delta.
	\end{align*}
	On the other hand, since $\pi'(p^x_a)=\sum_{\mathbf{c}\in A^k}\varphi(p_{x\mathbf{c}})\otimes\pi(p^{\mathbf{c}}_a)=1\otimes\pi(p^{\mathbf{a}}_a)$ and similarly $\pi'(p^y_b)=1\otimes\pi(p^{\mathbf{b}}_b)$, so
	\begin{align*}
		\sum_{c,d\in A}\norm{[p^x_c,p^y_d]}_{\tau'}^2\geq\norm{[p^x_a,p^y_b]}_{\tau'}^2&=\norm{[p^{\mathbf{a}}_a,p^{\mathbf{b}}_b]}_\tau^2\geq\varepsilon.
	\end{align*}
	Hence $(G,x,y)$ is not an $a$-robust $Q$-commutativity gadget, giving the wanted contradiction.
\end{proof}

\ifthenelse{\boolean{anonymous}}{\clearpage}{}
\begin{theorem}\label{thm:c-v-robust-char}
	Let $A$ be a relational structure and $Q\subseteq qc$ be a min-tensor monoid. The following are equivalent.
	\begin{enumerate}[(i)]
		\item $\Mor^{c-v}_{\mbb{u}}(A^k,A)$ is $Q$-stably commutative for all $k\in\N$.
		\item $\Mor^{c-v}_{\mbb{u}}(A^{|A|^2},A)$ is $Q$-stably commutative.
		\item $A$ admits a $c$-$v$-robust $Q$-commutativity gadget.
		\item $\Mor^{c-c}_{\mbb{u}}(A^k,A)$ is $Q$-stably commutative for all $k\in\N$.
		\item $\Mor^{c-c}_{\mbb{u}}(A^{|A|^2},A)$ is $Q$-stably commutative.
		\item $A$ admits a $c$-$c$-robust $Q$-commutativity gadget.
	\end{enumerate}
\end{theorem}

The equivalence (iii) $\iff$ (vi) is proved for $Q=qa$ in \cite[Lemma 4.3]{CDdBVZ25}; the more general case is identical. We prove a stronger form of the equivalence (i) $\iff$ (iv) in the following lemma.

\begin{lemma}\label{lem:c-v-c-c-stable}
	Let $A$ and $B$ be relational structures over a signature $\sigma$. Suppose that for all $a\in A$, there exists $R\in\sigma$, $\mbf{a}\in R^A$, and $i\in[\ar(R)]$ such that $a_i=a$. Then $\Mor_{\mbb{u}}^{c-v}(A,B)$ is $Q$-stably commutative iff $\Mor_{\mbb{u}}^{c-c}(A,B)$ is.
\end{lemma}

Note that the condition on $A$ is not very strong. From the point of view of the constraint-variable game, it is simply asking that every variable in $A$ is asked in some constraint.

\begin{proof}
	Due to \cite{CM24}, the inclusion is a $4L$-homomorphism $\alpha:\Mor_{\mbb{u}}^{c-c}(A,B)\rightarrow\Mor_{\mbb{u}}^{c-v}(A,B)$, where $L=\max_{R\in\sigma}\ar(R)$; and there exists a $m$-homomorphism $\beta:\Mor_{\mbb{u}}^{c-v}(A,B)\rightarrow\Mor_{\mbb{u}}^{c-c}(A,B)$ that acts as identity on the generators $\Phi^{R,\mathbf{a}}_{\mathbf{b}}$, where $m=\sum_{R\in\sigma}|R^A|$. 
	
	First suppose $\Mor_{\mbb{u}}^{c-v}(A,B)$ is $Q$-stably commutative. Let $\varepsilon>0$. Then, there exists $\delta>0$ such that for any tracial state $\tau$ on $\Mor_{\mbb{u}}^{c-v}(A,B)$ over $Q$, if $\defect(\tau)<m\delta$, then $\norm{[p^a_b,p^{a'}_{b'}]}_\tau^2<\varepsilon$,  $\norm{[\Phi^{R,\mathbf{a}}_{\mathbf{b}},p^{a}_{b}]}_\tau^2<\varepsilon$, and  $\norm{[\Phi^{R,\mathbf{a}}_{\mathbf{b}},\Phi^{R',\mathbf{a}'}_{\mathbf{b}'}]}_\tau^2<\varepsilon$. Now, let $\tau$ be a tracial state on $\Mor_{\mbb{u}}^{c-c}(A,B)$ over $Q$ and suppose $\defect(\tau)<\delta$. By properties of $C$-homomorphisms, $\defect(\tau\circ\beta)<m\delta$, which implies that in $\Mor_{\mbb{u}}^{c-v}(A,B)$, $\norm{[\Phi^{R,\mathbf{a}}_{\mathbf{b}},\Phi^{R',\mathbf{a}'}_{\mathbf{b}'}]}_{\tau\circ\beta}^2<\varepsilon$. Since $\beta$ acts as the identity on these elements, $\norm{[\Phi^{R,\mathbf{a}}_{\mathbf{b}},\Phi^{R',\mathbf{a}'}_{\mathbf{b}'}]}_\tau^2<\varepsilon$ in $\Mor_{\mbb{u}}^{c-c}(A,B)$. Hence $\Mor_{\mbb{u}}^{c-c}(A,B)$ is $Q$-stably commutative.
	
	Now, suppose $\Mor_{\mbb{u}}^{c-c}(A,B)$ is $Q$-stably commutative. Let $\varepsilon>0$. Then, there exists $\delta_0>0$ such that for any tracial state $\tau$ on $\Mor_{\mbb{u}}^{c-c}(A,B)$ over $Q$, if $\defect(\tau)<\delta_0$, then $\norm{[\Phi^{R,\mathbf{a}}_{\mathbf{b}},\Phi^{R',\mathbf{a}'}_{\mathbf{b}'}]}_\tau^2<\frac{\varepsilon}{25S^4}$, where $S=\max_{R\in\sigma}|R^B|$. Let $\delta=\min\{\frac{\delta_0}{4L},\frac{\varepsilon}{50mL}\}$. Now, let $\tau$ be a tracial state on $\Mor_{\mbb{u}}^{c-v}(A,B)$ over $Q$ and suppose $\defect(\tau)<\delta$. By properties of $C$-homomorphisms, $\defect(\tau\circ\alpha)<4L\delta\leq\delta_0$, so $\norm{[\Phi^{R,\mathbf{a}}_{\mathbf{b}},\Phi^{R',\mathbf{a}'}_{\mathbf{b}'}]}_{\tau\circ\alpha}^2<\frac{\varepsilon}{25S^4}$ in $\Mor_{\mbb{u}}^{c-v}(A,B)$. Since $\alpha$ is the inclusion map, we get $\norm{[\Phi^{R,\mathbf{a}}_{\mathbf{b}},\Phi^{R',\mathbf{a}'}_{\mathbf{b}'}]}_\tau^2<\frac{\varepsilon}{25S^4}\leq\varepsilon$. To get the bounds on the remaining commutators, writing $\Phi^{R,\mathbf{a},i}_{b}=\sum_{\mathbf{b}\in R^B:\,b_i=b}\Phi^{R,\mathbf{a}}_{\mathbf{b}}$, we have that
	\begin{align*}
		\norm{\Phi^{R,\mathbf{a},i}_{b}-p^{a_i}_b}_\tau^2&=\tau(\Phi^{R,\mathbf{a},i}_{b}-p^{a_i}_b)+2\tau(\Phi^{R,\mathbf{a},i}_{b}(1-p^{a_i}_b))\\
		&\leq2\sum_{b\in B}\tau(\Phi^{R,\mathbf{a},i}_{b}(1-p^{a_i}_b))\\
		&=2\sum_{\mathbf{b}\in R^B}\norm{\Phi^{R,\mathbf{a}}_{\mathbf{b}}(1-p^{a_i}_{b_i})}_\tau^2\\
		&\leq2mL\defect(\tau)<2mL\delta\leq\frac{\varepsilon}{25}.
	\end{align*}
	Then, taking $\mathbf{a}'\in R'$ such that $a'_i=a$,
	\begin{align*}
		\norm{[\Phi^{R,\mathbf{a}}_{\mathbf{b}},p^{a}_{b}]}_\tau^2&=\norm[\Big]{\sum_{b_i'=b}[\Phi^{R,\mathbf{a}}_{\mathbf{b}},\Phi^{R',\mathbf{a}'}_{\mathbf{b}}]-[\Phi^{R,\mathbf{a}}_{\mathbf{b}},\Phi^{R',\mathbf{a}',i}_{b}-p^{a}_b]}_\tau^2\\
		&\leq3S\sum_{b_i'=b}\norm{[\Phi^{R,\mathbf{a}}_{\mathbf{b}},\Phi^{R',\mathbf{a'}}_{\mathbf{b'}}]}_\tau^2+6\norm{\Phi^{R',\mathbf{a}',i}_{b}-p^{a}_b}_\tau^2\\
		&<3S^2\frac{\varepsilon}{25S^4}+6\frac{\varepsilon}{25}\leq\varepsilon.
	\end{align*}
	Also, taking $\mathbf{a}\in R$ and $\mathbf{a}'\in R'$ such that $a_i=a$ and $a'_{i'}=a'$,
	\begin{align*}
		\norm{[p^a_b,p^{a'}_{b'}]}_\tau^2&=\norm[\Big]{\sum_{\substack{\mathbf{b}\in R^B:\,b_i=b\\\mathbf{b}'\in (R')^B:\,b_{i'}'=b}}[\Phi^{R,\mathbf{a}}_{\mathbf{b}},\Phi^{R',\mathbf{a}'}_{\mathbf{b}'}]-[\Phi^{R,\mathbf{a},i}_{\mathbf{b}},\Phi^{R',\mathbf{a}',i'}_{\mathbf{b}'}-p^{a'}_{b'}]-[\Phi^{R,\mathbf{a},i}_{\mathbf{b}}-p^a_b,p^{a'}_{b'}]}_\tau^2\\
		&\leq5S^2\sum_{\substack{\mathbf{b}\in R^B:\,b_i=b\\\mathbf{b}'\in (R')^B:\,b_{i'}'=b}}\norm{[\Phi^{R,\mathbf{a}}_{\mathbf{b}},\Phi^{R',\mathbf{a}'}_{\mathbf{b}'}]}_\tau^2+10\norm{\Phi^{R',\mathbf{a}',i'}_{\mathbf{b}'}-p^{a'}_{b'}}_\tau^2+10\norm{\Phi^{R,\mathbf{a},i}_{\mathbf{b}}-p^{a}_{b}}_\tau^2\\
		&<5S^4\frac{\varepsilon}{25S^4}+20\frac{\varepsilon}{25}=\varepsilon.\qedhere
	\end{align*}
\end{proof}

\begin{proof}[Proof of \cref{thm:c-v-robust-char}]
	Due \cite[Lemma 4.3]{CDdBVZ25}, we have the equivalence (iii) $\iff$ (iv) and due to \cref{lem:c-v-c-c-stable} we have the equivalences (i) $\iff$ (iv) and (ii) $\iff$ (v). Thus, it suffices to show (i) $\iff$ (ii) $\iff$ (iii).
	
	\fbox{(i)$\Rightarrow$(ii)} Immediate.
	
	\fbox{(ii)$\Rightarrow$(iii)} Suppose $\Mor^{c-v}_{\mbb{u}}(A^{|A|^2},A)$ is $Q$-stably commutative. Following the notation of \cref{lem:basic-converse}, we claim $(A^{|A|^2},\mathbf{x},\mathbf{y})$ is a $c$-$v$-robust $Q$-commutativity gadget for $A$. Since $\Mor^{c-v}_{\mbb{u}}(A^{|A|^2},A)$ is $Q$-stably commutative, we know that every representation of the oracular quantum homomorphism algebra $\Mor^{o+}(A^{|A|^2},A)=C^\ast_u(\Mor^{c-v}_{\mbb{u}}(A^{|A|^2},A)/\angnormal{\supp\mu_{c-v,\mbb{u}}})$ over $Q$ is commutative, and hence $\mor^{Q}(A^{|A|^2},A)\subseteq\mor^c(A^{|A|^2},A)$. By \cref{thm:orac-basic-char}, $(A^{|A|^2},\mathbf{x},\mathbf{y})$ is an oracular $Q$-commutativity gadget for $A$. For the additional $c$-$v$-robustness property, it follows from the stable commutativity property that for any $\varepsilon>0$, there exists $\delta>0$ such that for any tracial state $\tau$ on $\Mor^{c-v}_{\mbb{u}}(A^{|A|^2},A)$ over $Q$, if $\defect(\tau)<\delta$, then $\norm{[p^{\mathbf{a}}_{a},p^{\mathbf{b}}_{b}]}_\tau^2\leq\frac{\varepsilon}{|A|^2}$. So, this implies that $\sum_{a,b\in A}\norm{[p^{\mathbf{x}}_a,p^{\mathbf{y}}_b]}_\tau^2<\varepsilon$, and hence $(A^{|A|^2},\mathbf{x},\mathbf{y})$ is a $c$-$v$-robust $Q$-commutativity gadget.
	
	\fbox{(iii)$\Rightarrow$(i)} Note first that as in \cref{lem:c-v-c-c-stable}, we know that for any tracial state $\tau$ on $\Mor^{c-v}_{\mbb{u}}(A^k,A)$, $\norm{\Phi^{R,\mathbbf{c},i}_c-p^{\mathbf{c}_i}_c}_\tau^2\leq 2mL\defect(\tau)$. This implies that
	\begin{align*}
		&\norm{[\Phi^{R,\mathbbf{c},i}_c,p^{\mathbf{d}}_d]}_{\tau}^2\leq 3\norm{[p^{\mathbf{c}_i}_c,p^{\mathbf{d}}_d]}_\tau^2+12mL\defect(\tau)\\
		&\norm{[\Phi^{R,\mathbbf{c},i}_c,\Phi^{S,\mathbbf{d},j}_d]}_{\tau}^2\leq 5\norm{[p^{\mathbf{c}_i}_c,p^{\mathbf{d}_j}_d]}_\tau^2+40mL\defect(\tau)
	\end{align*}

	As in \cref{lem:basic-nogo}, we show the contrapositive by contradiction. Suppose $(G,x,y)$ is a $c$-$v$-robust $Q$-commutativity gadget. There exists $k\in\N$ such that $\Mor^{c-v}_{\mbb{u}}(A^k,A)$ is not $Q$-stably commutative. In particular that means that it is not $Q$-stably commutative with respect to the generating set $\{\Phi^{R,\mathbbf{c},i}_{c}\}_{R\in\sigma,\mathbbf{c}\in R^{A^k},i\in[\ar(R)],c\in A}\cup\{p^{\mathbf{c}}_c\}_{\mathbf{c}\in A^k,c\in A}$ via \cref{lem:generator-independence}. Write $m=\sum_{R\in\sigma}|R^G|$ and $M=\sum_{R\in\sigma}\abs{R^{A^k}}$. Then, there exists $\varepsilon>0$ such that for any $\delta>0$ there is a trace $\tau$ on $\Mor^{c-v}_{\mbb{u}}(A^k,A)$ over $Q$ such that $\defect(\tau)<\frac{\delta}{M}$ but $\norm{[\Phi^{R,\mathbbf{a},i}_{a},\Phi^{S,\mathbbf{b},j}_{b}]}_\tau^2\geq10\varepsilon$ for some $R,S\in\sigma$, $\mathbbf{a}\in R^{A^k}$, $\mathbbf{b}\in S^{A^k}$, $i\in[\ar(R)]$, $j\in[\ar(S)]$, and $a,b\in A$; or $\norm{[\Phi^{R,\mathbbf{a},i}_{a},p^{\mathbf{b}}_b]}_\tau^2\geq10\varepsilon$ for some $R\in\sigma$, $\mathbbf{a}\in R^{A^k}$, $i\in[\ar(R)]$, $\mathbf{b}\in A^k$, and $a,b\in A$; or $\norm{[p^{\mathbf{a}}_a,p^{\mathbf{b}}_b]}_\tau^2\geq10\varepsilon$ for some $\mathbf{a},\mathbf{b}\in A^k$ and $a,b\in A$. By the above argument, this implies that in either of the three cases, $\norm{[p^{\mathbf{a}}_a,p^{\mathbf{b}}_b]}_\tau^2\geq2\varepsilon-8mL\defect(\tau)$ for some $\mathbf{a},\mathbf{b}\in A^k$ and $a,b\in A$. We may assume without loss of generality that $\frac{\delta}{M}\leq\frac{\varepsilon}{8mL}$, giving $\norm{[p^{\mathbf{a}}_a,p^{\mathbf{b}}_b]}_\tau^2\geq\varepsilon$
	
	Fix $\delta>0$, and the corresponding $\tau,\mathbf{a},\mathbf{b},a,b$ as above. Let $\pi:\Mor^{c-v}_{\mbb{u}}(A^k,A)\rightarrow\mc{M}$ be the GNS representation of $\tau$, with a tracial state $\rho$ on $\mc{M}$ such that $\tau=\rho\circ\pi$. Let $\varphi$ be the representation of $\Mor^{o+}(G,A^k)$ defined in \cref{lem:basic-nogo}. Write $\nu$ for the tracial state on $\im\varphi$. Now, define the $\ast$-representation $\pi'$ of $\Mor^{c-v}_{\mbb{u}}(G,A)$ on the generators by
	\begin{align*}
		&\pi'(\Phi^{R,\mathbf{z}}_{\mathbf{c}})=\sum_{\mathbbf{c}\in R^{A^k}}\varphi(p_{z_1\mathbf{c}_1}\cdots p_{z_{\ar(R)}\mathbf{c}_{\ar(R)}})\otimes\pi(\Phi^{R,\mathbbf{c}}_\mathbf{c}),\\
		&\pi'(p^{z}_c)=\sum_{\mathbf{c}\in A^k}\varphi(p_{z\mathbf{c}})\otimes\pi(p^{\mathbf{c}}_c).
	\end{align*}
	To see that this is in fact a $\ast$-representation, it suffices to note that the relations of $\Mor^{c-v}_{\mbb{u}}(A^{|A|^2},A)$ are satisfied. First, $\pi'(\Phi^{R,\mathbf{z}}_{\mathbf{c}})$ and $\pi'(p^{z}_c)$ are projections as they are sums of orthogonal projections --- $p^{z_1}_{\mathbf{c}_1}\cdots p^{z_{\ar(R)}}_{\mathbf{c}_{\ar(R)}}$ are orthogonal projections that sum to $1$ over $\mathbbf{c}$ as the terms in the product commute. Also,
	\begin{align*}
		\sum_{\mathbf{c}\in R^A}\pi'(\Phi^{R,\mathbf{z}}_{\mathbf{c}})&=\sum_{\mathbf{c}\in R^A}\sum_{\mathbbf{c}\in R^{A^k}}\varphi(p_{z_1\mathbf{c}_1}\cdots p_{z_{\ar(R)}\mathbf{c}_{\ar(R)}})\otimes\pi(\Phi^{R,\mathbbf{c}}_\mathbf{c})\\
		&=\sum_{\mathbbf{c}\in R^{A^k}}\varphi(p_{z_1\mathbf{c}_1}\cdots p_{z_{\ar(R)}\mathbf{c}_{\ar(R)}})\otimes1=1,
	\end{align*}
	and
	\begin{align*}
		\sum_{c\in A}\pi'(p^{z}_c)&=\sum_{c\in A}\sum_{\mathbf{c}\in A^k}\varphi(p_{z\mathbf{c}})\otimes\pi(p^{\mathbf{c}}_c)\\
		&=\sum_{\mathbf{c}\in A^k}\varphi(p_{z\mathbf{c}})\otimes1 =1.
	\end{align*}%
	Now, since $(\nu\otimes\rho)$ is a tracial state on $\im\pi'$, $\tau'=(\nu\otimes\rho)\circ\pi'$ is a tracial state on $\Mor^{c-v}_{\mbb{u}}(G,A)$. The defect of this tracial state is
	\begin{align*}
		\defect(\tau')&=\frac{1}{m}\sum_{\substack{R\in\sigma,\mathbf{z}\in R^G,\\i\in[\ar(R)],\mathbf{c}\in R^A}}\frac{1}{\ar(R)}\tau'(\Phi^{R,\mathbf{z}}_{\mathbf{c}}(1-p^{z_i}_{c_i}))\\
		&=\frac{1}{m}\sum_{\substack{R\in\sigma,\mathbf{z}\in R^G,\\i\in[\ar(R)],c\neq d\in A}}\frac{1}{\ar(R)}\tau'(\Phi^{R,\mathbf{z},i}_{c}p^{z_i}_{d})\\
		&=\frac{1}{m}\sum_{\substack{R\in\sigma,\mathbf{z}\in R^G,\\i\in[\ar(R)],c\neq d\in A}}\frac{1}{\ar(R)}\sum_{\mathbbf{c}\in R^{A^k},\mathbf{d}\in  A^k}(\nu\circ\varphi)(p_{z_1\mathbf{c}_1}\cdots p_{z_{\ar(R)}\mathbf{c}_{\ar(R)}}p_{z_i\mathbf{d}})\tau(\Phi^{R,\mathbf{c},i}_cp^{\mathbf{d}}_d)\\
		&=\frac{1}{m}\sum_{\substack{R\in\sigma,\mathbf{z}\in R^G,\\i\in[\ar(R)],c\neq d\in A}}\frac{1}{\ar(R)}\sum_{\mathbbf{c}\in R^{A^k}}(\nu\circ\varphi)(p_{z_1\mathbf{c}_1}\cdots p_{z_{\ar(R)}\mathbf{c}_{\ar(R)}})\tau(\Phi^{R,\mathbf{c},i}_cp^{\mathbf{c}_i}_d)\\
		&\leq\frac{1}{m}\sum_{\substack{R\in\sigma,\mathbf{z}\in R^G,\\i\in[\ar(R)],c\neq d\in A}}\frac{1}{\ar(R)}\sum_{\mathbbf{c}\in R^{A^k}}\tau(\Phi^{R,\mathbf{c},i}_cp^{\mathbf{c}_i}_d)\\
		&=M\defect(\tau)<M\frac{\delta}{M}=\delta
	\end{align*}
	On the other hand, since $\pi'(p^x_a)=\sum_{\mathbf{c}\in A^k}\varphi(p_{x\mathbf{c}})\otimes\pi(p^{\mathbf{c}}_a)=1\otimes\pi(p^{\mathbf{a}}_a)$ and similarly $\pi'(p^y_b)=1\otimes\pi(p^{\mathbf{b}}_b)$, so
	\begin{align*}
		\sum_{c,d\in A}\norm{[p^x_c,p^y_d]}_{\tau'}^2\geq\norm{[p^x_a,p^y_b]}_{\tau'}^2&=\norm{[p^{\mathbf{a}}_a,p^{\mathbf{b}}_b]}_\tau^2\geq\varepsilon.
	\end{align*}
	Hence $(G,x,y)$ is not an $c$-$v$-robust $Q$-commutativity gadget, giving the wanted contradiction.
\end{proof}

To finish this section, we use the characterisation of robust commutativity gadgets, as well the properties of stable commutativity from \cref{sec:stability} to show some implications and equivalences between different types of commutativity gadget.

\begin{theorem}\label{thm:implications-equivalences}
	Let $A$ be a relational structure.
	\begin{enumerate}[(i)]
        \item If $Q\subseteq Q'$ and $A$ admits a $Q'$-commutativity gadget, then $A$ admits a $Q$-commutativity gadget.
        \item If $Q\subseteq Q'$ and $A$ admits an oracular $Q'$-commutativity gadget, then $A$ admits an oracular $Q$-commutativity gadget.
        \item If $Q\subseteq Q'$ and $A$ admits an $a$-robust $Q'$-commutativity gadget, then $A$ admits an $a$-robust $Q$-commutativity gadget.
        \item If $Q\subseteq Q'$ and $A$ admits a $c$-$v$-robust $Q'$-commutativity gadget, then $A$ admits a $c$-$v$-robust $Q$-commutativity gadget.
		\item If $A$ admits $Q$-commutativity gadget, then $A$ admits an oracular $Q$-commutativity gadget.
		\item If $A$ admits an $a$-robust $Q$-commutativity gadget, then $A$ admits a $c$-$v$-robust $Q$-commutativity gadget.
		\item If $A$ admits a $C^\ast$-commutativity gadget, then $A$ admits an $a$-robust $Q$-commutativity gadget for all $Q$.
		\item If $A$ admits an oracular $C^\ast$-commutativity gadget, then $A$ admits a $c$-$v$-robust $Q$-commutativity gadget for all $Q$.
		\item $A$ admits a $qa$-commutativity gadget iff $A$ admits an $a$-robust $qa$-commutativity gadget iff $A$ admits an $a$-robust $q$-commutativity gadget.
		\item $A$ admits an oracular $qa$-commutativity gadget iff $A$ admits a $c$-$v$-robust $qa$-commutativity gadget iff $A$ admits an $c$-$v$-robust $q$-commutativity gadget.
		\item $A$ admits a $qc$-commutativity gadget iff $A$ admits an $a$-robust $qc$-commutativity gadget.
		\item $A$ admits an oracular $qc$-commutativity gadget iff $A$ admits a $c$-$v$-robust $qc$-commutativity gadget.
	\end{enumerate}
\end{theorem}

The results of this theorem are summarised in \cref{fig:implications}

\begin{proof}
    (i) to (iv) follow as the commutativity gadget constructed in \cref{lem:basic-converse} has deterministic morphisms $\pi_{a,b}$ and every $Q$-morphism is a $Q'$-morphism. (v) and (vi) follow as the commutativity gadgets constructed in \cref{thm:basic-char,thm:a-robust-char} satisfy the conditions of~\cite[Lemma 4.7]{CDdBVZ25}. (vii) and (viii) follow using \cref{lem:quotient-commutativity} to note that a commutative quantum polymorphism algebra implies $Q$-stable commutativity. (ix) and (x) follow by using \cref{lem:q-to-qa} to equate $q$-stable commutativity and $qa$-stable commutativity, and using \cref{lem:ungapping} to equate $q$-stable commutativity and commutativity under Connes-embeddable representations. Finally, (xi) and (xii) follow by \cref{lem:qc-ungapping}.
\end{proof}

\begin{figure}
    \centering
    \ifthenelse{\boolean{anonymous}}{\def\myscale{1}}{\def\myscale{.98}\vspace*{-7mm}}
    \begin{tikzpicture}[scale=\myscale]
        \node (c) at (0,0) {$C^\ast$};
        \node (qc) [right=of c]  {$qc$};
        \node (qa) [right=of qc]  {$qa$};
        \node (q) [right=of qa]  {$q$};
        \node (rqc) [below=of qc] {$qc$};
        \node (rqa) [below=of qa] {$qa$};
        \node (rq) [below=of q] {$q$};
        \node (roqc) [above right=1 and 0.5 of qc] {$qc$};
        \node (roqa) [above right=1 and 0.5 of qa] {$qa$};
        \node (roq) [above right=1 and 0.5 of q] {$q$};
        \node (oqc) [above=of roqc] {$qc$};
        \node (oqa) [above=of roqa] {$qa$};
        \node (oq) [above=of roq] {$q$};
        \node (oc) [left=of oqc] {$C^\ast$};
        \node [left=2.3 of oc] {Oracular};
        \node [left=4.1 of roqc] {$c$-$v$-robust};
        \node [left=1.1 of c] {Non-oracular};
        \node [left=2.9 of rqc] {$a$-robust};
        \draw[gray] (-1,3.7) -- (-1,-2.2);
        \draw[->] (oc) -- (oqc);
        \draw[->] (oqc) -- (oqa);
        \draw[->] (oqa) -- (oq);
        \draw[->] (roqc) -- (roqa);
        \draw[->] (roqa) -- (roq);
        \draw[->,densely dotted] (oc) -- (roqc);
        \draw[->,gray] (roqc) -- (oqc);
        \draw[->,gray] (roqa) -- (oqa);
        \draw[->,gray] (roq) -- (oq);

        \draw[blue, thick, rounded corners] ($(oqa.north west)+(-0.3,0.3)$) -- ($(roqa.south west)+(-0.3,-0.3)$) -- ($(roq.south east)+(0.3,-0.3)$) -- ($(roq.north east)+(0.3,0.3)$) -- ($(roqa.north east)+(0.3,0.3)$) -- ($(oqa.north east)+(0.3,0.3)$) -- cycle;
        \draw[blue, thick, rounded corners] ($(oqc.north west)+(-0.3,0.3)$) -- ($(roqc.south west)+(-0.3,-0.3)$) -- ($(roqc.south east)+(0.3,-0.3)$) -- ($(oqc.north east)+(0.3,0.3)$) -- cycle;

        \draw[->,loosely dashed] (c) -- (oc);
        \draw[->,loosely dashed] (qc) -- (oqc);
        \draw[->,loosely dashed] (qa) -- (oqa);
        \draw[->,loosely dashed] (q) -- (oq);
        \draw[->,loosely dashed] (rqc) -- (roqc);
        \draw[->,loosely dashed] (rqa) -- (roqa);
        \draw[->,loosely dashed] (rq) -- (roq);
        
        \draw[->] (c) -- (qc);
        \draw[->] (qc) -- (qa);
        \draw[->] (qa) -- (q);
        \draw[->] (rqc) -- (rqa);
        \draw[->] (rqa) -- (rq);
        \draw[->,densely dotted] (c) -- (rqc);
        \draw[->,gray] (rqc) -- (qc);
        \draw[->,gray] (rqa) -- (qa);
        \draw[->,gray] (rq) -- (q);

        \draw[blue, thick, rounded corners] ($(qa.north west)+(-0.3,0.3)$) -- ($(rqa.south west)+(-0.3,-0.3)$) -- ($(rq.south east)+(0.3,-0.3)$) -- ($(rq.north east)+(0.3,0.3)$) -- ($(rqa.north east)+(0.3,0.3)$) -- ($(qa.north east)+(0.3,0.3)$) -- cycle;
        \draw[blue, thick, rounded corners] ($(qc.north west)+(-0.3,0.3)$) -- ($(rqc.south west)+(-0.3,-0.3)$) -- ($(rqc.south east)+(0.3,-0.3)$) -- ($(qc.north east)+(0.3,0.3)$) -- cycle;
    \end{tikzpicture}
    \caption{Relationships between commutativity gadget classes shown by~\cref{thm:implications-equivalences}. Trivial implications are denoted by solid grey arrows, the implications (i) -- (iv) are denoted by solid black arrows, the implications (v) and (vi) are denoted by dashed black arrows, the implications (vii) and (viii) are denoted by dotted black arrows, and the equivalences (ix) -- (xii) are denoted by blue boxes.}
    \label{fig:implications}
\end{figure}
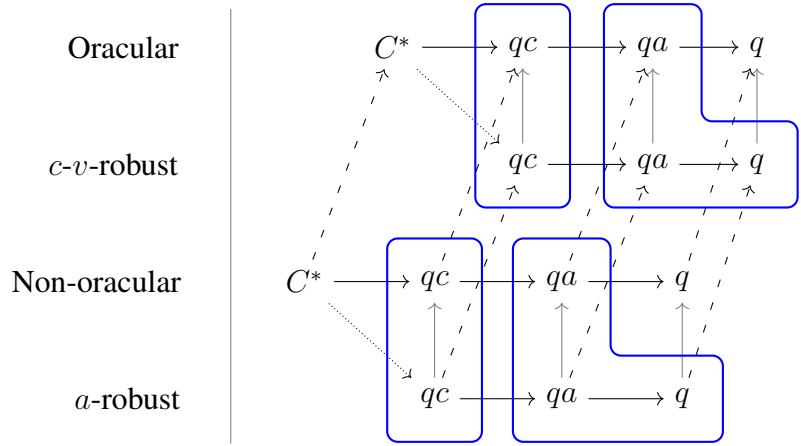

\ifthenelse{\boolean{anonymous}}{}{\FloatBarrier} 
\section{Polymorphisms of the complete graphs}\label{sec:complete-graphs}

In this section, we study the quantum polymorphisms of the complete graphs $K_n$, which we know to be nonclassical~\cite{CDdBVZ25}.

Before studying the quantum case, we mention the classical case, which will motivate our analysis. First, it is known that the polymorphisms $f:K_n^{k}\rightarrow K_n$ have the form $f=\sigma\circ\pi_i$, where $\pi_i:K_n^{k}\rightarrow K_n$ is the canonical projection $\pi_i(\mathbf{x})=x_i$ and $\sigma:K_n\rightarrow K_n$ is a permutation. As such, it follows that the commutative algebra $\Mor(K_n^{ k},K_n)$ is generated by the projections $p_{\mathbf{x}a}(\sigma\circ\pi_i)=\delta_{\sigma(x_i),a}$. To realise the bijection $\mor(K_n^{ k},K_n)=S_n^{\sqcup k}$ at the level of the algebra, let $\mathbf{x},\mathbf{x}'\in K_n^{ k}$ such that $x_i=x_i'$ but $x_j\neq x_j'$ for all $j\neq i$, and consider $\Pi_i=\sum_{a}p_{\mathbf{x}a}p_{\mathbf{x}'a}$. Then, we have that $\Pi_i(\sigma\circ\pi_j)=\delta_{i,j}$, which gives that $\Pi_i$ is independent of the choice of $\mathbf{x},\mathbf{x}'$ satisfying $x_j=x_j'\iff j=i$ for all $j$, and that $\{\Pi_i\}_{i\in[k]}$ forms a resolution of identity. By construction, $\Pi_i\Mor(K_n^{ k},K_n)$ admits the coproduct $\Delta(f)(\sigma\circ\pi_j,\tau\circ\pi_{j'})=\delta_{i,j,j'}f(\sigma\circ\tau\circ\pi_i)$, giving that $\Pi_i\Mor(K_n^{ k},K_n)\cong C(S_n)$. 

\begin{theorem}\label{thm:sn-decomposition}
	For $n\geq 3$, $\mor^+(K_n^{ k},K_n)$ is isomorphic to $(S_n^+)^{\sqcup k}$ in the following way: there exists a central PVM $\{\Pi_i\}_{i=1}^k$ in $\Mor^+(K_n^{ k},K_n)$ such that for each $i$, the $C^\ast$-algebra $$\Pi_i\Mor^+(K_n^{ k},K_n)\cong\Mor^+(K_n,K_n)=S_n^+$$ via the $\ast$-homomorphism extending $\Pi_ip_{\mathbf{x}a}\mapsto p_{x_i a}$.
\end{theorem}

In the following we always assume $n\geq 3$.

\begin{lemma}\label{lem:sn-diagonals}
	Let $\mathbf{x}_1,\ldots,\mathbf{x}_n\in K_n^{ k}$ be such that $x_{ij}\neq x_{i'j}$ for all $i\neq i'$. Then, for all $a\in K_n$, $p_{\mbf{x}_1a}+\ldots+p_{\mbf{x}_na}=1$ in $\Mor^+(K_n^{ k},K_n)$.	
\end{lemma}

\begin{proof}
	Since $\mbf{x}_i\sim_{K_n^{ k}}\mbf{x}_j$ for all $i\neq j$ by hypothesis, we have that $p_{\mathbf{x}_ia}p_{\mathbf{x}_ja}=0$. As such, the $p_{\mathbf{x}_ia}$ are orthogonal projections, giving $p_{\mbf{x}_1a}+\ldots+p_{\mbf{x}_na}\leq 1$. Summing over $a$ gives
	$$n=\sum_i1=\sum_{i,a}p_{\mathbf{x}_ia}\leq\sum_{a}1=n,$$
	and hence the inequalities must all be equalities.
\end{proof}

\begin{lemma}\label{lem:the-swapperrr}
	Let $\mathbf{x},\mathbf{x}'\in K_n^{ k}$ such that $x_i\neq x_i'$ for all $i$. Fix $S\subseteq[k]$ and define $\mathbf{y},\mathbf{y}'\in K_n^{ k}$ as $y_i=x_i$ and $y_i'=x_i'$ if $i\in S$, and $y_i=x_i'$ and $y_i'=x_i$ otherwise. Then, for all $a\in K_n$, $p_{\mathbf{x}a}+p_{\mathbf{x}'a}=p_{\mathbf{y}a}+p_{\mathbf{y}'a}$ in $\Mor^+(K_n^{ k},K_n)$.
\end{lemma}

\begin{proof}
	There exist $\mathbf{x}_3,\ldots,\mathbf{x}_n\in K_n^{ k}$ such that $x_{ij}\neq x_{i'j}$ for all $i\neq i'$, and $x_{ij}\neq x_j,x_j'$. Then, by \cref{lem:sn-diagonals}, we have that $p_{\mbf{x}a}+p_{\mbf{x}'a}+p_{\mbf{x}_3a}+\ldots+p_{\mbf{x}_na}=1$ and  $p_{\mbf{y}a}+p_{\mbf{y}'a}+p_{\mbf{x}_3a}+\ldots+p_{\mbf{x}_na}=1$. Putting these together gives the wanted result.
\end{proof}

\begin{lemma}\label{lem:lotsa-commutation}
	For all $\mathbf{x},\mathbf{y}\in K_n^{ k}$ and $a\in K_n$, $p_{\mathbf{x}a}$ and $p_{\mathbf{y}a}$ commute in $\Mor^+(K_n^{ k},K_n)$.
\end{lemma}

\begin{proof}
	Consider first the case that $x_i\neq y_i$ for all $i$. Then, as $\mbf{x}\sim_{K_n^{ k}}\mbf{y}$, $p_{\mbf{x}a}p_{\mbf{y}a}=0$, and hence they commute by the defining relations of $\Mor^+(K_n^{ k},K_n)$. In the other case, let $S=\set*{i}{x_i=y_i}\neq\varnothing$. Now, let $\mathbf{x}',\mathbf{y}'\in K_n^{ k}$ such that $x_i'=y_i'\neq x_i$ for $i\in S$, and $x_i'=y_i$ and $y_i'=x_i$ for $i\notin S$. Hence $\mbf{x},\mbf{x}',\mbf{y},\mbf{y}'$ satisfy the conditions of \cref{lem:the-swapperrr}, and hence $p_{\mbf{x}a}+p_{\mbf{x}'a}=p_{\mbf{y}a}+p_{\mbf{y}'a}$. Hence, we find that
	\begin{align*}
		[p_{\mbf{x}a},p_{\mbf{y}a}]=[p_{\mbf{x}a},p_{\mbf{x}a}+p_{\mbf{x}'a}-p_{\mbf{y}'a}]=-[p_{\mbf{x}a},p_{\mbf{y}'a}].
	\end{align*}
	In particular, this tells us that $[p_{\mbf{x}a},p_{\mbf{y}a}]$ does not depend on $y_i$ for $i\notin S$. Let $\mathbf{z}\in K_n^{ k}$ be such that $z_i=x_i=y_i$ for all $i\in S$, and $z_i\neq x_i$ and $z_i\neq y_i$ for all $i\notin S$. Therefore,
	$$[p_{\mathbf{x}a},p_{\mathbf{y}a}]=[p_{\mathbf{x}a},p_{\mathbf{z}a}]=-[p_{\mathbf{z}a},p_{\mathbf{x}a}]=-[p_{\mathbf{z}a},p_{\mathbf{y}a}]=-[p_{\mathbf{x}a},p_{\mathbf{y}a}],$$
	giving that $[p_{\mathbf{x}a},p_{\mathbf{y}a}]=0$.
\end{proof}

\begin{lemma}\label{lem:cool-projections}
	Let $S\subseteq[k]$ and define $\Pi_S=\sum_{a\in K_n}p_{\mathbf{x}a}p_{\mathbf{y}a}$ for some $\mathbf{x},\mathbf{y}\in K_n^{ k}$ such that $x_i=y_i\iff i\in S$. Then,
	\begin{enumerate}[(i)]
		\item $\Pi_S$ does not depend on the choice of $\mathbf{x},\mathbf{y}$ satisfying $x_i=y_i\iff i\in S$;
		\item $\Pi_S$ is a projection;
		\item $\Pi_S$ is central;
		\item $\Pi_Sp_{\mathbf{x}a}$ does not depend on $x_i$ for $i\notin S$; and
		\item For $T\subseteq S$, $\Pi_S=\Pi_T+\Pi_{S\backslash T}$.
	\end{enumerate}
\end{lemma}

\begin{proof}

\begin{enumerate}[(i)]
	\item Let $\tilde{\mathbf{x}},\tilde{\mathbf{y}}\in K_n^{ k}$ such that $\tilde{x}_i=\tilde{y}_i\iff i\in S$ Define $\mathbf{x}',\mathbf{y}'\in K_n^{ k}$ such that $x_i'=y_i'\neq x_i,\tilde{x}_i$ for all $i\in S$, and $x_i'=y_i$ and $y_i'=x_i$ for all $i\notin S$. Then,
	\begin{align*}
		\Pi_S&=\sum_{a\in K_n}p_{\mathbf{x}a}p_{\mathbf{y}a}=\sum_{a\in K_n}p_{\mathbf{x}a}(p_{\mbf{x}a}+p_{\mbf{x}'a}-p_{\mbf{y}'a})=1-\sum_{a\in K_n}p_{\mathbf{x}a}p_{\mathbf{y}'a}
	\end{align*}
	Next, let $\mathbf{x}'',\mathbf{y}''\in K_n^{ k}$ such that $x_i''=y_i'$ and $y_i''=x_i$ for all $i\in S$, and $x_i''=y_i''\neq x_i,\tilde{x}_i$ for all $i\notin S$. Then,
	\begin{align*}
		\Pi_S&=1-\sum_{a\in K_n}(p_{\mathbf{y}'a}+p_{\mathbf{y}''a}-p_{\mathbf{x}''a})p_{\mathbf{y}'a}=\sum_{a\in K_n}p_{\mathbf{x}''a}p_{\mathbf{y}'a}.
	\end{align*}
	Then, let $\mathbf{x}''',\mathbf{y}'''\in K_n^{ k}$ such that $x_i'''=y_i'''\neq y_i,\tilde{y}_i$ for all $i\in S$, and $x_i'''=y_i'$ and $y_i'''=x_i''$ for all $i\notin S$. Then,
	\begin{align*}
		\Pi_S&=\sum_{a\in K_n}p_{\mathbf{x}''a}(p_{\mbf{x}''a}+p_{\mbf{x}'''a}-p_{\mbf{y}'''a})=1-\sum_{a\in K_n}p_{\mathbf{x}''a}p_{\mathbf{y}'''a}.
	\end{align*}
	By construction, $\mathbf{x}''$ is arbitrary such that $x_i''\neq x_i,\tilde{x}_i$ for all $i$, and $\mathbf{y}'''$ is arbitrary such that $y_i'''\neq y_i,\tilde{y}_i$ for all $i$. As such, we can reverse the process to get
	\begin{align*}
		\Pi_S&=1-\sum_{a\in K_n}p_{\mathbf{x}''a}p_{\mathbf{y}'''a}=\sum_{a\in K_n}p_{\tilde{\mathbf{x}}a}p_{\tilde{\mathbf{y}}a},
	\end{align*}
	giving the wanted independence.
	
	\item As $p_{\mathbf{x}a}$ and $p_{\mathbf{y}a}$ commute by \cref{lem:lotsa-commutation}, $p_{\mathbf{x}a}p_{\mathbf{y}a}$ is a projection. Further, if $a\neq a'$,
	$$(p_{\mathbf{x}a}p_{\mathbf{y}a})(p_{\mathbf{x}a'}p_{\mathbf{y}a'})=p_{\mathbf{y}a}p_{\mathbf{x}a}p_{\mathbf{x}a'}p_{\mathbf{y}a'}=0,$$
	so $\Pi_S$ is a sum of orthogonal projections, and hence it is a projection.
	
	\item We have that $p_{\mathbf{x}a}\Pi_S=p_{\mathbf{x}a}p_{\mathbf{y}a}=p_{\mathbf{y}a}p_{\mathbf{x}a}=\Pi_Sp_{\mathbf{x}a}$.
	
	\item Let $\mathbf{x},\mathbf{x}'\in K_n^{ k}$ such that $x_i=x_i'$ for all $i\in S$. Then, let $\mathbf{y}\in K_n^{ k}$ such that $y_i=x_i$ for all $i\in S$ and $y_i\neq x_i,x_i'$ for all $i\notin S$. Then,
	$$\Pi_Sp_{\mathbf{x}a}=p_{\mathbf{x}a}p_{\mathbf{y}a}=\Pi_Sp_{\mathbf{y}a}=\Pi_Sp_{\mathbf{x}'a}.$$
	
	\item Let $\mathbf{x},\mathbf{y}\in K_n^{ k}$ such that $x_i=y_i\iff i\in T$. Define $\mathbf{y}',\mathbf{z},\mathbf{z}'\in K_n^{ k}$ such that $y'_i=z'_i\neq x_i$ and $z_i=x_i$ for $i\in T$, $y'_i=z_i=x_i$ and $z'_i=y_i$ for $i\in S\backslash T$, and $y_i'=z'_i\neq x_i,y_i$ and $z_i=y_i$ for $i\notin S$. By construction, \cref{lem:the-swapperrr} implies that $p_{\mathbf{y}a}+p_{\mathbf{y}'a}=p_{\mathbf{z}a}+p_{\mathbf{z}'a}$, and
	\begin{align*}
		\Pi_T=\sum_{a\in K_n}p_{\mathbf{x}a}p_{\mathbf{y}a}=\sum_{a\in K_n}p_{\mathbf{x}a}(p_{\mathbf{z}a}+p_{\mathbf{z}'a}-p_{\mathbf{y}'a})=\Pi_S-\Pi_{S\backslash T}.
	\end{align*}
\end{enumerate}

\end{proof}

\begin{proof}[Proof of \cref{thm:sn-decomposition}]
	Let $\Pi_i=\Pi_{\{i\}}$, as defined in \cref{lem:cool-projections}. Via \cref{lem:cool-projections} (ii) and (iii) these are central projections; and via \cref{lem:cool-projections} (v), $1=\Pi_{[k]}=\Pi_{1}+\Pi_{[k]\backslash\{1\}}=\ldots=\Pi_{1}+\ldots+\Pi_k$, so we get a PVM as wanted. Now, fix $i\in[k]$, and let $\varphi:\End^+(K_n)\rightarrow \Pi_i\Mor^+(K_n^{ k},K_n)$ be the $\ast$-homomorphism defined on the generators $\varphi(p_{xa})=\Pi_ip_{\mathbf{x}a}$ for $\mathbf{x}$ with $x_i=x$. First, due to \cref{lem:cool-projections} (iv), this is well-defined on the generators. To see that it extends to a $\ast$-homomorphism, it suffices to see that it satisfies the relations. First, since $\Pi_i$ is a central projection, $\Pi_ip_{\mathbf{x}a}$ is a projection. Next,
	\begin{align*}
		\sum_a\varphi(p_{xa})=\sum_a\Pi_ip_{\mathbf{x}a}=\Pi_i=\varphi(1).
	\end{align*}
	Finally, for all $x\neq y$, there exists some $\mathbf{x},\mathbf{y}$ such that $x_i=x$, $y_i=y$, and $x_j\neq y_j$ for all $j$. Then
	\begin{align*}
		\varphi(p_{xa})\varphi(p_{ya})=\Pi_ip_{\mbf{x}a}p_{\mbf{y}a}=0.
	\end{align*}
	So $\varphi$ is a $\ast$-homomorphism.
	
	For the other direction, let $\psi:\Mor^+(K_n^{ k},K_n)\rightarrow\End^+(K_n)$ be the $\ast$-homomorphism defined on the generators as $\psi(p_{\mathbf{x}a})=p_{x_ia}$. To see that it is in fact a $\ast$-homomorphism, we verify the relations: $p_{x_ia}$ is a projection;
	$$\sum_{a}\psi(p_{\mathbf{x}a})=\sum_ap_{x_ia}=1;$$
	and if $x_j\neq y_j$ for all $j$, then
	$$\psi(p_{\mathbf{x}a})\psi(p_{\mathbf{y}a})=p_{x_ia}p_{y_ia}=0.$$
	Also, since $\psi(\Pi_i)=\sum_ap_{x_ia}p_{x_ia}=1$, $\psi$ restricts to a $\ast$-homomorphism $\bar{\psi}:\Pi_i\Mor^+(K_n^{ k},K_n)\rightarrow\End^+(K_n)$. It is clear that $\bar{\psi}=\varphi^{-1}$, giving the wanted isomorphism.
\end{proof}

\begin{definition}
	Let $f\in\mor(A,B)$. Define the character $\hat{f}:\Mor^+(A,B)\rightarrow\C$ on the generators $\hat{f}(p_{ab})=\delta_{f(a),b}$.
\end{definition}

The map $f\mapsto\hat{f}$ induces the isomorphism $\mc{R}(\sigma)\cong\mc{R}^d(\sigma)$, so it is well-defined~\cite[Lemma 3.7]{CDdBVZ25}.

\begin{corollary}
	Let $\pi$ be an irreducible representation of $\Mor^+(K_n^{ k},K_n)$. Then, there exists $i\in[k]$ and a representation $\pi'$ of $\Aut(K_n)$ such that $\pi=\pi'\circ\hat{\pi}_i.$
\end{corollary}

\begin{proof}
	Since $\pi$ is irreducible, there exists a unique $i\in[k]$ such that $\pi(\Pi_i)\neq 0$ --- else $\pi$ is reducible into representations on the $\supp\pi(\Pi_i)$. Now, let $\pi'=\pi\varphi$. To finish, we have that
	\begin{align*}
		(\pi'\circ\hat{\pi}_i)(p_{\mathbf{x}a})&=(\hat{\pi}_i\otimes\pi')\Delta_{K_n^{ k},K_n}^{K_n}(p_{\mathbf{x}a})=\sum_{b\in K_n}\hat{\pi}_i(p_{\mathbf{x}b})\otimes\pi'(p_{ba})=\sum_{b\in K_n}\delta_{x_i,b}\pi\varphi(p_{ba})\\
		&=\pi\varphi(p_{x_ia})=\pi(\Pi_ip_{\mathbf{x}a})=\pi(p_{\mathbf{x}a}).\qedhere
	\end{align*}
\end{proof}

\section{Separating classes of commutativity gadgets}
\label{sec:separations}

In this section, we show the following result, which relies on the full characterisation of the quantum polymorphisms of the complete graph studied in the previous section.

\begin{theorem}\label{thm:separations}
    There exists a relational structure that admits a $q$-commutativity gadget but no $qa$-commutativity gadget.

    If there exists a non-hyperlinear group, then there exists a relational structure that admits a $qa$-commutativity gadget but no $qc$-commutativity gadget.
\end{theorem}

To construct these, we need some preliminary results.

\begin{definition}
    Let $A$ be a relational structure. We define the \emph{completion} of $A$ to be the relational structure $A_K$ with domain $\dom(A_K)=\dom(A)$ and relations $\rel(A_K)=\rel(A)\cup\{E\}$, where $E=\set*{(a,b)\in A^2}{a\neq b}$ is the (arity-$2$) edge relation of the complete graph.
\end{definition}

\begin{lemma}\label{cor:cheating}
	Let $A$ be a relational structure with $|A|\geq 3$. Then, the quantum polymorphism space of the completion is $\mor^+(A_K^{k}, A_K)\cong\mor^+(A_K)^{\sqcup k}$.
\end{lemma}

\begin{proof}
	Let $n=|A|$. We can see $\Mor^+(A_K^{k}, A_K)$ as the quotient of $\Mor^+(K_n^{k},K_n)$ by the additional relations coming from $A$; and similarly for $\Mor^+(A_K)$ with regards to $\Mor^+(K_n)$. Recall the isomorphisms from \cref{thm:sn-decomposition}: $\Psi:\Mor^+(K_n^{k},K_n)\rightarrow\Mor^+(K_n)^{\oplus k}$ and  $\Phi:\Mor^+(K_n)^{\oplus k}\rightarrow \Mor^+(K_n^{k},K_n)$ defined on the generators as $\Psi(p_{\mathbf{x}a})=p_{x_1 a}\oplus\ldots\oplus p_{x_ka}$ and $\Phi((p_{xa})_i)=\Pi_ip_{\mathbf{x}a}$ where $(p_{xa})_i=0\oplus \ldots 0\oplus p_{xa}\oplus 0\oplus\ldots\oplus 0$, for $p_{xa}$ in the $i$-th component, and $\mathbf{x}$ is such that $x_i=x$. Also, let $\iota_1:\Mor^+(K_n^{k},K_n)\rightarrow\Mor^+(A_K^{ k},A_K)$ and $\iota_2:\Mor^+(K_n)^{\oplus k}\rightarrow\Mor^+(A_K)^{\oplus k}$ be the quotient maps. First, we show that $\iota_2\Psi$ satisfies the relations of $\Mor^+(A_K^{ k},A_K)$. In fact, let $R\in\rel(A)$, $\mathbf{x}_1,\ldots,\mathbf{x}_{k}\in R$, and $\mathbf{a}\notin R$. Then,
	\begin{align*}
		&\iota_2\Psi(p_{(x_{11},\ldots,x_{k1})a_1})\cdots\iota_2\Psi(p_{(x_{1\ar(R)},\ldots,x_{k\ar(R)})a_{\ar(R)}})\\
		&=(p_{x_{11}a_1}\oplus\ldots\oplus p_{x_{k1}a_1})\cdots(p_{x_{1\ar(R)}a_{\ar(R)}}\oplus\ldots\oplus p_{x_{k\ar(R)}a_{\ar(R)}})\\
		&=(p_{x_{11}a_1}\cdots p_{x_{1\ar(R)}a_{\ar(R)}})\oplus\ldots\oplus(p_{x_{k1}a_1}\cdots p_{x_{k\ar(R)}a_{\ar(R)}})=0,
	\end{align*}
	as $p_{x_{i1}a_1}\cdots p_{x_{i\ar(R)}a_{\ar(R)}}=0$ in $\Mor^+(A_K)$. Therefore, $\Psi$ induces a $\ast$-homomorphism $\bar{\Psi}:\Mor^+(A_K^{ k},A')\rightarrow\Mor^+(A_K)^{\otimes k}$.
	
	For the other direction, we want to show that $\iota_1\Phi$ respects the relations of $\Mor^+(A_K)^{\oplus k}$. In fact, let $i\in[k]$, $R\in\rel(A)$, $\mathbf{x}\in R$, and $\mathbf{a}\notin R$. Then, letting $x_{ij}=x_i$ for all $i\in[k]$,
	\begin{align*}
		\iota_1\Phi((p_{x_1a_1})_i)\cdots\iota_1\Phi((p_{x_{\ar(R)}a_{\ar(R)}})_i)&=\Pi_ip_{\mathbf{x}_1a_1}\cdots\Pi_ip_{\mathbf{x}_{\ar(R)}a_{\ar(R)}}\\
		&=\Pi_i(p_{\mathbf{x}_1a_1}\cdots p_{\mathbf{x}_{\ar(R)}a_{\ar(R)}})=0.
	\end{align*}
	Therefore, $\Phi$ induces a $\ast$-homomorphism $\bar{\Phi}:\Mor^+(A_K)^{\oplus k}\rightarrow\Mor^+(A_K^{ k},A_K)$. To finish, note that $\bar{\Phi}=\bar{\Psi}^{-1}$, which follows since $\Phi=\Psi^{-1}$.
\end{proof}

Next, we consider a graph encoding the structure of a pair of relational structures from~\cite{DdBRS25}.

\begin{definition}
    Let $A$ and $B$ be relational structures over $\sigma$. $G(A,B)$ is the bipartite coloured graph with vertices
    $$V(G(A,B))=\set{(a,b)}{a\in A,\,b\in B}\cup\set{(R,\mbf{a},\mbf{b})}{R\in\sigma,\,\mbf{a}\in R^A,\,\mbf{b}\in R^B:\;a_i=a_j\Rightarrow b_i=b_j},$$
    edges $(a,b)\sim_{G(A,B)}(R,\mbf{a},\mbf{b})$ iff there exists there exists $i\in[\ar(R)]$ such that $a_i=a$ and $b_i=b$, and unary colour relations $C_a=\set{(a,b)}{b\in B}$ for all $a\in A$ and $C_{(R,\mbf{a})}=\set{(R,\mbf{a},\mbf{b})}{\mbf{b}\in R^B}$ for all $R\in\sigma$ and $\mbf{a}\in R^A$.
\end{definition}

The for linear constraint systems, the completion of the graph encoding has quantum transformations given by the solution group.

\begin{theorem}\label{thm:graph-to-lin}
     Let $A$ be a linear relational structure. Then $\mor^+(G(A,\mathrm{LIN})_K)\cong\mor^{o+}(A_H,\mathrm{LIN})$.
\end{theorem}

\begin{proof}
    First, note that due to the unary relations on $G(A,\mathrm{LIN})$, in $\Mor^+(G(A,B)_K)$, $p_{(a,b)(R,\mbf{a},\mbf{b})}=p_{(R,\mbf{a},\mbf{b})(a,b)}=p_{(a,b)(a',b')}=p_{(R,\mbf{a},\mbf{b})(R',\mbf{a}',\mbf{b}')}=0$ if $a\neq a'$ and $(R,\mbf{a})\neq(R',\mbf{a}')$. Then, define $p_{(a,b)b'}:=p_{(a,b)(a,b')}$ and $p_{(R,\mbf{a},\mbf{b})\mbf{b}'}:=p_{(R,\mbf{a},\mbf{b})(R,\mbf{a},\mbf{b}')}$. These generate $\Mor^+(G(A,\mathrm{LIN})_K)$ subject to the relations $\sum_{b'\in \Z_2}p_{(a,b)b'}=\sum_{b\in \Z_2}p_{(a,b)b'}=\sum_{\mbf{b}'\in R^{\mathrm{LIN}}}p_{(R,\mbf{a},\mbf{b})\mbf{b}'}=\sum_{\mbf{b}\in R^{\mathrm{LIN}}}p_{(R,\mbf{a},\mbf{b})\mbf{b}'}=1$, and $p_{(a,b)b'}p_{(R,\mbf{a},\mbf{b})\mbf{b}'}=0$ if there exists $i\in[\ar(R)]$ such that $a_i=a$ and for all such $i$, $b_i=b$ but $b_i'\neq b'$. First, note that the first two relations are $p_{(a,0)0}+p_{(a,0)1}=p_{(a,1)0}+p_{(a,1)1}=p_{(a,0)0}+p_{(a,1)0}=p_{(a,0)1}+p_{(a,1)1}=1$. This implies that $p_{(a,0)0}=p_{(a,1)1}$ and $p_{(a,0)1}=p_{(a,1)0}$. Next, we have that 
    \begin{align*}
        p_{(R,\mbf{a},\mbf{b})\mbf{b}'}=p_{(a_i,b_i)b_i'}p_{(a_j,b_j)b_j'}
        p_{(R,\mbf{a},\mbf{b})\mbf{b}'}=p_{(a_j,b_j)b_j'}p_{(a_i,b_i)b_i'}
        p_{(R,\mbf{a},\mbf{b})\mbf{b}'},
    \end{align*}
    so $[p_{(a_i,b_i)b_i'},p_{(a_j,b_j)b_j'}]p_{(R,\mbf{a},\mbf{b})\mbf{b}'}=0$. But, since $\mathrm{LIN}$ is boolean, this implies that for all $c,c'\in \Z_2$, $[p_{(a_i,b_i)c},p_{(a_j,b_j)c'}]p_{(R,\mbf{a},\mbf{b})\mbf{b}'}=0$, so we can sum over $\mbf{b}'$ and get $[p_{(a_i,b_i)c},p_{(a_j,b_j)c'}]=0$. By the relations on the $p_{(a,b)b'}$, this gives in fact that $[p_{(a_i,b)c},p_{(a_j,b')c'}]=0$ for all $b,b',c,c'\in \Z_2$. Then, we have that 
    \begin{align*}
        p_{(R,\mbf{a},\mbf{b})\mbf{b}'}=p_{(a_1,b_1)b_1'}\cdots p_{(a_{\ar(R)},b_{\ar(R)})b_{\ar(R)}'}p_{(R,\mbf{a},\mbf{b})\mbf{b}'}=p_{(a_1,b_1)b_1'}\cdots p_{(a_{\ar(R)},b_{\ar(R)})b_{\ar(R)}'}.
    \end{align*}
    In particular, this implies that $p_{(R,\mbf{a},\mbf{b})\mbf{b}'}=p_{(R,\mbf{a},\mbf{0})(\mbf{b}+\mbf{b}')}$.

    Now, define the map $\varphi:\Mor^+(G(A,\mathrm{LIN})_K)\rightarrow\Mor^{o+}(A_H,\mathrm{LIN})$ on the generators by $\varphi(p_{(a,b)b'})=p_{a(b+b')}$ and $\varphi(p_{(R,\mbf{a},\mbf{b})\mbf{b}'})=p_{a_1(b_1+b_1')}\cdots p_{a_{\ar(R)}(b_{\ar(R)}+b_{\ar(R)}')}$. To see that this extends to a $\ast$-homomorphism, it remains to see that the relations of $\Mor^+(G(A,\mathrm{LIN})_K)$ remain satisfied. In fact, the $\varphi(p_{(a,b)b'})$ and $\varphi(p_{(R,\mbf{a},\mbf{b})\mbf{b}'})$ are projections;
    \begin{align*}
        &\sum_{b'\in\Z_2}\varphi(p_{(a,b)b'})=\sum_{b'\in\Z_2}p_{a(b+b')}=1,\\
        &\sum_{b\in\Z_2}\varphi(p_{(a,b)b'})=\sum_{b\in\Z_2}p_{a(b+b')}=1,\\
        &\sum_{\mbf{b}'\in R^{\mathrm{LIN}}}\varphi(p_{(R,\mbf{a},\mbf{b})\mbf{b}'})=\sum_{\mbf{b}'\in R^{\mathrm{LIN}}}p_{a_1(b_1+b_1')}\cdots p_{a_{\ar(R)}(b_{\ar(R)}+b'_{\ar(R)})}=\sum_{\mbf{b}\in\Z_2^{\ar(R)}}p_{a_1b_1}\cdots p_{a_{\ar(R)}b_{\ar(R)}}=1,\\
        &\sum_{\mbf{b}\in R^{\mathrm{LIN}}}\varphi(p_{(R,\mbf{a},\mbf{b})\mbf{b}'})=\sum_{\mbf{b}\in R^{\mathrm{LIN}}}p_{a_1(b_1+b_1')}\cdots p_{a_{\ar(R)}(b_{\ar(R)}+b'_{\ar(R)})}=\sum_{\mbf{b}\in\Z_2^{\ar(R)}}p_{a_1b_1}\cdots p_{a_{\ar(R)}b_{\ar(R)}}=1;
    \end{align*}
    and for $\mathbf{a}\in R^A$, $\mathbf{b},\mathbf{b}'\in R^{\mrm{LIN}}$, and $b'\neq b_i'$,
    \begin{align*}
        \varphi(p_{(a_i,b_i)b'})\varphi(p_{(R,\mbf{a},\mbf{b})\mbf{b}'})&=p_{a_i(b_i+b')}p_{a_1(b_1+b_1')}\cdots p_{a_{\ar(R)}(b_{\ar(R)}+b_{\ar(R)'})}\\
        &=p_{a_i(b_i+b')}p_{a_i(b_i+b_i')}p_{a_1(b_1+b_1')}\cdots p_{a_{\ar(R)}(b_{\ar(R)}+b_{\ar(R)'})}=0.
    \end{align*}

    Next, define the map $\psi:\Mor^{o+}(A_H,\mrm{LIN})\rightarrow\Mor^+(G(A,\mrm{LIN})_K)$ on the generators by $\psi(p_{ab})=p_{(a,0)b}$. Again, to see that this extends to a $\ast$-homomorphism, it suffices to see that the relations of $\Mor^{o+}(A_H,\mrm{LIN})$ remain satisfied. In fact, the $\psi(p_{ab})$ are projectors;
    \begin{align*}
        \sum_{b\in\Z_2}\psi(p_{ab})=\sum_{b\in\Z_2}p_{(a,0)b}=1;
    \end{align*}
    if $\mbf{a}\in R^{A_H}$ but $\mbf{b}\notin R^{\mrm{LIN}}$,
    \begin{align*}
        \psi(p_{a_1b_1})\cdots\psi(p_{a_{\ar(R)}b_{\ar(R)}})=p_{(a_1,0)b_1}\cdots p_{(a_{\ar(R)},0)b_{\ar(R)}}=0,
    \end{align*}
    as $\sum_{\mbf{b}\in R^{\mrm{LIN}}}p_{(a_1,0)b_1}\cdots p_{(a_{\ar(R)},0)b_{\ar(R)}}=\sum_{\mbf{b}\in R^{\mrm{LIN}}}p_{(R,\mbf{a},\mbf{0})\mbf{b}}=1$; and if $\mbf{a}\in R^{A_H}$, we have that
    \begin{align*}
        [\psi(p_{a_ib}),\psi(p_{a_jb'})]=[p_{(a_i,0)b},p_{(a_j,0)b'}]=0.
    \end{align*}

    Finally, to finish the proof, note that $\varphi=\psi^{-1}$ since they are inverses on the generators.
\end{proof}

Next, we need some prior results.

\begin{theorem}[\cite{Slo19,Slo20}]\label{thm:william}
    There exists a linear relational structure $A$ such that $\Gamma(A)$ is noncommutative and hyperlinear, and $J\mapsto 1$ in every finite-dimensional representation.

    If there exists a non-hyperlinear group, then there exists a linear relational structure $A$ such that $\Gamma(A)$ is noncommutative and $J\mapsto 1$ in every Connes-embeddable representation.
\end{theorem}

Before passing to the main theorems of this section, we need some technical lemmata. A similar technique has been independently developed in~\cite{DRS26-in-preparation}.

\begin{lemma}\label{lem:josse}
    There exists a homogeneous linear relational structure $A$ with two distinguished elements $a_1,a_2\in A$ such that $\Gamma_H(A)$ is hyperlinear, $[x_{a_1},x_{a_2}]\neq 1$, $[x_{a_1},x_{a_2}]^2=1$, and $\Gamma_H(A)/\angnormal*{[x_{a_1},x_{a_2}]}$ is trivial. 
\end{lemma}

In particular, we have that the subgroup generated by $x_{a_1}$ and $x_{a_2}$ is isomorphic to the dihedral group $D_4$.

\begin{proof}
    Let $A$ be the linear relational structure defined as follows: $\dom(A)$ is the set of order-two elements of the alternating group $A_7$, $\LR{0,3}^A$ is the set of tuples $(a,b,c)$ such that $abc=1$ in $A_7$,\footnote{Note that the two conditions together imply that the subgroup $\langle a,b,c \rangle \subseteq A_7$ is abelian: if $abc = 1$ and $a,b,c$ have order $2$, then $ab = c^{-1} = c = (ab)^{-1} = b^{-1}a^{-1}  = ba $.} and $R^A=\varnothing$ for all the remaining relations. Now, due to \cite[Remark 6.9]{DdBRS25}, $\Gamma_H(A)$ is isomorphic to the unique perfect central extension of $A_7$ by $\Z_3$; thus it is finite and therefore hyperlinear. Also, the map $\phi:\Gamma_H(A)\rightarrow A_7$, $x_a\mapsto a$ is a surjective group homomorphism.
    
    Next, let $a_1=(1\,2)(3\,4)$ and $a_2=(2\,3)(5\,6)$. In $A_7$, $[a_1,a_2]=(1\,4)(2\,3)$ which is a nontrivial order-two element. By construction $x_{a_1}x_{a_2}x_{a_1}$ is an order-two element, whose image under the homomorphism $\phi:\Gamma_H(A)\rightarrow A_7$ is $a_1a_2a_1=(1\,3)(2\,4)$. As such, there exists $z\in\ker\phi$ such that $x_{a_1}x_{a_2}x_{a_1}=x_{a_1a_2a_1}z$. But, squaring both sides gives that $1=(x_{a_1}x_{a_2}x_{a_1})^2=x_{a_1a_2a_1}^2z^2=z^2$, the kernel is central. But since $\ker\phi\cong\Z_3$, we must have $z=1$ and thus $x_{a_1}x_{a_2}x_{a_1}=x_{a_1a_2a_1}$. Since $a=a_2$, $b=a_1a_2a_1$, and $c=a_1a_2a_1a_2$ are all order-two elements of $A_7$ satisfying $abc=1$, we have that $x_ax_bx_c=1$ in $\Gamma_H(A)$. This implies that
    $$x_c=x_bx_a=x_{a_1}x_{a_2}x_{a_1}x_{a_2}=[x_{a_1},x_{a_2}],$$
    which is therefore a non-identity order-two element of $\Gamma_H(A)$.
    
    It remains to show that the normal subgroup generated by $x_c$ is the whole of $\Gamma_H(A)$. Consider the subgroup $H=\angnormal*{x_c}\ker\phi$. Since $A_7\cong\Gamma_H(A)/\ker\phi$ is simple, $H/\ker\phi$ is either trivial or isomorphic to $A_7$. The first case would imply that $x_c\in\ker\phi$, which is impossible as it has order two. As such, the second case must hold, giving that $H=\Gamma_H(A)$. If $\angnormal*{x_c}\cap\ker\phi=\{1\}$, then $\Gamma_H(A)/\angnormal*{x_c}\cong\ker\phi$. Since $\ker\phi$ is abelian and $\Gamma_H(A)$ is perfect, this is a contradiction. This implies that $\ker\phi\subseteq\angnormal*{x_c}$, and hence $\angnormal*{x_c}=H=\Gamma_H(A)$.
\end{proof}

\begin{lemma}\label{lem:homogenise}
    Let $A$ be a linear relational structure. Then, there exists a homogeneous linear relational structure $B$ such that $\Gamma(A)\cong\Gamma_H(B)$.
\end{lemma}

\begin{proof}
    Define $B$ as follows. Let $\dom(B)=\dom(A)\cup\{j\}$, for every $\mbf{a}\in\LR{0,n}^A$ let $\mbf{a}\in\LR{0,n}^B$, for every $\mbf{a}\in\LR{1,n}^A$ let $(a_1,\ldots,a_n,j)\in\LR{0,n+1}^B$, and for every $a\in A$, let $(a,j,a,j)\in \LR{0,4}^B$. Then it is clear that $x_a\mapsto x_a$ and $J\mapsto x_j$ induces an isomorphism $\Gamma(A)\rightarrow\Gamma_H(B)$.
\end{proof}

\begin{lemma}\label{lem:h-triv}
    There exists a homogeneous linear relational structure $A$ with a distinguished elements $a\neq j\in A$ such that $\Gamma_H(A)$ is hyperlinear, $x_a\neq 1$ in $\Gamma_H(A)$, and $\Gamma_H(A)/\angnormal*{x_j}$ is trivial.
\end{lemma}

\begin{proof}
    Let $B$ be linear relational structure corresponding to the Mermin-Peres magic square, \textit{i.e.} $\dom(B)=[9]$, and relations $$\LR{0,3}^B=\set*{(1,2,3),(4,5,6),(7,8,9),(1,4,7),(2,5,8)}$$ and $\LR{1,3}^B=\set*{(3,6,9)}$. We know $\Gamma(B)$ is finite, and therefore hyperlinear, and that $[x_1,x_5]=[x_2,x_4]=J$. Let $B'$ be the homogeneous linear relational structure such that $\Gamma(B)\cong\Gamma_H(B')$ constructed in \cref{lem:homogenise}. So, in $\Gamma_H(B')$, $[x_1,x_5]=[x_2,x_4]=x_j$. Now, let $A_1$ and $A_2$ be two distinct copies of the relational structure from \cref{lem:josse} with distinguished elements $a_{11},a_{12}\in A_1$ and $a_{21},a_{22}\in A_2$. Take $A$ to be the homogeneous linear relational structure with domain $\dom(A)=\dom(B')\sqcup\dom(A_1)\sqcup\dom(A_2)$ and relations $R^{A}=R^{B'}\cup R^{A_1}\cup R^{A_2}$, except for $R=\LR{0,2}$ where $\LR{0,2}^A=\LR{0,2}^{B'}\cup \LR{0,2}^{A_1}\cup \LR{0,2}^{A_2}\cup\{(a_{11},1),(a_{12},5),(a_{21},2),(a_{22},4)\}$. Let the distinguished element $a=1$.
    
    By construction, $\Gamma_H(A)$ is the amalgamated free product of $\Gamma_H(B')$, $\Gamma_H(A_1)$, and $\Gamma_H(A_2)$ over the dihedral groups $D_4\cong\anggroup*{x_1,x_5}\cong\anggroup*{x_{a_{11}},x_{a_{12}}}$ and $D_4\cong\anggroup*{x_2,x_4}\cong\anggroup*{x_{a_{21}},x_{a_{22}}}$. Since we know that $\Gamma_H(B')$, $\Gamma_H(A_1)$, and $\Gamma_H(A_2)$ are hyperlinear, we have that $\Gamma_H(A)$ is hyperlinear. This also implies that $x_a=x_1\neq 1$, and that $[x_{a_{11}},x_{a_{12}}]=[x_1,x_5]=x_j$ and $[x_{a_{21}},x_{a_{22}}]=[x_2,x_4]=x_j$ in $\Gamma_H(A)$. Now, we take the quotient by the relation $x_j=1$. Since $[x_{a_{11}},x_{a_{12}}]=x_j$, we take the quotient of $\Gamma_H(A_1)$ by $[x_{a_{11}},x_{a_{12}}]$. By \cref{lem:josse}, this implies that the image of $\Gamma_H(A_1)$ in $\Gamma_H(A)/\angnormal*{x_j}$ is trivial, and hence $x_{a_{11}}=x_{a_{12}}=1$, which gives $x_1=x_5=1$. Proceeding identically with $\Gamma_H(A_2)$, we also get $x_2=x_4=1$. Now, the relations of $\Gamma(B)$ imply that $x_3=x_1x_2=1$, $x_6=x_4x_5=1$, $x_7=x_1x_4=1$, $x_8=x_2x_5=1$, and $x_9=x_3x_6x_j=1$. Therefore $\Gamma_H(A)/\angnormal*{x_j}$ is trivial.
\end{proof}

\begin{lemma}\label{lem:nontriv-generators}
				Let $A$ be a linear relational structure. Then, there exists a linear relational structure $A'$ such that there is a subset $S\subseteq\dom(A)$ such that $\set*{x_a}{a\in S}$ generates $\Gamma(A')$ and for all $a\in S$, $x_a\neq 1,J$ in $\Gamma(A')$, and there is an embedding $\Gamma(A)\hookrightarrow\Gamma(A')$ such that $J\mapsto J$. Further, if $\Gamma(A)$ is hyperlinear, so is $\Gamma(A')$. 
			\end{lemma}
			
			\begin{proof}
				Let $A'$ be the linear relational structure with domain $\dom(A')=\dom(A)\cup\dom(A)\times[2]$ and relations $R^{A'}=R^A$, except for $R=\LR{0,3}$, where we take $\LR{0,3}^{A'}=\LR{0,3}^A\cup\set*{(a,(a,1),(a,2))}{a\in A}$. This means that $\Gamma(A')$ is constructed from the $\Gamma(A)$ by spliting each $x_a$ into the product $x_{(a,1)}x_{(a,2)}$, where we know $x_{(a,1)}$ and $x_{(a,2)}$ commute. Since $\Gamma(A)\subseteq\Gamma(A')$, this immediately gives a natural embedding. Next, we take $S=\dom(A)\times[2]$. It is clear that $\set*{x_a}{a\in S}$ generates $\Gamma(A')$; and for $x_{a}$, $x_{(a,1)}$ is free except that it commutes with $x_a$, and hence we must have that neither $x_{(a,1)}$ nor $x_{(a,2)}=x_{(a,1)}x_a$ are trivial.
				
				To finish, we need to show that if $\Gamma(A)$ is hyperlinear, so is $\Gamma(A')$. For any $a\in A$, adding $x_{(a,1)}$ and $x_{(a,2)}$ can be seen as taking the amalgamated free product of $\Gamma(A)$ with a power of $\Z_2$ corresponding to $\ang{x_j,x_{(a,1)},x_{(a,2)}}$. If $x_{a}=1,x_j$ in $\Gamma(A)$, then it coresponds to amalgamating with $\Z_2^2$ over $\Z_2\cong\ang{x_j}$; else, it corresponds to amalgamating with $\Z_2^3$ over $\Z_2\cong\ang{x_j,x_a}\cong\ang{x_j,x_{(a,1)}x_{(a,2)}}$. As amalgamated free products over finite abelian groups preserve hyperlinearity, we get the wanted result.
			\end{proof}
			
			\begin{lemma}\label{lem:gamma-triv}
				Let $A$ be a linear relational structure. Then, there exists a homogeneous linear relational structure $B$ with a distinguished element $j\in B$ such that there exists an embedding $\Gamma(A)\hookrightarrow\Gamma_H(B)$ such that $J\mapsto x_j$, and $\Gamma_H(B)/\angnormal*{x_j}$ is trivial. Further, if $\Gamma(A)$ is hyperlinear, so is $\Gamma_H(B)$.
			\end{lemma}
			
			\begin{proof}
				Let $A'$ be the linear relational structure and let $S\subseteq\dom(A')$ be as constructed from $A$ in \cref{lem:nontriv-generators}. Let $A''$ be the homogeneous linear relational structure such that $\Gamma(A')\cong\Gamma_H(A'')$ constructed in \cref{lem:homogenise}. For each $a\in S$, let $A_a$ be a copy of the linear relational structure from \cref{lem:h-triv} with distinguished elements $\alpha_a\neq j_a\in A_a$. Now, define $B$ as the homogeneous linear relational structure with domain $\dom(B)=\dom(A'')\sqcup\bigsqcup_{a\in S}\dom(A_a)$, and relations $R^B=R^{A''}\cup\bigcup_{a\in S}R^{A_a}$ except for $R=\LR{0,2}$ where $\LR{0,2}^B=\LR{0,2}^{A''}\cup\bigcup_{a\in S}\parens*{\LR{0,2}^{A_a}\cup\{(a,\alpha_a),(j,j_a)\}}$.
				
				Now, $\Gamma_H(B)$ is the amalgamated free product of $\Gamma_H(A'')$ and the $\Gamma_H(A_a)$ over the subgroups $\Z_2^2\cong\ang*{x_a,x_j}\cong\ang*{x_{\alpha_a},x_{j_a}}$,since since $x_{\alpha_a}\neq1$ and $x_a\neq 1$ for all $a\in S$ by \cref{lem:h-triv} and \cref{lem:nontriv-generators}, respectively. By construction, if $\Gamma(A)$ is hyperlinear, so is $\Gamma_H(B)$. Next, we can take $\iota$ to be the homomorphism $\Gamma(A)\rightarrow\Gamma_H(B)$ extending $x_a\mapsto x_a$ for all $a\in A$, and $J\mapsto x_j$, which is an injection~\cite{KS70}.
				
				Next, consider the quotient of $\Gamma_H(B)$ by the relation $x_j=1$. Since $x_{j_a}=1$ in the image of $\Gamma_H(A_a)$, we have that that image is trivial by \cref{lem:h-triv}. In particular, $x_{\alpha_a}=x_a=1$, so $\Gamma_H(B)$ is trivial.
			\end{proof}

The main theorem follows naturally from the following two theorems.

\begin{theorem}\label{thm:no-finite}
    There exists a homogeneous linear relational structure $A$ such that $\Gamma_H(A)$ is noncommutative and hyperlinear but $\Gamma_H(A)$ is trivial under any finite-dimensional representation.
\end{theorem}

\begin{proof}
    Let $A_0$ be the linear relational structure such that $\Gamma(A_0)$ is noncommutative and hyperlinear, and $J\mapsto 1$ in every finite-dimensional representation from \cref{thm:william}. Next, let $A$ be the relational structure constructed from $A_0$ in \cref{lem:gamma-triv}. Suppose $\pi$ is a finite-dimensional representation of $\Gamma_H(A)$. Then, $\pi\circ\iota$ is a finite-dimensional representation of $\Gamma(A_0)$. Then, $\pi(x_j)=1$. As such, $\pi$ factors through to a representation of $\Gamma_H(A)/\angnormal*{x_j}$. But we know, this quotient is trivial, so $\pi$ is trivial.
\end{proof}

\begin{theorem}\label{thm:no-connes}
    Suppose there exists a non-hyperlinear group. Then, there exists a homogeneous linear relational structure $A$ such that $\Gamma_H(A)$ is noncommutative but $\Gamma_H(A)$ is trivial under any Connes-embeddable representation.
\end{theorem}

\begin{proof}
    We proceed in the same way as in \cref{thm:no-finite} starting from the relational structure $A_0$ such that $\Gamma(A_0)$ is noncommutative and $J\mapsto 1$ in every Connes-embeddable representation from \cref{thm:william}.
\end{proof}

\begin{proof}[Proof of \cref{thm:separations}]
    Let $A$ be the homogeneous linear relational structure constructed in \cref{thm:no-finite}. Let $G=G(A,\mathrm{LIN})_K$. Then, by \cref{cor:cheating,thm:graph-to-lin} $\Mor^+(G^k,G)\cong\Mor^+(G)^{\oplus k}\cong\Mor^{o+}(A,\mathrm{LIN})\cong C^\ast(\Gamma_H(A)^{ k})$. Now due to \cref{thm:no-finite}, every $\ast$-representation of $\Mor^+(G^k,G)$ into an element of $q$ is trivial, but there is a non-commutative $\ast$-representation into an element of $qa$.

    For the second part, we proceed the same way using the homogeneous linear relational structure constructed in \cref{thm:no-connes}.
\end{proof}

\begin{question}\label{que:one}
    Does there exist a separation between $C^\ast$- and $qc$-commutativity gadgets, or an unconditional separation between $qc$- and $qa$-commutativity gadgets? Do there exist similar separations in the oracular framework?
\end{question}

\begin{figure}

\centering

\begin{tikzpicture}
    \node (c) at (0,0) {$C^\ast$};
    \node (qc) [right=of c] {$qc$};
    \node (qa) [right=of qc] {$qa$};
    \node (q) [right=of qa] {$q$};
    \node (oc) [above=of c] {$oC^\ast$};
    \node (oqc) [above=of qc] {$oqc$};
    \node (oqa) [above=of qa] {$oqa$};
    \node (oq) [above=of q] {$oq$};
    \draw[->] (c) -- (qc);
    \draw[->] (qc) -- (qa);
    \draw[->] (qa) -- (q);
    \draw[->] (oc) -- (oqc);
    \draw[->] (oqc) -- (oqa);
    \draw[->] (oqa) -- (oq);
    \draw[->, dashed] (c) -- (oc);
    \draw[->, dashed] (qc) -- (oqc);
    \draw[->, dashed] (qa) -- (oqa);
    \draw[->, dashed] (q) -- (oq);

    \draw[red, thick] (-0.5,0.75) -- (5.5,0.75) (4.25,0.75) -- (4.25,-0.5);
    \draw[red, thick, densely dashed] (2.5,0.75) -- (2.5,-0.5);
\end{tikzpicture}

\caption{Separations between commutativity gadget classes; robust commutativity gadgets are not shown in this figure, due to the equivalences in~\cref{thm:implications-equivalences}. The separation between oracular and non-oracular commutativity gadgets is via $4$-colouring~\cite{CDdBVZ25}, and the remaining separations are due to~\cref{thm:separations}. Dashed line indicates a separation relying the existence of a non-hyperlinear group.}
\label{fig:separations}
\end{figure}
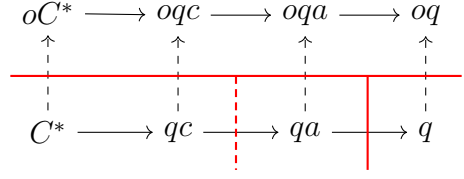

\section{Complexity of deciding the existence of commutativity gadgets}\label{sec:gadget-complexity}

\begin{theorem}\label{thm:gadget-existence-hard}
    \begin{enumerate}[(i)]
        \item The problem of deciding if a given relational structure admits a $q$-commutativity gadget is undecidable.

        \item The problem of deciding if a given relational structure admits a $qa$-commutativity gadget is undecidable.

        \item The problem of deciding if a given relational structure admits a $qc$-commutativity gadget is $\tsf{coRE}$-hard.
        
        \item The problem of deciding if a given relational structure admits a $C^\ast$-commutativity gadget is $\tsf{coRE}$-hard.
    \end{enumerate}
\end{theorem}

To show this, we make use of the hardness results for linear system games due to~\cite{Slo19}, as well as the reductions of the previous section.

\begin{theorem}[\cite{Slo19,Slo20}]\label{thm:hardness-william}
    Given a linear relational structure $A$ as input it is
    \begin{enumerate}[(i)]
        \item Undecidable to decide if there exists a finite-dimensional representation of $\Gamma(A)$ where $J$ is non-trivial.
        \item Undecidable to decide if there exists a Connes-embeddable representation of $\Gamma(A)$ where $J$ is non-trivial.
        \item $\tsf{coRE}$-complete to decide if $J$ is non-trivial in $\Gamma(A)$.
    \end{enumerate}
\end{theorem}

\begin{proof}[Proof of \cref{thm:gadget-existence-hard}]
    For (i), we make use of \cref{thm:hardness-william}(i). Let $A$ be a linear relational structure. Then, in the same way as \cref{thm:no-finite}, there exists a homogeneous linear relational structure $B$ such that $\Gamma(A)$ embeds into $\Gamma_H(B)$ and every finite-dimensional representation of $\Gamma_H(B)$ is trivial iff $J\mapsto 1$ in every finite-dimensional representation of $\Gamma(A)$. Next, as in the proof of \cref{thm:separations}, there exists a relational structure $G=G(B,\mathrm{LIN})_K$ such that $\Mor^+(G^k,G)\cong C^\ast(\Gamma_H(B)^k)$. Therefore, every finite-dimensional representation of $\Mor^+(G^k,G)$ is abelian iff $J\mapsto 1$ in every finite-dimensional representation of $\Gamma(A)$. Due to \cref{thm:basic-char}, $G$ admits a $q$-commutativity gadget iff $J\mapsto 1$ in every finite-dimensional representation of $\Gamma(A)$. Due to \cref{thm:hardness-william}, it is undecidable to decide the latter, so it is undecidable to decide if a relational structure admits a $q$-commutativity gadget.

    (ii) follows in the same way from \cref{thm:hardness-william}(ii); and (iv) follows in the same way from \cref{thm:hardness-william}(iii). Since relations in group algebras hold if and only if they hold for every $qc$-representation, (iii) is equivalent to (iv).
\end{proof}

\begin{question}\label{que:two}
    Do analogous hardness results hold for the oracular commutativity gadget existence problems? Can we upper-bound the hardness of deciding if a relational structure admits a commutativity gadget?
\end{question}

\section{Conclusion}
\label{sec:conclusion}

In this paper, we studied commutativity gadgets for CSPs in all models of quantum correlation.
For every quantum model $Q$, we proved that a CSP admits a $Q$-commutativity gadget if and only if all its quantum polymorphisms in this model are classical, thereby extending the characterization of \cite{CJM25} to all quantum models.
We then proved various implications, equivalences, and separations between those different models, showing some of the subtleties involved in the study of commutativity gadgets.
Finally, we proved that it is undecidable whether or not commutativity gadgets exist in the non-oracular framework, in any of the quantum models.

\heading{Open questions.} Many interesting open questions remain about commutativity gadgets and the complexity of quantum homomorphisms. First, we know some classes of CSP admitting a commutativity gadget, such as $\tsf{NP}$-complete boolean CSPs~\cite{Ji13,CM24}, oracular graph colouring~\cite{Ji13,CDdBVZ25}, odd cycle colouring~\cite{CJM25}; and some classes that do not admit a commutativity gadget, such as non-oracular $k$-colouring for $k\geq 4$, and graph homomorphsim to the diamond graph~\cite{CDdBVZ25}. Can we make use of the polymorphism characterisation of commutativity gadgets to gain a better understanding of which CSPs admit commutativity gadgets? For example, the smallest example of a CSP admitting no commutativity gadget has an alphabet of size $4$: do all ternary CSPs admit a commutativity gadget?

The known reductions between entangled CSPs rely on commutativity gadgets to lift classical reductions. However, there may be more types of reduction possible. Can we find reductions between CSPs that bypass commutativity gadgets? Can we use this to show hardness for CSPs that do not admit a commutativity gadget?

One of the results of this work is to find separations between different classes of commutativity gadget. However, separations between $C^\ast$- and $qc$-commutativity gadgets, or gadgets in the oracular framework remain elusive. Can we find CSPs that exhibit these separations? Or are some of these commutativity gadgets equivalent? This was posed in~\cref{que:one}.

Finally, we show undecidability results for ascertaining the existence of commutativity gadgets in the non-oracular framework. Does this hardness extend to the oracular framework? Also, for the $q$- and $qa$-commutativity gadget existence problems, the undecidability results do not provide an obvious upper bound on the complexity. Can we pinpoint the complexity of these problems more precisely? This was posed in~\cref{que:two}.

\small
\bibliographystyle{bibtex/bst/alphaarxiv.bst}
\bibliography{bibtex/bib/full.bib,bibtex/bib/quantum.bib,bibtex/quantum_new.bib}

\normalsize

\ifthenelse{\boolean{anonymous}}{
\appendix
\section{Background}

}{}

\end{document}